\documentclass[1p]{elsarticle}

\usepackage[utf8]{inputenc} 
\usepackage[T1]{fontenc}
\newcommand\footnoteref[1]{\protect\footnotemark[#1]}
\usepackage{amsmath}
\usepackage{amssymb}
\usepackage{amsfonts}
\usepackage{morefloats}
\usepackage{dcolumn}% Align table columns on decimal point
\usepackage{bm}% bold math
\usepackage{epsfig}
\usepackage{pifont} 
\usepackage{multirow}
\usepackage{subfigure}
\usepackage{float}
\usepackage{epstopdf}
\usepackage{MnSymbol}

\newcommand{\ba}{\begin{eqnarray}}
\newcommand{\ea}{\end{eqnarray}}
\newcommand{\be}{\begin{equation}}
\newcommand{\ee}{\end{equation}}
\newcommand{\bea}{\begin{eqnarray}}
\newcommand{\eea}{\end{eqnarray}}

\newcommand{\Jspace}{\hspace{-0.8mm}}
\newcommand{\doesntwork}{}
\newcommand{\T}{\mathrm{T}}

\newcommand{\Tj}{{}^{\scriptstyle J}\mathrm{T}}

\newcommand{\Ta}{\T^{\{ 1 \}}}
\newcommand{\Tb}{\T^{\{ 2 \}}}
\newcommand{\Ya}{\Y^{\{ 1 \}}}
\newcommand{\Yb}{\Y^{\{ 2 \}}}
\newcommand{\w}{\omega}
\newcommand{\Proj}{\mathcal{P}}

\newcommand{\Projk}{ \Proj^{ \{ k \} }}
\newcommand{\FullProj}{ \Proj^{ \{ 1\ldots N \} }}
\newcommand{\Y}{\mathrm{Y}}
\newcommand{\wigner}{\mathcal{W}}
\newcommand{\prestar}{\smallstar}
\newcommand{\abs}[1]{\ensuremath{\vert #1 \vert}}
\DeclareMathOperator{\tr}{tr}

\def\m {\mathfrak{m}}

\newcommand{\ket}[1]{\ensuremath{| #1 \rangle}{}}
\newcommand{\bra}[1]{\ensuremath{\langle #1 |}{}}
\newcommand{\SU}{\mathrm{SU}}

\usepackage{xcolor}

\newcommand{\JC}{\nu}

\usepackage{bbm}
\newcommand{\unity}{\mathbbmss{1}}
\newcommand{\C}{\mathbb C}

\usepackage{amsthm}

\newtheorem{theorem}{Theorem}
\newtheorem{lemma}[theorem]{Lemma}
\newtheorem{corollary}{Corollary}

\newtheorem{result}{Result}

\allowdisplaybreaks

\usepackage{hyperref}

\makeatletter 
\def\ps@pprintTitle{% 
     \let\@oddhead\@empty
     \let\@evenhead\@empty
     \def\@oddfoot{\footnotesize\itshape\hfill\today}%
     \let\@evenfoot\@oddfoot}
\makeatother
\begin{document}

\renewcommand*{\today}{July 30, 2018}

\begin{frontmatter}

\title{Time evolution of coupled spin systems\\ in a generalized Wigner representation}

%% Group authors per affiliation:
\author{Bálint Koczor}
\ead{balint.koczor@tum.de}
\author{Robert Zeier}
\ead{zeier@tum.de}
\author{Steffen J. Glaser}
\ead{glaser@tum.de}
\address{Technische Universität München, Department Chemie,\\ Lichtenbergstrasse 4, 85747 Garching, Germany}

\begin{abstract}
Phase-space representations as given by Wigner functions are a powerful tool for 
representing the quantum state and characterizing its time evolution
in the case of infinite-dimensional quantum systems and have been widely 
used in quantum optics and beyond. Continuous phase spaces have also been studied 
for finite-dimensional quantum systems such as spin systems.
However, much less is known for finite-dimensional, \emph{coupled} systems,
and we present a complete theory of Wigner functions for this case.
In particular, we
provide a self-contained Wigner formalism for describing and predicting the time evolution
of coupled spins which lends itself to visualizing the high-dimensional structure of
multi-partite quantum states. We completely treat the case of an arbitrary number
of coupled spins $1/2$, thereby establishing the equation of motion using Wigner functions.
The explicit form of the time evolution is then calculated for 
up to three spins $1/2$. The underlying physical principles of our Wigner representations 
for coupled spin systems
are illustrated
with multiple examples 
which are easily translatable to other experimental scenarios.
\end{abstract}


\end{frontmatter}

%%%%%%%%%%%%%%%%%%%%%%%%%%%%%%BALINT%%%%%%%%%%%%%%%%%%%%%%%%%%%%%%%%%%%%%%%%%%%%%

\section{Introduction}

Prior to the emergence of quantum mechanics, 
geometric intuition played a particularly strong role in the formulation of classical physics.
Breaking with this tradition, quantum mechanics is often
formulated abstractly  by Hilbert-space operators
such as the density operator describing the quantum state 
or the Hamiltonian corresponding to the total energy of the system.
The demand for an intuitive formulation of quantum mechanics has driven the development of
the so-called Wigner-Weyl formalism \cite{Wey27,Weyl31,Weyl50,wigner1932}
which is equivalent to other formulations of quantum mechanics, but its phase-space
approach mirrors the classical phase space. Moreover, the time evolution of these
quantum systems can be entirely characterized on the level of 
Wigner functions
\cite{Gro46,Moy49} and in a similar fashion 
as the evolution of a statistical ensemble of classical particles.

Similar as quantum systems with an infinite-dimensional Hilbert space
as studied in quantum optics, the Wigner formalism 
for describing the time evolution 
on a continuous phase space has been extended to
finite-dimensional quantum systems such as spins (see Sec.~\ref{prior}).
However, much is less is known for finite-dimensional, \emph{coupled} systems.
We present 
a complete theory of Wigner functions 
applicable to arbitrary density matrices and operators
of coupled spin systems.
In particular, we specify Wigner functions for operators of
arbitrary coupled spin systems,
i.e., systems that consist of an arbitrary number of coupled spins $J$.
Furthermore, we address the following questions for coupled quantum systems:
How does the Wigner function of a quantum state evolve in this case? 
Given the Wigner function of a Hamiltonian, how can 
one predict the Wigner function of a 
quantum state at a later time without relying on explicit matrices?

In finite-dimensional 
quantum systems, so far only the Wigner formalism for systems consisting of uncoupled spins has been fully developed and only special results for systems consisting of two coupled spins have been reported in the literature.
Here we solve the open question of how to compute the time evolution
of arbitrary coupled spin systems using a consistent Wigner formalism.
Our characterization of the time evolution relies on explicit partial derivatives of Wigner functions.
Moreover, our Wigner representation is also suited
for graphically visualizing the high-dimensional structure of
multi-partite quantum states or operators in a compact and instructive form.
This allows for geometric reasoning beyond matrix mechanics
and provides a novel didactic approach complementary 
to matrix treatments of the time evolution.

As our results might be of interest to a wider audience, we present our work on 
several levels.
Most importantly, the underlying physical principles are 
\emph{first}
highlighted through a set of illustrative examples 
for coupled spin systems, which
are easily translatable to 
experimental approaches for realizing qubits in trapped ions, quantum dots, 
or superconducting circuits. This demonstrates
that our novel approach for calculating the time evolution
nicely conforms with conventional Hilbert-space quantum mechanics.
Building on the intuition from the examples, we \emph{then}
develop and discuss
the mathematical
formulation 
of our Wigner representation coupled quantum systems
and its time evolution
in sufficient detail
for facilitating theoretical extensions in the future. These theoretical
advances on computing the time evolution for coupled quantum systems
in a consistent Wigner frame
constitute our central results.

We continue this introduction by first reviewing
basic properties of Wigner functions
of infinite-dimensional quantum systems
which will motivate and guide our approach.
Then, we summarize
results from the literature for both 
Wigner functions and visualization techniques of finite-dimensional quantum systems.
Finally, before starting the main text, we provide a summary of our results, motivate them further, and outline the structure of this work.

\subsection{Wigner functions of infinite-dimensional quantum systems\label{introtowignerf}}
Even though we almost exclusively focus on Wigner functions 
of finite-dimensional quantum systems such as spin systems, we will shortly review 
how the time evolution is established for Wigner functions of infinite-dimensional
quantum systems. This will also set the stage for 
related techniques in the finite-dimensional case. In general,
quantum mechanics describes how the quantum state evolves 
under the action of a Hamiltonian and there are at least
three independent approaches to this description:
the Hilbert-space formalism relying on matrices and operators \cite{cohen1991quantum}, 
the path-integral method \cite{feynman2005}, and the 
Wigner phase-space approach \cite{carruthers1983,hillery1997,kim1991,lee1995,
	gadella1995,zachos2005,schroeck2013,SchleichBook,Curtright-review}.

We consider here the latter approach which particularly
eases the comparison with classical mechanics.
Groenewold \cite{Gro46} and Moyal \cite{Moy49} formalized quantum mechanics
as a statistical theory on a classical phase space by associating 
the density operator in the Hilbert space with a function on the phase
space and interpreting this correspondence 
as the inverse of the Weyl transformation
\cite{Wey27,Weyl31,Weyl50}.
In particular, the density operator $\rho$ can be represented by 
a Wigner function \cite{wigner1932} $W_\rho(x,p)$ which
constitutes a quasiprobability function 
in classical phase-space coordinates $x$ and $p$.
A general framework for this theory was given by Bayen \emph{et al.}
\cite{1bayen1978,2bayen1978}.

More precisely, the Wigner formalism represents the density operator $\rho$
of an infinite-dimensional quantum system  
as the Fourier transformation (cf.\ p.~68 in \cite{SchleichBook})
\begin{equation*}
W_\rho(x,p)= \wigner(\rho) = 1/(2 \pi h)
\int_{-\infty}^{\infty}   \tilde{\rho}(x,\xi) \exp(-i p \xi/h) 
\,\mathrm{d} \xi,
\end{equation*}
where $\tilde{\rho}(x,\xi)$ is given by $\langle x+\xi/2 | \rho | x-\xi/2 \rangle$,
and $\wigner(\rho)$ denotes the Wigner transformation of $\rho$.
The Wigner function $W_\rho(x,p)$ is real, normalized
(i.e., $\int W_\rho(x,p)  \,\mathrm{d}x \,\mathrm{d}p= \tr(\rho) =1$), and bounded
(i.e., $-2/h \leq W_\rho(x,p) \leq 2/h$). Integrating $W_\rho(x,p)$ over the variable $p$
results in  the quantum-mechanical probability densities $P(x)$  in the coordinate $x$,
and \emph{vice versa} if $x$ and $p$ are exchanged. More generally, an arbitrary operator $A$
is associated with its Wigner function $W_A(x,p)$, and 
the quantum-mechanical expectation
value $\langle A \rangle=\int W_\rho(x,p) W_A(x,p)  \,\mathrm{d}x \,\mathrm{d}p$
is then computed as a classical, statistical average
over phase-space distributions.

In the Hilbert-space formalism, 
a quantum state  is described by the density operator
$\rho$ and its time evolution 
is governed by the von-Neumann equation 
%-------
\begin{equation}\label{NeumannEq}
i \frac{\partial \rho }{\partial t} = [ \mathcal{H} , \rho ]:=\mathcal{H} \rho - \rho \mathcal{H}
\end{equation}
%-------
where $\mathcal{H}$ denotes the Hamiltonian of the quantum system.
The time evolution of a Wigner function can be directly calculated in the phase-space 
representation by introducing a so-called 
star product \cite{Gro46,1bayen1978,2bayen1978}, which mimics
the Wigner function $\wigner(AB)=W_A(x,p) \star W_B(x,p)$
of the product $AB$ of Hilbert space operators. The appropriate form 
$\star=\exp( i \hbar \{\cdot,\cdot\} /2)$
of the star product was given by Groenewold \cite{Gro46} where
$\{\cdot,\cdot\}=\overleftarrow{\partial}_x \overrightarrow{\partial}_p - \overleftarrow{\partial}_p \overrightarrow{\partial}_x$
is the Poisson bracket known from classical physics (cf.\ Vol.~1, §42 in \cite{landau1976}). 
As an important consequence reflecting a classical equation of motion, the 
time evolution of a Wigner function is given by (see Eq.~(10) in \cite{zachos2005})
\begin{equation} \label{Moyal_infinite}
i \hbar \frac{\partial W_\rho }{\partial t} =  [ W_\mathcal{H} , W_\rho  ]_\star  
:= W_\mathcal{H} \star W_\rho -  W_\mathcal{H} \star W_\rho .
\end{equation}
and 
can be determined as an expansion series in $\hbar$, whose first term
is given by the Poisson bracket $\{W_\mathcal{H},W_\rho\}$.
The Wigner formalism and its star product for infinite-dimensional quantum systems 
are well established 
and widely used in quantum optics and beyond. 
Along with what is known
for the star product of
finite-dimensional quantum systems (see Sec.~\ref{prior}),
it is our aim to develop an analogous theory
leading to a version of 
a differential star product for \emph{coupled} spin systems
which is
comparable to the one in
Eq.~\eqref{Moyal_infinite}.

\subsection{Prior work on Wigner functions of finite-dimensional quantum systems\label{prior}}

Fundamental postulates for 
phase-space representations of finite-dimensional quantum systems 
were proposed by
Stratonovich \cite{stratonovich}
(see \ref{stratpostulatessingle}),
and these build the foundations for 
Wigner functions of finite-dimensional quantum systems.
Reflecting the translational covariance
of Wigner functions of infinite-dimensional quantum systems, 
one of these postulates states that the Wigner function
has to transform naturally under rotations. 
The rotational covariance constrains continuous Wigner representations of spins
into a spherical coordinate systems. The resulting Wigner functions
can then be given by linear combinations of spherical harmonics,
which offers a convenient tool for visualizing spins (see Sec.~\ref{intro_visual}).

The Wigner transformation of single-spin operators 
was developed by V{\'a}rilly and Gracia-Bond{\'i}a \cite{VGB89}
and was then further extended by Brif and Mann \cite{Brif97,Brif98}.
In particular,
\cite{VGB89} provides an explicit formula for
the Wigner transformation for a single spin $J$, which
satisfies the Stratonovich postulates.
This formula uses a rank-$j$-dependent kernel which is based on  
mapping transition operators $\ket{J m}\bra{J m'}$ onto 
their corresponding Wigner functions $W_{\ket{J m}\bra{J m'}}$; the connection
between tensor operators $\Tj_{jm}$ and spherical harmonics $\Y_{jm}$
was also mentioned.
A more general kernel for the continuous phase-space representation of a single spin was stated
in \cite{Brif97,Brif98}. It defines Wigner functions 
of tensor operators $\Tj_{jm}$ of single spins as the corresponding
spherical harmonics
$\Y_{jm}$.
We build on these results in Sec.~\ref{AppendixWignerRepr} while also unifying 
normalization factors.

\begin{table}[tb]
	\centering
	\caption{Results from the literature for the Wigner formalism 
		of $N$ spins with spin number $J$. 
		Explicit references are given for
		the Wigner transformation $\wigner(A)$, the star product $f \star g$, and the equation
		of motion $\partial W_\rho / \partial t$.
		Cases not considered in the literature are left blank.\label{ResultInLiterature}}
\begin{minipage}{\textwidth}
	\begin{center}
	\begin{tabular}{@{\hspace{23mm}} l @{\hspace{16mm}} c @{\hspace{16mm}} c @{\hspace{23mm}} } 
		\\[-2mm]
		\hline\hline
		\\[-3mm]
		$N$   & Descr. & Arb. $J$ (incl.\ $J=1/2$)
		\\[1mm]  \hline \\[-2mm]
		$N=1$
		& $\wigner(A)$ &   Eq.~(2.14) in \cite{VGB89}
		\\  &&   + Eq.~(3.29) in \cite{Brif98} 
		\\[1mm]
		&  $f \star g$	 &  \cite{StarProd}   \\[1mm]  
		&   $\partial W_\rho / \partial t$ & 
		\footnote{Equation~(61) in \cite{StarProd} states the equation of motion for the limit of $J \rightarrow \infty$
			using only the Poisson bracket.}
		\footnote{Equation~(5.13) in \cite{VGB89} provides the equation of motion 
			of a spin $J$ for linear Hamiltonians using only the Poisson bracket.}
		\\[4mm] 
		$N=2$
		& $\wigner(A)$ & Eq.~(30) in \cite{klimov2005classical} 
		\\[1mm]
		&$f \star g$	     \\[1mm]  
		& $\partial W_\rho / \partial t$  
		& \footnote{The semiclassical equation of motion (for $J\gg1$) is computed for 
			a particular Hamiltonian ($\chi I_{1z}I_{2z}$) in Eq.~(34) of \cite{klimov2005classical}
			and conforms with our results shown in Sec.~\ref{twospinwignereval}.}
		\\[4mm]
		Arb. $N$ 
		& $\wigner(A)$ & 
		\footnote{
			Phase-space representations are given in terms of the so-called displaced parity 
			operator for a single spin $J$ (see Eq.~(8) in \cite{tilma2016})
			and for $N$ coupled spins $1/2$ (see Eq.~(9) in  \cite{tilma2016}).
			However, our Result~\ref{result1} for coupled spins-$J$ differs from the
			approach of \cite{tilma2016}: we view their Wigner function
			as a linearly shifted Q function $a Q_\rho - b$
			(for $J>1/2$), which also relaxes the Stratonovich postulates (iiia) and (iiib)
			from \ref{stratpostulates}.
		}
		\\[1mm]  \hline \hline  \\[-2mm]
	\end{tabular}
	\end{center}
\end{minipage}
\end{table}

Parallel to our work, a general approach for phase-space representations
was proposed in \cite{tilma2016} which is based on the
so-called displaced parity operator \cite{bishop1994}. The explicit form of the transformation
kernel is computed for the special cases of a single spin $J$ (see Eq.~(8) in \cite{tilma2016})  
and for $N$ coupled spins $1/2$ (see Eq.~(9) in  \cite{tilma2016}).
This also mostly conforms with our results on spin Wigner functions
and fulfills the covariance property under local operations.
However, our results differ from the
approach of \cite{tilma2016} since we view their Wigner function
as a linearly shifted Q function $a Q_\rho - b$
(for $J>1/2$), which also relaxes the Stratonovich postulates (iiia) and (iiib)
from \ref{stratpostulates}.

Complementing the star-product formalism in the infinite-dimensional case,
V{\'a}rilly and Gracia-Bond{\'i}a \cite{VGB89} discuss both the integral and differential form of a (twisted) star
product of Wigner representations in finite dimensions.
Carrying out explicit calculations for particular Hamiltonians (containing only $I_x$, $I_y$, $I_z$ \cite{EBW87}) 
they conclude that in this case a stronger
version of the Ehrenfest theorem holds for the equation of motion.
Namely the time evolution is exactly given by the Poisson bracket known from classical physics.

Klimov and Espinoza \cite{StarProd} determined an exact form of the differential star product for an arbitrary spin number
$J$. This star product is a sum of a pointwise product of two functions and combinations of derivatives with 
respect to spherical coordinates.
The method also requires a rank-$j$-dependent correction in the spherical-harmonics  decomposition
which defines the Wigner function, as well as the truncation of the maximal rank $j$. 
For the  restriction of their expression to a spin number of $1/2$, the calculation of the
star product is more complicated then in the current paper (details for generalizing 
our approach to an arbitrary spin number $J$
will be discussed elsewhere),
however, the derived equation of motion results in the Poisson bracket \cite{klimov2002ExactEvolution},
just as in \cite{VGB89} and in the current paper.
Similarly, the results of \cite{Gratus97} contain spherical functions in a particular limit 
in which the star commutator is given by the Poisson bracket.
For the case of a global $\SU(2)$-dynamical symmetry, 
a Wigner representation
and its corresponding star product 
was developed in
\cite{klimov2008generalized} along the lines of  \cite{StarProd},
leading to a representation which is not unique in the general case of coupled spins
(without global $\SU(2)$ symmetry).
See Table~\ref{ResultInLiterature} for a summary of 
results known in the literature.

\subsection{Discrete Wigner functions}
Several approaches \cite{wootters1987,leonhardt1996,gibbons2004,ferrie2009} 
exist to construct a discrete analog of  Wigner functions
for finite-dimensional quantum systems (see Table~1 in \cite{ferrie2009}).
For example, Wootters \cite{wootters1987} proposed
a discrete Wigner function by introducing
a discrete phase space on a discrete square grid of 
$p \times p$ points for
each Hilbert space of prime dimension $p$.
For composite systems such as coupled spins,  the 
Hilbert space is composed
of subsystems of prime dimension and the corresponding
discrete phase space contains
a Cartesian product of discrete square grids of prime dimension.
The Wigner function is finally defined over this grid and forms a flat,
but discretized analogue of the continuous classical phase space. 
The negativity of discrete and general Wigner functions will be discussed
in the conclusion (see Sec.~\ref{conclusion}).

As discrete Wigner functions are not discussed in the main text,
we shortly contrast them with our approach of finite-dimensional 
(continuous)
Wigner
functions for coupled quantum systems. 
Building on the work of Stratonovich \cite{stratonovich}, we 
obtain a spherical phase space which features 
continuous spherical functions as Wigner functions.
In contrast,
discrete Wigner functions on a square grid exhibit a considerable different 
geometry. In particular, the continuous degrees of freedom of our finite-dimensional Wigner representation
allow us to describe the time evolution in terms of partial derivatives of Wigner functions leading
to a differential form of the star product as an analog to the infinite-dimensional case
in Eq.~\eqref{Moyal_infinite}. This differential picture is not entirely natural for discrete
Wigner functions, and therefore integral forms of the time evolution are usually considered in the discrete case
(cf.\ \cite{marchiolli2013}).

\subsection{Visualization techniques for spins\label{intro_visual}}

There are numerous approaches for visualizing finite-dimensional quantum systems.
Feynman \emph{et al.}~\cite{Feynman_Vernon_57} represents 
operators in a two-level quantum system using three-dimensional (real)
vectors which can be interpreted as a Bloch vector,  field vector, or 
rotation vector. This representation is widely used in many fields, including
magnetic
resonance imaging \cite{handbook_04,EBW87} and spectroscopy \cite{EBW87} or quantum
optics \cite{SchleichBook}.

As in the present work, spin operators (as tensor operators \cite{Racah42})
have often been represented and visualized by spherical harmonics \cite{Jac99}.
In early work by Pines \emph{et al.}~\cite{Pines}, selected density operator
terms of a spin-$1$ particle are illustrated using spherical harmonics, see also 
\cite{HalsMNMRIX,SH91}. Quantum states of a collection of indistinguishable two-level atoms
are depicted by Wigner functions in Dowling \emph{et al.}~\cite{DowlingAgarwalSchleich}.
We also refer to similar illustrations in \cite{JHKS}.
Single-spin systems are visualized in \cite{PhilpKuchel} using spherical harmonics
while stressing applications in nuclear-magnetic resonance. The appendix of
\cite{PhilpKuchel} also discusses whether their method could 
be extended to coupled spins. Certain restricted cases 
of two coupled spins were considered in \cite{MJD}.
Harland \emph{et al.}~\cite{Harland} introduces a method for visualizing particular states
in two- and three-spin systems.

A general method for representing and visualizing 
arbitrary operators of coupled spin 
systems was proposed in \cite{DROPS}: 
This so-called DROPS representation establishes a bijective mapping between
operators and spherical functions by mapping 
tensor operators to spherical harmonics. Important features
as symmetries under simultaneous rotations or certain permutations and  
the set of involved spins are preserved and highlighted in its pictorial presentation.
We discuss relations 
to our visualization method in \ref{dropsappendix}.

The theoretical methods used in \cite{tilma2016} (as discussed in
Sec.~\ref{prior}) also yield a visualization technique for
finite-dimensional coupled quantum systems, which 
is covariant under rotations as in this work. For single spins,
their approach (see Fig.~1(a)-(c) in \cite{tilma2016})
is comparable to \cite{DROPS} and this work. However
for coupled spins, \cite{tilma2016} depicts only 
slices of their high-dimensional Wigner functions
(see Figure~1(d)-(f) and Figure~2 in \cite{tilma2016}).
In our representation, high-dimensional
Wigner functions are visualized by decomposing
them into sums of  product operators.

\subsection{Summary of results\label{summary_res}}

We will now summarize and discuss our results
for finite-dimensional coupled quantum systems, while
emphasizing 
the mathematical and theoretical advances contained in Sec.~\ref{theorysection}.
The current work systematically develops a generalized Wigner formalism
for finite-dimensional coupled quantum systems.
Most importantly, we solve the open question of how to compute the time evolution 
of coupled quantum systems
using a consistent Wigner frame.

It is our goal to describe
the time evolution of these coupled systems only using Wigner functions and not relying
on operators or matrices. Wigner functions of coupled quantum systems
can be uniquely characterized using multiple spherical coordinates
$\theta_k$ and $\phi_k$. This leads to intricate, high-dimensional functions
which are difficult to visualize. We resolve this difficulty and present an approach 
that decomposes a high-dimensional Wigner function
into a linear combination of products of spherical harmonics, which can be conveniently 
visualized while highlighting crucial properties of coupled quantum systems.
We denote our approach by the abbreviation PROPS which is assembled
from the letters of ``product operators.'' We emphasize that a given high-dimensional
Wigner function has usually multiple different but equivalent PROPS representations.

Even though the visualization in the PROPS representation might require
in general exponentially many terms 
as the number $N$ of spins grows, the Wigner function is still uniquely characterized
by a single $2N$-variate function (see Result~\ref{result1}). 
This necessary exponential scaling might limit plotting to a moderate
number of spins. However, visualizing
the dynamics of even a moderate number of spins
is useful for many active areas of research and education,
such as quantum information
processing \cite{NC00} or magnetic resonance \cite{EBW87}.
This visualization technique will allow us, for example, to explore the underlying mechanisms
of efficient experimental control schemes \cite{Roadmap2015}.
We want to also stress that the potential plotting inefficiencies do not 
affect our main theory as presented in Sec.~\ref{theorysection} as it 
directly operates on the defining $2N$-variate function (see Result~\ref{result1}).
For a single spin-$1/2$ system, our 
Wigner functions are similar to
the Bloch vector (cf.\ Sec.~\ref{intro_visual}).
But even in this simple case, our Wigner approach is more expressive
and
allows for a natural representation and visualization of
non-hermitian operators 
as given by coherence order operators (see Fig.~\ref{NonHermitian} below)
which cannot be represented using a single Bloch vector. 
For coupled spins, our Wigner representation
can be compared to a collection of Bloch vectors
for the special cases shown in Figs.~\ref{Figure4} and \ref{Figure5} below.
However, our Wigner representation goes well beyond a simple collection
of Bloch vectors as the number of elements in the linear decomposition
for the PROPS representation is in general not constant
(see Figs.~\ref{Figure2} and \ref{Figure8} below).

Our characterization of the time evolution leads 
to a self-contained theory  of finite-dimensional
quantum systems, while we focus in this work 
mainly on coupled spins $1/2$.
The determination of the correct star product for coupled spins $1/2$
constitutes the cornerstone of our approach for providing
a replacement of the von-Neumann equation applicable to Wigner functions. 
The explicit equation of motion is calculated for an arbitrary number of coupled 
spins $1/2$ in Result~\ref{result4} and discussed separately for the
particular case of natural coupling Hamiltonians in Corollary~\ref{result5}.
Refer to Table~\ref{SummaryofResultsTable} for an overview of 
the results presented in the
current work. Further details for a single spin~$J$ and coupled spins
are respectively summarized 
in following Sections \ref{coreres1} and \ref{coreres2}.

\begin{table}[tb]
	\centering
	\caption{Results presented in this work for the Wigner 
		formalism of $N$ coupled spins with spin number $J$.
		Explicit references 
		are given for 
		the Wigner transformation $\wigner(A)$, the star product $f \star g$, and the equation
		of motion $\partial W_\rho / \partial t$.
		Cases not considered here are left blank.}
	\label{SummaryofResultsTable}
\begin{minipage}{\textwidth}
\begin{center}
	\begin{tabular}{@{\hspace{15mm}} l @{\hspace{16mm}} c @{\hspace{16mm}} l @{\hspace{15mm}} l @{\hspace{16mm}}} 
		\\[-2mm]
		\hline\hline
		\\[-3mm]
		$N$   & Descr. & $J=1/2$ & Arb. $J$ 
		\\[1mm]  \hline \\[-2mm]
		$N=1$
		&  $\wigner(A)$ 
		& &Eq.~\eqref{WignerTransform}    \\[1mm]
		&  $f \star g$	
		& Result~\ref{result2} 
		& \doesntwork  \doesntwork  \doesntwork  \\  [1mm]
		& $\partial W_\rho / \partial t$	  
		& Eq.~\eqref{NeumannWignerEvolution} &
		\footnote{ The case of linear Hamiltonians is considered in Eq.~\eqref{evolutionArbitraryJ}.}
		\\[3mm] 
		$N=2$
		& $\wigner(A)$
		& & Sec.~\ref{TwoSpinSummary}
		    \\[1mm]
		&  $f \star g$	  
		& Corollary~\ref{twospinstarprodcorollary} 
		& \doesntwork  \doesntwork  \doesntwork  \\[1mm]  
		& $\partial W_\rho / \partial t$	  
		&  Corollary~\ref{twospinEQMcorollary} &     \\[3mm]  
		$N=3$
		&  $\wigner(A)$ 
		& & Sec.~\ref{threespinsummary}        \\[1mm]
		&  $f \star g$	          
		& Corollary~\ref{threespinstarprodcorollary}
		& \doesntwork  \doesntwork  \doesntwork  \\[1mm]  
		&  $\partial W_\rho / \partial t$				
		&  Corollary~\ref{threespinEQMcorollary}\footnote{\label{foo} A simplified form for natural 
		Hamiltonians with linear and bilinear terms
		is given in Corollary~\ref{result5}.} 
		&   \\[3mm] 
		Arb. $N$
		& $\wigner(A)$ 
		& & 	Result~\ref{result1}
		\\[1mm]
		&  $f \star g$			
		&  Result~\ref{result3}
		& \doesntwork  \doesntwork  \doesntwork \\[1mm]
		&  $\partial W_\rho / \partial t$ & 
		Result~\ref{result4}\textsuperscript{\ref{foo}}
		\doesntwork &    \doesntwork \doesntwork 
		\\[1mm]  \hline \hline
	\end{tabular}
	\end{center}
\end{minipage}
\end{table}

\subsubsection{Results for a single spin \texorpdfstring{$J$}{J} \label{coreres1}}
A single-spin operator $A$ is mapped to a Wigner function $W_A(\theta, \phi)$
by decomposing $A$ into a linear combination of tensor
operators which can be directly mapped to spherical harmonics
(see Sec.~\ref{AppendixWignerRepr}), and  the corresponding Wigner transformation is 
stated in Eq.~\eqref{WignerTransform}. For specifying
the time evolution of Wigner functions, one needs to introduce the star product
$W_A \star W_B$
of Wigner functions $W_A$ and  $W_B$
which is defined by its relation to the Wigner function 
$W_{AB}=W_A \star W_B$
of a product of operators $A$ and $B$
(see Sec.~\ref{subsecStarProdDef}). 
There are two approaches to compute the star product.
The first approach is known as the integral star product
and relies on an integral transformation of the functions $W_A$ and $W_B$
using a so-called trikernel \cite{VGB89}, and we detail this form in \ref{IntegralStarProd}
for a single spin $J$.

The second approach features a differential form which
is more convenient for computations.
This differential star product relies on the partial derivatives of the corresponding Wigner
functions $W_A$ and $W_B$, which is comparable to
the infinite-dimensional case discussed
in Sec.~\ref{introtowignerf}.
We calculate this new differential form of
the exact star product for a single spin $1/2$ in Result~\ref{result2}: it is a sum 
of a pointwise product $W_A W_B$  and the Poisson bracket $\{ W_A,W_B \}$
followed by a projection.
This form was not reported previously in the literature,
and provides a simplified approach as compared to the results of \cite{StarProd}
while its structure is similar to the structure of the infinite-dimensional star product.
We also derive
an algebraic expansion formula for  the star product of spherical harmonics
in general (see Sec.~\ref{deepformulas}), paving the way for a generalization to an arbitrary 
spin number $J$.
The explicit form of the star product determines the equation of motion
for a single spin $1/2$ and we obtain a particularly simple form
given by the Poisson bracket [see Eq.~\eqref{NeumannWignerEvolution}].

We also point out connections to similar characterizations.
The Poisson bracket 
is directly related to the canonical angular momentum operator
${\mathcal{L}}={r} \times {p}$ which
generates infinitesimal rotations of the sphere
and is known from infinite-dimensional quantum mechanics
(see  Sec.~\ref{connangmom}). 
Even though spins have no classical counterparts,
a classical description emerges from the quantum one
in the limit of $J\rightarrow \infty
$ as detailed in Sec.~\ref{InfDimConnection}.
This leads to a localized distribution
and arbitrary large values in the Wigner function.
Relations to 
Wigner functions of infinite-dimensional quantum systems
are investigated in Sec.~\ref{InfDimConnection}:
the flat phase-space coordinates $(p,q)$ known from classical mechanics
are replaced with spherical phase-space coordinates  $(R\cos\theta,\phi)$
for spins (see Section~\ref{connectionPScordinates}),
and the resulting equation of motion given in 
Sec.~\ref{connectionComutators} is formally equivalent to
the one obtained in the main text [see Eq.~\eqref{NeumannWignerEvolution}].
The star product for a single spin $1/2$ given in the main text can also be interpreted as
a quaternionic product (see Sec.~\ref{Quaternions}).

\subsubsection{Results for coupled spins \label{coreres2}}
The Wigner representation is generalized in Result~\ref{result1}
to an arbitrary number of coupled spins $J$ by extending the formula
for the Wigner transformation of product operators.
We consider Wigner functions for coupled spins of identical spin number $J$, but a generalization
to systems that are composed of particles of different spin number $J$ is straightforward. 
The Wigner functions for $N$ coupled spins 
are defined on a spherical phase space of $N$ spheres and 
have $2N$ coordinates of the form
$R\cos\theta_k$ and $\phi_k$. This setup satisfies the Stratonovich postulates 
of \ref{stratpostulatesmulti},
which includes the
covariance under arbitrary local rotations, and generalizes the covariance under simultaneous 
spin rotations in \cite{DROPS}.

The Wigner formalism for coupled spins $1/2$ is obtained
by extending the differential star product to multiple spins
in Result~\ref{result3}, and this 
also establishes the equation of motion,
which we computed in Result~\ref{result4}
for an arbitrary number of coupled spins $1/2$.
This allows us to describe the quantum properties
as corrections to the classical case which are given 
in a finite power-series expansion. Truncations 
to this expansion could be used to characterize 
a semi-classical approximation.
The equation of motion is then applied
in Sec.~\ref{summaryofresults}
in order to derive its simplified form
for up to three spins $1/2$. In Corollary~\ref{result5},
the simplified equation of motion for an arbitrary number of spins is explicitly given 
for the case of natural Hamiltonians
containing only linear and bilinear operators.

\subsection{Motivation}
Let us now further motivate our approach by highlighting 
its benefits as well as
crucial differences to other work
in the context of Wigner functions (and phase-space representations).
This discussion aims to clarify the choices made in this work.

The previous parts of this introduction have already emphasized 
our focus on coupled spin systems. In principle, their Wigner functions
could be defined by interpreting the coupled spin system as a single spin 
with a large enough spin number and applying the Wigner function techniques 
for single spins as in \cite{DowlingAgarwalSchleich}
(similarly as discussed in Sec.~\ref{prior} and \ref{coreres1}).
This would, however, neglect important features of coupled spin systems we want to \emph{stress}
in our approach such as symmetries under spin-local rotations or permutations 
of spins as well as spin-local properties of the quantum system (cf.\ p.~3 in \cite{DROPS}).
This distinguishability of spins is of crucial importance in, e.g., quantum information processing \cite{NC00}
where local control is often assumed.
These locality features are a focal point of our work and they critically depend on
describing the system in a suitably chosen basis which highlights the underlying
tensor-product structure.
In this regard, Result~\ref{result4} (see Sec.~\ref{app_coupled}) provides a 
novel perspective of expanding the time evolution into its parts according to their degree of nonlocality. 
Therefore, our results for the time evolution of Wigner functions
go well beyond the established theory 
for Wigner functions of single spins (see Sec.~\ref{prior}) and enable contributions 
into a significant new direction.
And this aim to highlight nonlocality properties is quite natural as one can infer, for example, from 
work on matrix product states in many-body physics (see, e.g., \cite{Scholl11,Orus14})
or, more generally, entanglement in quantum information (see, e.g., \cite{NC00,HHHH09,GT09,ES14}).

We want to also emphasize that---in the context of Wigner functions---this focus on coupled
spin systems (and their locality features) emerged only recently \cite{PhilpKuchel,MJD,Harland,DROPS,tilma2016}
(refer to Sec.~\ref{intro_visual}). The Wigner function for coupled spin systems
is defined as a unique $2N$-variate function (see Result~\ref{result1}).
Due to its high dimensionality, the Wigner function cannot be directly plotted
in three dimensions, and one would have to resort to plotting slices as discussed at the end of Sec.~\ref{intro_visual}.
This has motivated us to depict a Wigner function of $N$ coupled spins
using three-dimensional figures (without loss of information)
which are denoted as the PROPS representation (see Sec.~\ref{summary_res} and \ref{Examplessection})
and show decompositions into tensor-product operators.
For a moderate number of coupled spins,
our results can therefore be used as an analytic tool for characterizing the time evolution
in application areas such as quantum information processing \cite{NC00}, magnetic resonance \cite{EBW87}, 
or quantum control \cite{Roadmap2015}. 
We want to emphasize that our plotting choice of using the PROPS representation does not affect the theory
in Sec.~\ref{theorysection} as it directly operates on the $2N$-variate function defining the Wigner function.
All relevant operators in Sec.~\ref{theorysection}, including the star product, act linearly on its arguments resulting
in completely natural PROPS plots. In addition, the PROPS representation stresses---as intended---important nonlocality features
of the depicted Wigner function.

We recall that bosonic quantum systems (and similarly fermionic ones)
demonstrate different characteristics as compared to coupled spin systems with distinguishable spins. Foremost, 
the dimensionality of a bosonic quantum system is polynomial in the number of particles while 
the dimensionality of a coupled spin system is exponential in the number of spins. Also, due its permutation symmetry,
a bosonic system does not exhibit any localized properties and can be therefore
(for a fixed number of particles) \emph{naturally} 
embedded into a single spin with a large enough spin number
(see \cite{Dicke1954,stockton2003,toth2010,lucke2014,KZG17,sakurai1995modern,schwinger65}).
As discussed above, the same does not apply to general coupled spin systems and one needs to be cautious
in extending intuition from bosonic quantum systems to coupled spin systems considered in this work.

Finally, we want to clarify that this work does \emph{not} provide any progress on
reducing the complexity of simulating the time evolution of coupled spin systems.
Long-standing complexity hypotheses suggest that simulating the time evolution
of a coupled spin system with a classical computer
should (in general) have an exponential complexity in the number of spins
\cite{NC00}. We believe that this applies to both matrix methods and Wigner-function techniques.

\subsection{Structure of this work}

Our work is structured as follows: We start in Sec.~\ref{Examplessection}
with a set of introductory examples which establish and illustrate
the main ideas of our Wigner formalism for spins.
The theoretical methods 
that form the central results of this work
are detailed in Sec.~\ref{theorysection}
where the Wigner transformation of coupled 
spin operators and their star product are developed;
later parts of this work can be read first
as they do not explicitly depend on the detailed argument
contained in Sec.~\ref{theorysection}.
In Sec.~\ref{advancedexamples},
we apply our methods to advanced examples
further exploring the Wigner formalism in the case of two and three coupled spins 
and also considering
the creation of entanglement.
We discuss connections 
to other important concepts 
in Sec.~\ref{auxmaterial}
which includes
the quantum angular momentum, the infinite-dimensional Wigner formalism, 
quaternionic Wigner functions,
and the evolution of non-hermitian states.	
We conclude in Sec.~\ref{conclusion}, and certain details are deferred to appendices.

\section{Introductory examples\label{Examplessection}}
Our approach to directly determine the time evolution of a quantum system 
using Wigner functions is now illustrated with concrete examples, while the corresponding
theory is detailed in Sec.~\ref{theorysection} below. We start in Sec.~\ref{Ex_one}
with the case of a single spin $1/2$ and juxtapose the well-known matrix method
with our Wigner function approach.
We also analyze the case of 
two coupled spins $1/2$ (see Sec.~\ref{Ex_two}) and consider in particular the
time evolution under a scalar coupling. 

More advanced examples are deferred to Sec.~\ref{advancedexamples}
considering the evolution of two coupled spins $1/2$ under the CNOT gate (see Sec.~\ref{Ex_CNOT}),
and the evolution of three coupled spins $1/2$ (see Sec.~\ref{Ex_three}).
Recall that in all these cases the von-Neumann equation (see Eq.~\ref{NeumannEq})
determines the time evolution of the density operator by specifying its time derivative.

\subsection{Time evolution of a single spin\label{Ex_one}}

\subsubsection{Evolution of the density operator}
A simple example is presented
which considers the precession of a single spin 1/2 in an external
magnetic field.
Here, explicit matrices are used, and these
are decomposed into irreducible tensor operators.
In Sec.~\ref{Ex_one_Wigner}, the time evolution is then computed directly in 
the Wigner representation. 
Recall the irreducible tensor operators~\cite{EBW87}
%------------------------------------------------------------------
\begin{subequations}
	\label{TensorOperatorMatricesALL}
	\begin{align} \label{TensorOperatorMatrices}
	&\T_{00}:= \tfrac{1}{\sqrt{2}} \unity_2 =
	\tfrac{1}{\sqrt{2}}
	\begin{pmatrix}
	1 & 0\\
	0 & 1
	\end{pmatrix}, \quad
	\T_{1,-1}:= I_- =
	\begin{pmatrix}
	0 & 0\\
	1 & 0
	\end{pmatrix},\\
	&\T_{10}:= \sqrt{2} I_z =
	\tfrac{1}{\sqrt{2}}
	\begin{pmatrix}
	1 & 0\\
	0 & -1
	\end{pmatrix}, \quad
	\T_{11}:= -I_+ =
	\begin{pmatrix}
	0 & -1\\
	0 & 0
	\end{pmatrix}
	\label{TensorOperatorMatrices2}
	\end{align}
\end{subequations}
%------------------------------------------------------------------
for the case of a single spin $1/2$.
For arbitrary spin number $J$, the definition of tensor 
operators $\Tj_{jm}$ is based on their commutation relations 
%------------------------------------------------------------------
\begin{equation}
	\label{TOpDefALL}
	[ ^{J}I_z , \Tj_{jm} ] =  m \: \Tj_{jm} \; \text{ and }\;
	[ ^{J}I_{\pm} , \Tj_{jm} ] =  \sqrt{\left( j \mp m \right)\left( j \pm m + 1  \right)} \:    \Tj_{j, m\pm1},
\end{equation}
%------------------------------------------------------------------
as described by Racah \cite{Racah42},
where $^{J}\Jspace I_z$  and $^{J}\Jspace I_{\pm}= {}^{J} \Jspace I_x \pm  i\, {}^{J} \Jspace I_y$ 
are representations of arbitrary
spin-$J$  operators;\footnote{
	As usual, the Cartesian spin operators are defined as 
	$I_x=\sigma_x/2$, $I_y=\sigma_y/2$, and $I_z=\sigma_z/2$, where
	the Pauli matrices are
	$
	\sigma_x=\left(
	\begin{smallmatrix}
	0 & 1\\
	1 & 0
	\end{smallmatrix}
	\right)
	$,
	$
	\sigma_y=\left(
	\begin{smallmatrix}
	0 & -i\\
	i & 0
	\end{smallmatrix}
	\right)
	$, and
	$
	\sigma_z=\left(
	\begin{smallmatrix}
	1 & 0\\
	0 & -1
	\end{smallmatrix}
	\right)
	$.
}	
the  index $J$  is dropped in the spin-$1/2$ case.

An arbitrary spin-1/2 density matrix can be written as
$\rho = r_0 \, \unity_2 + \sum_\alpha r_\alpha I_\alpha$,
with $\alpha \in \{ x, y, z \}$.
Even though our Wigner representation is completely general and applicable 
to arbitrary density matrices and operators,
we omit the identity part $r_0 \, \unity_2 $
in some of the following examples
without affecting the time evolution and continue our discussion considering only the second term
$\sum_\alpha r_\alpha I_\alpha$.
This term 
is usually referred to as the deviation density matrix in quantum information processing 
(cf.\ Eq.~7.166 on p.~336 in \cite{NC00} or Eq.~2.5.13 on p.~47 in \cite{oliveira2011})
or as partial density matrix in magnetic resonance
(cf.\ Eq.~6 in \cite{farrar1990}, Eq.~2.125 on p.~55  in \cite{cavanagh1995}, or p.~243 in \cite{abragam1961}).
Although this simplification is also valid for individual quantum systems
(consisting of one or more coupled spins),
it is especially useful when considering thermal ensemble states 
for sufficiently large temperatures, i.e.\ for $r_0 \gg \sqrt{\sum_\alpha r_\alpha^2}$.

In our example, a rotation around the $z$ 
axis with an angular frequency $\w$ is generated
by the Hamiltonian
%------------------------------------------------------------------
\begin{equation}\label{H_single}
\mathcal{H}=\w I_z = \w \tfrac{1}{\sqrt{2}}\T_{10},
\end{equation}
%------------------------------------------------------------------
and the quantum state of a single spin at time $t=0$ is chosen
as the traceless deviation density matrix
%------------------------------------------------------------------
\begin{equation}\label{rho_single}
\rho(0)=I_x 
= \tfrac{1}{2} \T_{1,-1} -  \tfrac{1}{2} \T_{11}.
\end{equation}
%------------------------------------------------------------------
The time evolution is described by the von-Neumann equation, see Eq.~\eqref{NeumannEq},
and the first time derivative 
$i \partial \rho(0) / \partial t$ is determined by the commutator
%------------------------------------------------------------------
\begin{equation*}
[ \mathcal{H} ,  \rho(0)  ] 
=  \tfrac{\w}{2\sqrt{2}}  [  \T_{10}  \, , \T_{1,-1}  ]  
-  \tfrac{\w}{2\sqrt{2}}  [ \T_{10} \, , \T_{11}  ] 
= - \tfrac{\w}{2}  \T_{1,-1}
- \tfrac{\w}{2}  \T_{11} = i \omega I_y,
\end{equation*}
%------------------------------------------------------------------
whose form can also be inferred from the definitions in Eq.~\eqref{TOpDefALL}.
The solution of this differential
equation results in\footnote{Similarly as for $\rho(0)$, the time derivative of $\rho(t)$ decomposes 
	into a linear combination of the tensor operators $\T_{1,-1}$ and $\T_{11}$.
	It follows that the general solution can be parameterized as 
	$ \rho(t) =a(t)\, \T_{1,-1} - b(t)\, \T_{11} $ with $a(0)=b(0)=1/2$.
	This formula is substituted back into the von-Neumann equation
	[see Eq.~\eqref{NeumannEq}] and yields ${\partial} [   a(t) \T_{1,-1}  - b(t) \T_{11} ] /{\partial t}
	=   i  \w   [   a(t)  \T_{1,-1} + b(t)  \T_{11}  ]$,
	which splits up into the equations
	${\partial} a(t) /{\partial t}= i \w a(t)$ and ${\partial} b(t)/{\partial t} = -i \w b(t)$.
	Consequently, the solution is given by $a(t)=\exp{(i \w t)}/2$ and $b(t)=\exp{(-i \w t)}/2$.}
\begin{equation*}
\rho(t) =\tfrac{1}{2} e^{i \w t} \T_{1,-1} -  \tfrac{1}{2} e^{-i \w t} \T_{11}
= \cos(\w t)  I_x + \sin(\w t)  I_y.
\end{equation*}

\addtocounter{footnote}{1}
\footnotetext[\value{footnote}]{Hermitian operators
	result only in positive and negative values depicted
	as red (dark gray) and green (light gray).}
\newcounter{colors}
\setcounter{colors}{\value{footnote}}

%-------------------------
\begin{figure}[tb]
	\centering
	\includegraphics{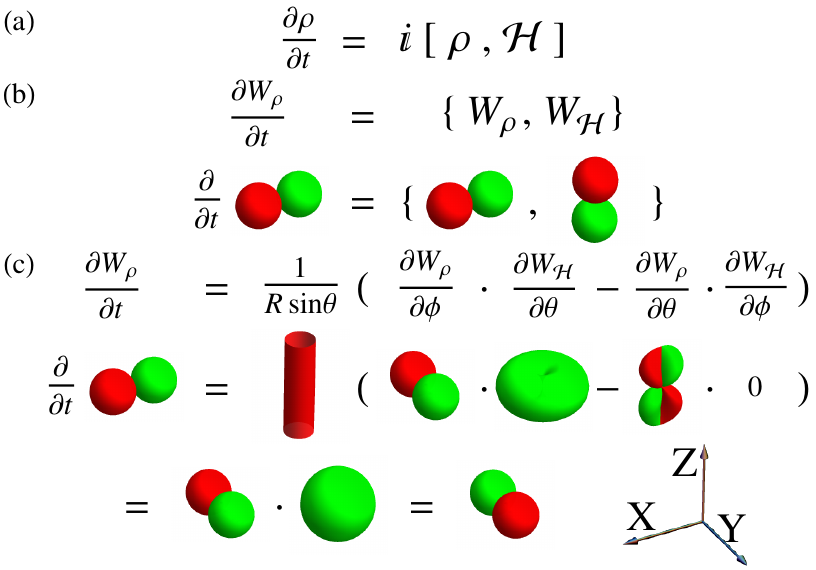}
	\caption{\label{graphicalPB}(Color online) (a) The equation of motion in the Hilbert space given by
		the von-Neumann equation.
		(b) The equation of motion of a single spin 1/2 in Winger space
		(upper row) and its graphical representation
		(lower row) in the particular case of $\mathcal{H}=\w I_z$
		with $\w=1$ and $\rho=I_x$.
		(c) Graphical representation of the Poisson bracket from (b).
		The spherical functions corresponding to the partial derivatives
		of the Wigner function $W_\rho=R \sin{\theta} \cos\phi$
		and the Hamiltonian $W_{\mathcal{H}}=R \,  \cos\theta$ are shown and
		$1/(R \sin{\theta})$ is visualized as an infinitely long
		cylinder.\footnoteref{\thecolors}
	}
\end{figure}

\subsubsection{Evolution of the Wigner functions\label{Ex_one_Wigner}}
Mirroring the preceding discussion in terms of matrices, 
the Hamiltonian $\mathcal{H}$ and
the traceless deviation density matrix $\rho(0)$   
from Eqs.~\eqref{H_single}--\eqref{rho_single}
are mapped to their Wigner functions
%-------------------------
\begin{equation} \label{InitSingleSpin}
W_{\mathcal{H}}(\theta, \phi) = \tfrac{\w} {\sqrt{2}} \Y_{10}= \w  \, R   \cos\theta
\; \text{ and } \;
W_{\rho}(\theta, \phi,t=0) = 
\tfrac{1}{2} (  \Y_{1,-1} - \Y_{11}  ) = R \sin{\theta} \cos\phi 
\end{equation}
%-------------------------
by replacing the tensor operators $\T_{jm}$ by the corresponding spherical harmonics
$\Y_{jm}=\Y_{jm}(\theta,\phi)$ \cite{Jac99}. This basic example conforms with
the general discussion in Sec.~\ref{AppendixWignerRepr} [see Eq.~\eqref{TensorOpSphHarm}];
note that  $R:=\sqrt{{3}/(8 \pi)}$. 

Here, the Wigner function $W_A (\theta,\phi) = |W_A(\theta,\phi)|\allowbreak \exp[i \eta (\theta,\phi)]$
of a single spin is visualized in the following way:
A surface is plotted whose surface element in the direction $(\theta,\phi)$ is at a
distance $|W_A(\theta,\phi)|$ from the origin.
The complex phase factor $ \exp[i \eta (\theta,\phi)]$ of the Wigner function is
represented by the color of its surface element.
This method visualizes spherical functions as three-dimensional shapes
(see Figs.~\ref{graphicalPB} and \ref{Figure1}).

The time evolution is governed 
by the shape of the appearing Wigner functions
via its angular derivatives.
The von-Neumann equation for a single spin $1/2$ translates in the Wigner 
representation to the equation (see Fig.~\ref{graphicalPB})
%-------------------------
\begin{equation} \label{Equation_Motion}
\frac{  \partial W_{\rho}(\theta, \phi, t)  }{\partial t}  =  
\{       W_{\rho}(\theta, \phi,t)  \:  ,  \: W_{\mathcal{H}}(\theta, \phi)   \} 
 =  \frac{1}{R \sin\theta} \left(\frac{  \partial W_{\rho}  }{\partial \phi} 
\frac{  \partial W_{\mathcal{H}}  }{\partial \theta} 
- \frac{  \partial W_{\rho}  }{\partial \theta} 
\frac{  \partial W_{\mathcal{H}}  }{\partial \phi}\right).
\end{equation}
%-------------------------
The Poisson bracket 
$\{W_{\rho}(\theta, \phi, t),\allowbreak{}W_{\mathcal{H}}(\theta, \phi) \}$ is further 
detailed in Sec.~\ref{singlespinstar}.
In our example, the time derivative at time $t=0$  is given by
%-------------------------
\begin{figure}[tb]
	\centering
	\includegraphics{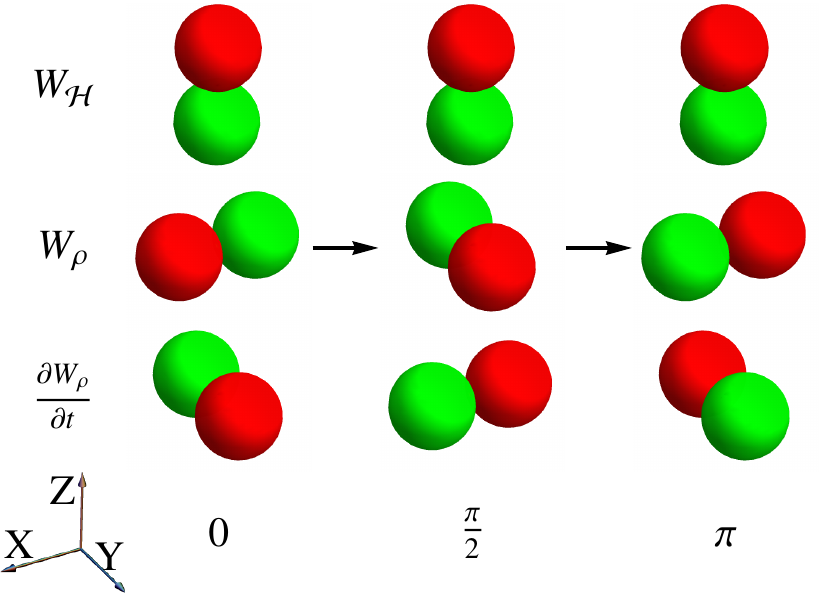}
	\caption{\label{Figure1}(Color online) Plots of the Wigner functions 
		$W_{\mathcal{H}}$  of the Hamiltonian $\mathcal{H}=\w I_z$ (upper row), 
		$W_{\rho}$  of the traceless deviation density matrix $\rho$ (middle row), and the time derivative 
		$\partial W_{\rho}(\theta, \phi, t)  / \partial t$ (lower row) corresponding to $\partial \rho / \partial t$ 
		at times $t=0$, $\w t = \pi /2$, and $\w t = \pi$.\footnoteref{\thecolors}
	}
\end{figure}
%-------------------------
%-------------------------
\begin{equation}
\frac{  \partial W_{\rho}(\theta, \phi, 0)  }{\partial t}
=  \w  R^2	\{   \sin\theta \cos\phi, \cos\theta    \} 
= \w 	R  \sin\theta \sin\phi. \label{OneSpinWignerDiffEquation}
\end{equation}
%-------------------------
Refer also to Figure~\ref{graphicalPB} for a graphical representation of this particular example.
The solution of the differential equation in Eq.~\eqref{OneSpinWignerDiffEquation}
is given by\footnote{
	The Wigner function can be written as 
	$W_{\rho}(\theta, \phi,t)=R \sin\theta [ a(t)  \cos\phi +  b(t) \sin\phi   ]$ with $a(0)=1$ and $b(0)=0$. 
	Substituting this parametrization back into Eq.~\eqref{Equation_Motion}, 
	one derives the differential  equation 
	$ \sin\theta  [  (\partial a(t) / \partial t)  \cos\phi + ( \partial b(t) / \partial t) \sin\phi   ]
	= \w \sin\theta [  a(t)  \sin\phi  - b(t)  \cos\phi  ]$, which splits up into
	$\partial a(t)/\partial t=-\w b(t)$ and $\partial b(t)/\partial t=\w a(t)$.}
%-------------------------
\begin{equation}
W_{\rho}(\theta, \phi,t)=R \sin\theta [ \cos( \w t)   \cos{\phi} + \sin( \w t) \sin\phi  ].
\end{equation}
%-------------------------
Figure~\ref{Figure1} shows the Wigner function 
of the Hamiltonian and the density matrix evolving in time, including the time derivative
of $W_{\rho}(\theta, \phi,t)$. 
All operators depicted in Figure~\ref{Figure1} are hermitian 
and consequently only the colors red (dark gray) and green (light gray) appear,
which correspond to
positive and negative real values in their Wigner functions.
Note how the 
shapes govern the rotation of $W_{\rho}(\theta, \phi,t)$ around the $z$ axis.
The state of a spin 1/2 can be characterized by the Bloch vector
and the unitary
time evolution translates to rotations of this three-dimensional vector. The Wigner 
function provides a description similar to the Bloch vector (refer to Section~\ref{Quaternions})
and its time evolution, which is
supported by the Poisson bracket corresponds to the rotation of the Wigner function
(refer to Section~\ref{connangmom}). Non-hermitian parts of the density operator are relevant for
coherent spectroscopy \cite{Glaser1998} but cannot be represented by a single Bloch vector. The 
Wigner function, however, provides a natural way to represent, visualize, and predict the
time evolution of these non-hermitian operators.
An example describing the time evolution of non-hermitian spin states
is detailed in Section~\ref{singlespinnonhermitian}.

%--------------------------------------------------------------------------------------------------------
\subsection{Time evolution of two coupled spins\label{Ex_two}}
As discussed in Sec.~\ref{Ex_one}, the Wigner function 
of single spin-$1/2$ states
is similar to the Bloch-vector description, and
the unitary time evolution translates to the rotation of these
Wigner functions.
For multiple, coupled spins the underlying quantum dynamics becomes significantly more involved,
and the time evolution can in general not be fully characterized
in terms of rotations of Bloch vectors (see Fig.~\ref{Figure2} below).
In contrast, the quantum state can still be represented uniquely by a single Wigner function 
which is visualized by its PROPS representation
using a linear combination of products of spherical harmonics.
The corresponding time evolution is governed by a generalization of the Poisson bracket.
In the following example, we present one of the simplest examples where the Bloch vector picture
breaks down. Even though the initial quantum state is representable by a Bloch vector,
the time evolution creates a superposition of states (see Fig.~\ref{Figure2} below).

\subsubsection{Evolution of the density matrix\label{Ex_two_matrix}}
We consider now the time evolution for an example of two coupled spins $1/2$.
Similarly as in Sec.~\ref{Ex_one}, we first rely on explicit matrices, and our approach
using Wigner functions is detailed in Sec.~\ref{twospinwignereval} below.
In order to simplify and highlight the transformation to the Wigner space, 
we will subsequently differentiate between
tensor operators acting on different spins: 
The linear operators $\Ta_{jm}=\T_{jm} \otimes \T_{00}$ and $\Tb_{jm}=\T_{00} \otimes  \T_{jm}$ 
act respectively on the first and second spin, and they are constructed using a 
tensor product leading to four-by-four matrices.
A bilinear operator $\Ta_{j_1m_1} \Tb_{j_2m_2}$ 
acts on both spins and consists of a matrix product of single-spin operators.
Details on definitions and properties of  product operators
are deferred to Sec.~\ref{ProdOpFormalism}
(see Table~\ref{tensordef}) and a short summary is given in
\ref{EmbeddedOperatorsExplanation}.

Let us now consider a system of two coupled spins which evolve under 
the bilinear Hamiltonian 
\begin{equation}\label{Ham_IIB1}
\mathcal{H}= \pi \JC 2 I_{1z} I_{2z} = \pi \JC 2 \Ta_{10} \Tb_{10},
\end{equation}
which can arise from a heteronuclear scalar or dipolar coupling. Here, we have applied 
the notations $\Ta_{10}:=\T_{10} \otimes  \T_{00}
=(\T_{10} / \sqrt{2})  \otimes (\sqrt{2} \T_{00})=I_z\otimes \unity_2=I_{1z}$ and 
$\Tb_{10}:=\T_{00} \otimes  \T_{10}=I_{2z}$ (see Sec.~\ref{ProdOpFormalism}).
In addition, we specify the traceless deviation density matrix at time $t=0$ as
\begin{equation}\label{rho_IIB1}
\rho(0) =  I_{1x} = (\Ta_{1,-1}  -  \Ta_{11}) /\sqrt{2}.
\end{equation}
Equation~\eqref{NeumannEq} determines
the time differential $i \partial \rho(0) / \partial t$ 
as the commutator
%-------
\begin{equation*}
[ \mathcal{H} ,  \rho(0)  ] 
=  - \sqrt{2} \pi \JC\, 
(   \Ta_{1,-1}  {+}  \Ta_{11} ) \Tb_{10} 
=  i \pi \JC  2 I_{1y} I_{2z}.
\end{equation*}
%-------
One deduces that only the four tensor operators $\Ta_{11}$, $\Ta_{1,-1}$, $
\Ta_{11} \Tb_{10} $, and $ \Ta_{1,-1} \Tb_{10}$ can appear in the decomposition
of $\rho(t)$, as $\partial^2 \rho(0) / \partial t^2 \propto \rho(t)$.\footnote{ 
	Computing
	the second time derivative $-\partial^2 \rho(0) / \partial t^2$  
	via the double commutator $[ \mathcal{H} ,[ \mathcal{H} ,  \rho(0)  ] ]
	=\pi^2 \JC^2 \rho(0)$ and applying the formulas 
	$[  \Ta_{10} \Tb_{10},  \Ta_{1,-1} \Tb_{10}]\allowbreak{}=-\Ta_{1,-1}/4$
	and $[  \Ta_{10} \Tb_{10},  \Ta_{11}\Tb_{10}] = \Ta_{1,1}/4$, the result follows.}
Consequently, the time-dependent deviation density matrix can be written 
as 
\begin{gather} \label{ABdefinition}
\rho(t)  = a(t) A + b(t) B\text{, where}\\
A =\tfrac{1}{\sqrt{2}}(\Ta_{1,-1}  {-}  \Ta_{11} ) = I_{1x}, \quad \nonumber
B =\sqrt{2} ( \Ta_{1,-1}  {+}  \Ta_{11} ) \Tb_{10}  =  2 I_{1y} I_{2z},
\end{gather}
and  $a(0)=1$ and $b(0)=0$. 
Substituting this back into Eq.~\eqref{NeumannEq}, one obtains
the solution\footnote{
	The differential equation $\partial[a(t) A  + b(t) B]/\partial t
	=\allowbreak{}   \pi \JC\allowbreak{}  [a(t)  B -  b(t) A]$,
	decomposes into
	$\partial a(t)/\partial t=-\pi \JC  b(t)$
	and $\partial b(t)/\partial t= \pi \JC a(t)$.
	The solution follows from  $a(0)=1$, $b(0)=0$, 
	$a(t)=\cos( \pi \JC t )$, and
	$b(t)=\sin( \pi \JC t )$.}
\begin{equation}
\rho(t) = \cos( \pi \JC t )  I_{1x} + \sin( \pi \JC t)   2 I_{1y} I_{2z} .
\end{equation}
The detectable NMR signal  
is proportional to $\cos( \pi \JC t )$, and
one obtains a doublet spectrum with equal intensities and lines separated by $\JC$.
%--------------------------------------------------------------------------------------------------------

%------------------------------------------------------------------------------------------------------------------------------------
\begin{table}[tb]
	\centering 
	\caption{\label{CartesianWignerRepr} Wigner 
		representations of the identity as well as linear and bilinear Cartesian operators;
		$\lambda=R / \sqrt{2 \pi}^{ N-1 }$ with $R=\sqrt{3/(8\pi)}$ and  $a,b\in\{x,y,z\}$.
		The projection $I_{k \beta}$ onto the pure state $\ket{\beta}$ of a single spin 
		is discussed in Sec.~\ref{CNOT_evol}.}
	\begin{tabular}{@{\hspace{4mm}}l@{\hspace{8mm}}l@{\hspace{16mm}}l@{\hspace{8mm}}l@{\hspace{4mm}} }
		\\[-2mm]
		\hline\hline
		\\[-2mm]
		A  &  $\wigner(A)$ & A & $\wigner(A)$
		\\[1mm] 
		\hline 
		\\[-2mm]
		%------------------------------------------------------------------------------------------------------------------------------------
		$\unity_{2^N}$  &   $1/\sqrt{2 \pi}^{N}$ &   $I_{k \beta}$
		&  $1/(2\sqrt{2 \pi}^{N})  - \lambda \cos \theta_k$  \\[1mm]
		$I_{kx}$     &    $\lambda \sin \theta_k \cos \phi_k$     
		&      $I_{kx}I_{\ell x}$  &  $ \sqrt{2 \pi}^N \lambda^2  \sin \theta_k \cos \phi_k \sin \theta_\ell \cos \phi_\ell $   
		\\[1mm] 
		$I_{ky}$  &   $\lambda \sin \theta_k \sin \phi_k$  
		&   $I_{kx}I_{\ell z}$  &  $ \sqrt{2 \pi}^N \lambda^2  \sin \theta_k \cos \phi_k \cos \theta_\ell $  
		\\[1mm] 
		$I_{kz}$  &   $\lambda \cos \theta_k $ &    $I_{k a}I_{\ell b}$    &    $\sqrt{2 \pi}^N \wigner(I_{k a}) \wigner(I_{\ell b}) $  
		\\[2mm] \hline \hline
	\end{tabular} 
\end{table}
%------------------------------------------------------------------------------------------------------------------------------------

\subsubsection{Evolution of Wigner functions\label{twospinwignereval}} 

We switch now to the Wigner picture and explain shortly how product operators 
in a two-spin system are represented as Wigner functions, while
details will be given in Sec.~\ref{AppendixWignerRepr} below.
Moreover, we translate the von-Neumann equation for two spins
into the Wigner picture.
This is then applied to
the example of Sec.~\ref{Ex_two_matrix}.

Wigner representations of operators acting on different spins are distinguished by 
different variables, thus an operator acting
on the first spin is transformed to its Wigner function by mapping the basis states $ \Ta_{jm}$ to their
corresponding spherical harmonics  
$\Ya_{jm} = \Y_{jm}(\theta_1,\phi_1) / \sqrt{4 \pi}$. Similarly, $ \Tb_{jm}$  
is mapped onto $\Yb_{jm} = \Y_{jm}(\theta_2,\phi_2) / \sqrt{4 \pi}$.
Product operators are constructed as simple pointwise products 
of their
Wigner functions and
$\Ta_{j_1m_1} \Tb_{j_2m_2}$ is mapped to the 
product 
$2 \pi \Ya_{j_1m_1} \Yb_{j_2m_2}= \Y_{j_1m_1}(\theta_1,\phi_1) \Y_{j_2,m_2}(\theta_2,\phi_2) /2$. 
Important examples are summarized in Table~\ref{CartesianWignerRepr}.
Suitable prefactors are introduced to ensure
consistent normalizations for matrix representations
and Wigner functions (see Sec.~\ref{AppendixWignerRepr}), and
the different normalization factors are also illustrated in Figure~\ref{GraphicalNormalization}.

The time evolution of the 
density matrix via the von-Neumann equation translates 
for Wigner functions of two spins
to the equation [see Sec.~\ref{TwoSpinSummary} and Corollary~\ref{twospinEQMcorollary}]
%-------------------------
\begin{equation} \label{TwoSpinTimeDeriv}
\hspace{-1.4mm}\partial W_{\rho} / \partial t = 
\sqrt{2 \pi}   \Proj ^{\{ 1,2 \}}   
(  \{ W_{\rho}  , W_{\mathcal{H}}  \}^{\{ 1 \}}  {+}  \{ W_{\rho}   ,  W_{\mathcal{H}}  \}^{\{ 2 \}}  ).
\end{equation}
%-------------------------
Here, the Poisson brackets from Eq.~\eqref{Equation_Motion} gain an additional index in order to identify
their spin dependence, i.e., the Poisson bracket $\{ f_a  , f_b  \}^{\{ 1 \}}$
contains derivatives with respect to the variables $\theta_1$ and $\phi_1$, 
while $\{ f_a  , f_b  \}^{\{ 2 \}}$ is defined with reference to $\theta_2$ and $\phi_2$.
As spherical harmonics with rank two or higher
are not allowed for spins $1/2$,
the projector $\Proj^{\{ 1,2 \}}$  removes these superfluous contributions, 
but leaves spherical harmonics $\Y_{jm}$
with rank  $j$ equal to zero or one unchanged. 
The Poisson brackets in Eq.~\eqref{TwoSpinTimeDeriv} can be simplified for
product operators $W_\rho=W_{\rho_1}(\theta_1,\phi_1)W_{\rho_2}(\theta_2,\phi_2)$ 
and $W_{\mathcal{H}}=W_{\mathcal{H}_1}(\theta_1,\phi_1)W_{\mathcal{H}_2}(\theta_2,\phi_2)$ 
into the form 
\begin{equation}
\partial W_{\rho} / \partial t = 
\sqrt{2 \pi}  \,   \{ W_{\rho_1}  , W_{\mathcal{H}_1}  \}^{\{ 1 \}}
\Proj ^{\{ 2 \}} (  W_{\rho_2}  W_{\mathcal{H}_2}   )    \label{ProdOpTimeEvolution}
+   \sqrt{2 \pi}  \,  \{ W_{\rho_2}  , W_{\mathcal{H}_2}  \}^{\{ 2 \}}
\Proj ^{\{ 1 \}} (  W_{\rho_1}  W_{\mathcal{H}_1}   ); 
\end{equation}
refer to Figure~\ref{GraphicalTimeEvolution}(a) for a visualization
of this computation.
In the PROPS representation,
product operators are indicated as overlapping circles 
(refer also to \ref{visualizationExplain}) and 
the overall Wigner function of a tensor product 
$W_{1,2}=W_1 W_2$ is given as a product of 
its parts. The corresponding Wigner functions $W_1$ is drawn in the left
circle and the Wigner functions $W_2$  in the right one. 

%-------------------------
\begin{figure}[tb]
	\centering
	\includegraphics{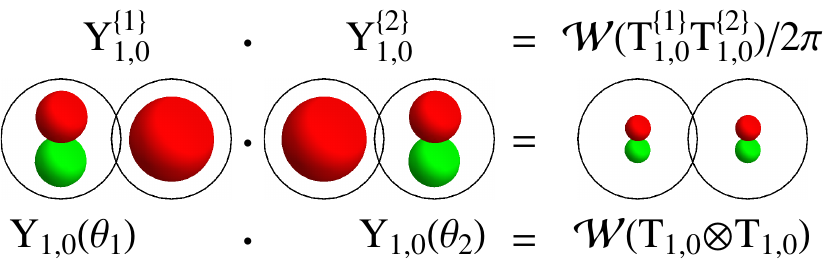}
	\caption{\label{GraphicalNormalization}
		(Color online)  The tensor
		product of the normalized $2\times2$ matrix $\T_{10}$ with itself results in a normalized $4\times4$ matrix 
		$\T_{10} \otimes \T_{10}$
		whose Wigner representation $\Y_{10}(\theta_1) \Y_{10}(\theta_2)$ 
		with $\Y_{10}(\theta_j)=\sqrt{2} R \cos\theta_j$
		is also normalized (lower row).
		While the $4\times4$ matrix $\Ta_{1,0}$ is normalized, the matrix product
		$\Ta_{10} \Tb_{10} = \T_{10} \otimes \T_{10}/2$ is not.  
		Similarly, 
		$\Ya_{10}=1/\sqrt{2\pi} R \cos \theta_1$ and $\Yb_{10}$ are normalized but their 
		pointwise product is not (middle and upper row). 
		Note that both the middle and the upper row is by a factor of $4\pi$ larger than the lower row.
		In summary,
		the norm changes in general
		when spherical functions are multiplied.\footnoteref{\thecolors}
	}
\end{figure}
%-------------------------

Using aforementioned techniques,
the Hamiltonian and the deviation density matrix from Eqs.~\eqref{Ham_IIB1}-\eqref{rho_IIB1}
can be transformed into their respective Wigner function
%-------------------------
\begin{gather} \label{InitTwoSpinH}
W_{\mathcal{H}} =4 \pi^2 \JC \, \Ya_{10} \Yb_{10}=\pi \JC 2 R^2 \cos{\theta_1} \cos{\theta_2},
\\ \label{InitTwoSpinR}
W_{\rho}(0) = 
( \Ya_{1,-1} -  \Ya_{1,1} ) /\sqrt{2}
=
R \sin{\theta_1} \cos{\phi_1}  /\sqrt{2 \pi}.
\end{gather}
%-------------------------
The time derivative of the initial Wigner function 
$ W_{\rho}(0)$ is now given by Eq.~\eqref{TwoSpinTimeDeriv} and
it depends only on the variables $\theta_1, \phi_1$, consequently, the second
Poisson bracket $\{ W_{\rho}(0), W_{\mathcal{H}}    \}^{\{ 2 \}}$ is zero. 
One applies Eq.~\eqref{Equation_Motion}
and obtains up to projections
%-------------------------
\begin{align}
\partial W_{\rho}(0) / \partial t 
& = \sqrt{2 \pi}  \{  W_{\rho}(0), W_{\mathcal{H}}    \}^{\{ 1 \}}   
 = \pi \JC 2 R^3   \{ \sin{\theta_1} \cos{\phi_1} , \cos{\theta_1}   \}^{\{ 1 \}} \cos{\theta_2}
\nonumber \\ 
& =   \pi \JC 2 R^3   
[ \sin{\theta_1}   \frac{\partial \cos{\phi_1}}{\partial \phi_1}     
\frac{\partial \cos{\theta_1}}{\partial \theta_1}]  \cos{\theta_2}   /(R \sin{\theta_1}) \nonumber \\
& =  \pi \JC 2 R^2 \cos{\theta_2} \sin{\theta_1}  \sin{\phi_1}  
=\pi \JC  \wigner (2  I_{1y} I_{2z} ).
\end{align}
%-------------------------
Here, $\wigner (2  I_{1y} I_{2z} )$ denotes the Wigner transformation of $2  I_{1y} I_{2z} $, 
refer to Table~\ref{CartesianWignerRepr}.
For the graphical representation of this computation refer to Figure~\ref{GraphicalTimeEvolution}(a).
One deduces that the time derivative of $\wigner (2  I_{1y} I_{2z} )$
is up to projections  proportional to
%-------------------------
%-------------------------
\begin{figure}[tb]
	\centering
	\includegraphics{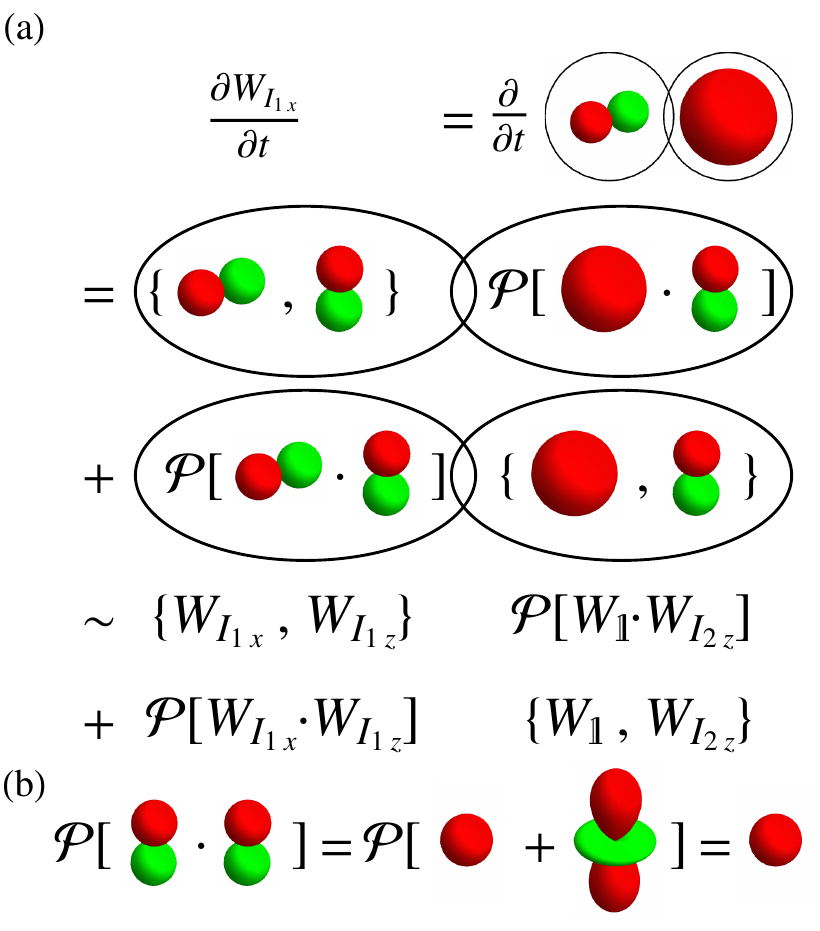}
	\caption{\label{GraphicalTimeEvolution}(Color online) (a) Graphical representation of the equation of motion
		for two spins 1/2 for the Hamiltonian $\mathcal{H}=\pi \JC 2 I_{1z} I_{2z}$ and the deviation density matrix
		$\rho=I_{1x}$, see Eqs.~\eqref{TwoSpinTimeDeriv}-\eqref{ProdOpTimeEvolution}. 
		(b) Graphical representation of the projection onto rank-one and rank-zero
		spherical harmonics for the example of $\Proj (\cos^2{\theta_2} )=1/3$ as
		discussed around Eq.~\eqref{decomp}.\footnoteref{\thecolors}
	}
\end{figure}
%-------------------------
\begin{gather*}
\partial^2 W_{\rho}(0) / \partial t^2   \propto  \partial  
(\cos{\theta_2} \sin{\theta_1}  \sin{\phi_1} ) / \partial t 
\\
\propto  \{ \sin{\theta_1}  \sin{\phi_1} , \cos{\theta_1}    \}^{\{ 1 \}}  \cos{\theta_2}  \cos{\theta_2} 
+ \cos{\theta_1}  \sin{\theta_1}  \sin{\phi_1}  \{  \cos{\theta_2}  ,  \cos{\theta_2}   \}^{\{ 2 \}}.
\end{gather*}
%-------------------------
Since $ \{ \cos{\theta_2}   , \cos{\theta_2}   \}^{\{ 2 \}} =0$ and $ \{ \sin{\theta_1}  \sin{\phi_1} , \cos{\theta_1}    \}^{\{ 1 \}}
=-\sin{\theta_1}\cos{\phi_1}\sin{\theta_1}/(R \sin{\theta_1})$, we obtain up to projections that
%-------------------------
\begin{equation*}
\partial \wigner (2  I_{1y} I_{2z} )/ \partial t= 
- 2 \sqrt{3} R^2 \pi \JC \sin{\theta_1}  \cos{\phi_1}   \cos^2{\theta_2}.
\end{equation*}
%-------------------------
It is however important to understand that the term 
\begin{equation}\label{decomp}
\cos^2{\theta_2}=[ \sqrt{4 \pi} \Y_{00}(\theta_2, \phi_2) 
+ 4 \sqrt{ \pi/ 5} \: \Y_{20}(\theta_2, \phi_2) ] /3
\end{equation}
linearly decomposes into
spherical harmonics of rank zero and two, as shown in Figure~\ref{GraphicalTimeEvolution}(b).
After applying the projector $\Proj^{\{ 2 \}}$ from Eq.~\eqref{ProdOpTimeEvolution}, 
only a term proportional to $\Y_{00}=1/\sqrt{4 \pi}$ remains;\footnote{
	The term $\cos^2{\theta_2}$ is proportional to $\wigner(I_{2z})\wigner(I_{2z})$.
	Note that in general $\Proj ^{\{ k \}} \wigner (I_{ka}) \wigner (I_{kb})= \delta_{ab}/ [4(2\pi)^{N}]$
	with $a,b\in\{x,y,z\}$ holds for the pointwise product of Wigner functions, where $N$ denotes the number of spins $1/2$.}
and this leads to
%-------------------------
\begin{figure}[tb]
	\centering
	\includegraphics{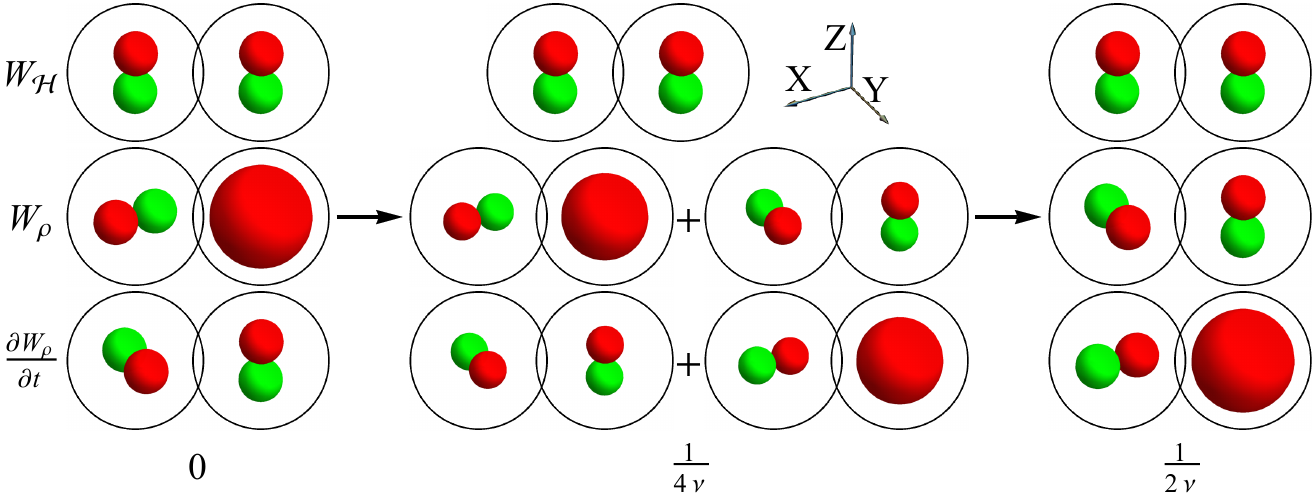}
	\centering
	\caption{\label{Figure2} (Color online) Time evolution of the initial deviation density
		matrix $I_{1x}$ under the bilinear 
		coupling Hamiltonian $\mathcal{H}=\pi \JC 2 I_{1z} I_{2z}$,
		computed with our Wigner formalism.
		PROPS representation of the Wigner 
		functions of the Hamiltonian
		(upper row), the traceless deviation density matrix (middle row), 
		and the time derivative of the deviation density matrix (lower row).
		Note that $\rho(0)=I_{1x}$,  $\rho[1/(4\JC)]=(I_{1x}+2 I_{1y} I_{2z})/ \sqrt{2}$, and 
		$\rho[1/(2\JC)]=2 I_{1y} I_{2z}$.\footnoteref{\thecolors}
	}
\end{figure}
%-------------------------
%-------------------------
\begin{equation*}
\partial \wigner (2  I_{1y} I_{2z} )/ \partial t= - \tfrac{ 2 R^2 \pi \JC}{ \sqrt{3}} \sin{\theta_1}  \cos{\phi_1}
=\wigner (- \pi \JC I_{1x}).
\end{equation*}
%-------------------------
It is now apparent, that the second time derivative of $W_{\rho}(0)$ 
is proportional to $W_{\rho}(0)$, i.e.,
$\partial^2 W_{\rho}(0) / \partial t^2   =-(\pi \JC)^2 W_{\rho}(0)$ and that the time
evolution of $W_{\rho}$ is parametrized by only 
two Wigner functions
%-------------------------
\begin{equation*}
W_A:=\wigner (I_{1x} ) =  R \sin{\theta_1}  \cos{\phi_1} / \sqrt{2 \pi},
\quad W_B:=\wigner ( 2 I_{1y} I_{2z} ) =   2 R^2 \sin{\theta_1}  \sin{\phi_1} \cos{\theta_2} .
\end{equation*}
%-------------------------
Note the similarity with Eq.~\eqref{ABdefinition}. The solution for the Wigner
function is then given by\footnote{
	The parametrized
	Wigner function $W_{\rho}(t)= a(t) W_A + b(t) W_B$ is substituted into Eq.~\eqref{TwoSpinTimeDeriv},
	and one obtains $\partial[  a(t) W_A + b(t) W_B ] /\partial t = \pi \JC[ a(t) W_B - b(t) W_A]$. 
	This splits into $\partial a(t)/\partial t=-\pi \JC b(t)$
	and $\partial b(t)/\partial t=\allowbreak{} \pi \JC a(t)$ and results in
	$W_{\rho}(t)= \cos( \pi \JC t ) W_A + \sin( \pi \JC t ) W_B $.}
%-------------------------
\begin{equation}
W_{\rho}(t)  =  R\, c(t)   \sin{\theta_1}  \cos{\phi_1} / \sqrt{2 \pi}
 + 2 R^2 s(t)   \sin{\theta_1}  \sin{\phi_1}  \cos{\theta_2}, \label{twospinexamplesolution}
\end{equation}
%-------------------------
where $c(t):=\cos( \pi \JC t )$ and $s(t):=\sin( \pi \JC t)$.
Figure~\ref{Figure2} illustrates the time evolution of
Wigner functions for
the Hamiltonian $\mathcal{H}= \pi \JC 2 I_{1z} I_{2z}$
(upper row), the traceless deviation density matrix
$\rho(t)=c(t)  I_{1x} + s(t)   2 I_{1y} I_{2z}$ (middle row),
and the corresponding time derivative 
$\partial \rho  / \partial t = -s(t)  I_{1x} + 
c(t)   2 I_{1y} I_{2z}$ (lower row) 
at different times
$t=0$, $\pi \JC t = \pi /4$, and $\pi \JC t = \pi /2$. 
Arbitrary operators are visualized in the PROPS representation by 
decomposing them into sums of product operators
(refer also to \ref{visualizationExplain}).
An alternative representation of the Wigner function in Eq.~\eqref{twospinexamplesolution}
based on a decomposition into non-hermitian operators is given in
Sec.~\ref{twospinnonheritiansection}.

\section{Theory: Wigner formalism for the time evolution of coupled spins}\label{theorysection}

The central parts of the Wigner formalism for coupled spins are now systematically developed.
Most of our theoretical concepts were
already introduced
in Sec.~\ref{Examplessection} using easily understandable examples.
We present now the mathematical details for our Wigner formalism
for coupled spins
which form the main results of this work.
Identifying
the underlying physical principles,
the description of the time evolution of coupled spin systems 
is consequently derived via our theoretical approach.

We want to emphasize that our theoretical approach, 
which relies on the irreducible tensor operators
is applicable to arbitrary density matrices and operators
of coupled spin systems, even though
our technical tools (including Clebsch-Gordan coefficients and
Wigner $6$-$j$ symbols \cite{messiah1962})
might spuriously suggest some
superficial similarity 
to the description of indistinguishable particles using reducible 
representations of the special unitary group of dimension two. 
In particular, the star product specifying the time evolution
is not determined by a simple addition
of angular momenta, not even for a single spin [cf.\ Eq.~\eqref{StarProdSph} below].

First, we recall basic properties of product operators and the tensor product of matrices
(see Sec.~\ref{ProdOpFormalism}); the main properties are summarized in 
Table~\ref{tensordef}. We detail the Wigner transformation of spin operators for single and coupled
spin systems in Sec.~\ref{AppendixWignerRepr}. This then leads 
in Sec.~\ref{singlespinstar} to a
simplified approach
to compute the star product of single-spin-1/2 operators and allows us to 
derive the corresponding
equation of motion. In Sec.~\ref{app_coupled}, the star product is then extended 
to coupled spin-$1/2$ operators and the corresponding equation of motion is determined. 
Finally, we provide an optimized form of
our formalism for two and three coupled spins $1/2$ as well as natural Hamiltonians
for multiple spins $1/2$ (see Sec.~\ref{summaryofresults}).

\begin{table}[tb]
	\centering 
	\caption{Normalized product operators and tensor products for 
		basis operators embedded 
		into an $N$-spin system; see also \ref{EmbeddedOperatorsExplanation}.
		The names in the first column refer to the case when
		all indices $j_k > 0$; but the prefactors are correct even if some 
		$j_k$ are zero.
		\label{tensordef}}
	\begin{tabular}{@{\hspace{1mm}}l@{\hspace{3mm}}l@{\hspace{2mm}}l@{\hspace{1mm}}}
		\\[-2mm]
		\hline\hline
		\\[-2mm]
		Type  &  Product-operator notation & Tensor-product notation 
		\\[1mm] 
		\hline 
		\\[-2mm]
		%------------------------------------------------------------------------------------------------------------------------------------
		Linear  &  $\Tj^{\{k\}}_{jm}$   & 		
		$  \underbrace{\Tj_{00}}_{\#1} \otimes \: \cdots \otimes \Tj_{00} \otimes  
		\underbrace{\Tj_{jm}}_{\#k} \otimes \: \Tj_{00}  \otimes \cdots  \otimes \underbrace{\Tj_{00}}_{\#N}$   
		\\[7mm] 
		%------------------------------------------------------------------------------------------------------------------------------------
		Bilinear  &  $\sqrt{2J{+}1}^{N}  \Tj^{\{k\}}_{j_k m_k} \Tj^{\{\ell\}}_{j_\ell m_\ell}$   & 		
		$\underbrace{\Tj_{00}}_{\#1} \otimes \: \cdots \otimes
		\underbrace{\Tj_{j_k m_k}}_{\#k} \otimes \:    \cdots 
		\otimes   \underbrace{\Tj_{j_\ell m_\ell}}_{\#\ell} \otimes \: \cdots  \otimes \underbrace{\Tj_{00}}_{\#N}$ 
		\\[7mm] 
		%------------------------------------------------------------------------------------------------------------------------------------
		$M$-linear  &  
		$\sqrt{2J{+}1}^{N(M-1)}$
		 & 
		$\bigotimes_{k\in\{1,\ldots,N\}}	A_k,$
		\\ & \hspace{2mm} $\times \prod_{k\in K, \abs{K}=M} \Tj^{\{k\}}_{j_{k} m_{k}}$ 
		   & \hspace{6.5mm} $\text{where }
		   A_k:=\begin{cases}
		   \Tj_{j_{k} m_{k}} & \text{if } k \in K,\\
		   \Tj_{00} & \text{otherwise.}
		   \end{cases}
		     $
		\\[7mm] 
		%------------------------------------------------------------------------------------------------------------------------------------
		$N$-linear  &  $\sqrt{2J{+}1}^{N(N-1)} $   & 		
		$\Tj_{j_1 m_1} \otimes \cdots  \otimes \Tj_{j_N m_N}$     
		\\ & \hspace{2mm} $\times \Tj^{\{1\}}_{j_1 m_1} \cdots \Tj^{\{N\}}_{j_N m_N}$ 
		\\[2mm] \hline \hline
	\end{tabular} 
\end{table}

\subsection{Product-operator and tensor-product notation\label{ProdOpFormalism}}
We recapitulate elementary definitions and properties of tensor operators acting
on a coupled $N$-spin system consisting of spin-$J$ particles. For single spins, the tensor components $\Tj_{jm}$
are indexed with rank $j\in\{0,\ldots,2J\}$ and order $m\in\{-j,\ldots,j\}$; the index $J$ can be dropped if $J=1/2$.
Recall the defining relations of tensor operators in Eq.~\eqref{TOpDefALL}, and their matrix elements
	given in the standard basis $|Jm\rangle$ can be specified in terms of
	Clebsch-Gordan coefficients \cite{messiah1962,Brif98,BL81,Fano53}
	\begin{equation}
	[\Tj_{jm}]_{m_1 m_2} := \langle J m_1| \Tj_{jm} |Jm_2\rangle =  \sqrt{\tfrac{2j{+}1}{2J{+}1}} \, C^{J m_1}_{J m_2, j m}
	=(-1)^{J-m_2}\, C^{jm}_{Jm_1J,-m_2},
	\end{equation} 
where $m_1,m_2 \in \{J,\ldots,-J\}$.
The single-spin operator $\Tj_{jm}$ is normalized and
can be embedded as the product operator 
\begin{equation} \label{Eq:embedded}
\Tj^{\{k\}}_{jm} :=  \Tj_{00} \otimes  \cdots \otimes \Tj_{00} \otimes 	
\Tj_{jm} \otimes  \Tj_{00} \otimes \cdots  \otimes \Tj_{00}
\end{equation}
acting on the $k$th spin of
an $N$-spin system 
(recall that $\unity_{2J{+}1}=\sqrt{2J{+}1} \: \Tj_{00}$).
More generally, the embedded form of a single-spin operator $A$ acting
on the $k$th spin is denoted by $A^{\{k\}}:=\Tj_{00} \otimes  \cdots \otimes \Tj_{00} \otimes 	
A \otimes  \Tj_{00} \otimes \cdots  \otimes \Tj_{00}$.
We consider products of single-spin tensor operators, and certain cases are summarized 
in Table~\ref{tensordef} as normalized product operators $A$ with $\tr(A^\dagger A)=1$, 
Table~\ref{tensordef} also shows the complementary tensor-product notation, and
we will
switch between product operators and the tensor-product notation.
Elementary properties of product operators can usually directly be inferred from properties of tensor 
products of operators $A_k$ and $B_k$ with $k\in\{1,\ldots,N\}$:
\begin{gather*}
        (A_1 \cdots A_N) \otimes (B_1 \cdots B_N) = (  A_1 {\otimes} B_1  ) 
	\: \cdots \: ( A_N {\otimes} B_N ),\;
	( A_1  \otimes \cdots  \otimes A_N )^{\dagger} = A_1^{\dagger} \otimes 
	\cdots \otimes A_N^{\dagger},
	\\
	\text{and } \; \tr( A_1  \otimes \cdots  \otimes A_N ) = \tr( A_1) 
	\: \cdots \:  \tr( A_N ).
\end{gather*}
Moreover, we gather the following properties of embedded single-spin product operators:
\begin{lemma}\label{single-spin}
	Embedded single-spin product operators have the following properties:
	\begin{align*}
	& \text{\textup{(a)}}\; \tr[ (  \Tj^{ \{ k_1 \}}_{j_1 m_1} )^{\normalfont\selectfont\dagger}  \Tj^{ \{ k_2 \}}_{j_2 m_2}  ]
	=\delta_{k_1 k_2} \delta_{j_1 j_2} \delta_{m_1 m_2}, \\
	& \text{\textup{(b)}}\; [\Tj^{\{k_1\}}_{j_1 m_1}, \Tj^{\{k_2\}}_{j_2 m_2}] = 
	\tfrac{\delta_{k_1 k_2}}{\sqrt{2J{+}1}^{N{-}1}} ( [\Tj_{j_1 m_1}, \Tj_{j_2 m_2}] )^{\{k_1\}}, \\
	& \text{\textup{(c)}}\;  \Tj^{\{k\}}_{j_1 m_1} \Tj^{\{k\}}_{j_2 m_2}= 
	\tfrac{1}{\sqrt{2J{+}1}^{N{-}1}} (\Tj_{j_1 m_1} \Tj_{j_2 m_2})^{\{k\}}.
	\end{align*}
\end{lemma}
These properties imply that 
normalized products of single-spin tensor operators give rise to
a basis of the full operator space of $N$ spins (cf.\ Table~\ref{tensordef}):
\begin{lemma} \label{basis}
	The normalized product operators 
	$$\sqrt{2J{+}1}^{N ( N-1)}  \Tj^{\{1\}}_{j_1 m_1} 
	\cdots \Tj^{\{N \}}_{j_N m_N}
	=\Tj_{j_1 m_1} \otimes \cdots  \otimes \Tj_{j_N m_N}$$
	form an orthonormal basis of the full $N$-spin system, and an arbitrary spin operator $A$
	can be decomposed as 
	$$
	A=  \hspace{-4mm}
	\sum_{j_1 m_1 \ldots j_N m_N} \hspace{-4mm}
	a_{j_1 m_1 \ldots j_N m_N } 
	\sqrt{2J{+}1}^{N ( N-1 )}   \Tj^{\{1\}}_{j_1 m_1}     \cdots \Tj^{\{N\}}_{j_N m_N}.
	$$
\end{lemma}
Given an $M$-linear operator acting on $M \leq N$ spins, all indices $j_k$ in the decomposition of $A$ 
given in Lemma~\ref{basis}
are greater equal to zero for $k\in\{1,\ldots,N\}$. Yet, certain indices $j_k$ have to be zero 
if $M < N$. \ref{EmbeddedOperatorsExplanation} provides a non-technical tutorial
on elementary
properties of product operators.

\subsection{The Wigner formalism for spins \label{AppendixWignerRepr}}
We describe the Wigner representation of spin operators
which are mapped by the Wigner transformation
to spherical functions. This bijective mapping fulfills the 
so-called Stratonovich postulates (and generalizations thereof)
which are discussed in \ref{stratpostulates}
for single spins as well as multiple coupled spins.

\subsubsection{Wigner representation of single spins}
%-------------------------
The continuous Wigner representation of an arbitrary operator $A$ acting on a single spin $J$ 
is defined as \cite{VGB89}
\begin{equation}\label{WignerTransform}
\wigner ( A ) := W_{A} (\theta, \phi)= \tr [  \Delta_{J}(\theta, \phi) A ],
\end{equation}
where $\wigner(A)$ denotes the Wigner transform of $A$.
In Eq.~\eqref{WignerTransform}, we have used the kernel
$\Delta_{J}(\theta, \phi)$ for a single spin $J$ which 
maps spin operators onto spherical functions, and it
is given by\footnote{
	We verify that $ \Delta_{J}(\theta, \phi)
	= \sum_{j=0}^{2J}  \sum_{m=-j}^{j}  \Y^*_{j,-m} \:  \Tj_{j,-m} = \Delta_{J}^{\dagger}(\theta, \phi)
	$
	is hermitian by using the Condon-Shortley phase convention and 
	$\Delta_{J}^{\dagger}(\theta, \phi) =  \sum_{j=0}^{2J}  \sum_{m=-j}^{j}  \Y_{jm} \:  \Tj^{\dagger}_{jm}$.}
%-------------------------
\begin{equation} \label{SingleKernelDefinition}
\Delta_{J}(\theta, \phi) := \sum_{j=0}^{2J}  \sum_{m=-j}^{j}  \Y^{*}_{jm}  (\theta, \phi)   
\:  \Tj_{jm} = \Delta_{J}^{\dagger}(\theta, \phi).
\end{equation}
%-------------------------
The form of the kernel builds on the  work of \cite{stratonovich,VGB89,Brif98,Brif97},
see, in particular, Eqs.~(4.16)-(4.17) in \cite{Brif98}, Eq.~2.14 in \cite{VGB89}, and Eq.~(9) in \cite{StarProd}. 
Here, the 
tensor operators
$ \Tj_{jm}$ for a given spin number $J$ form
an orthonormal set of basis operators for $(2J{+}1)\times (2J{+}1)$ matrices, i.e.,
%-------------------------
\begin{equation}\label{TOrtRelat}
\tr (   \Tj_{j_1 m_1}^{\dagger} \Tj_{j_2 m_2}  )
=
\delta_{j_1 j_2} \delta_{m_1 m_2};
\end{equation}
%-------------------------
likewise the spherical harmonics $\Y_{jm}  (\theta, \phi)$~\cite{Jac99}
are orthonormal with respect to the scalar product
%-------------------------
\begin{equation} 
\int_{\theta=0}^{\pi} \int_{\phi=0}^{2 \pi}  \Y^{*}_{j_1 m_1} (\theta, \phi)   
\Y_{j_2 m_2} (\theta, \phi)  \sin{\theta} \,\mathrm{d}\theta  \, \mathrm{d}\phi 
=
\delta_{j_1 j_2} \delta_{m_1 m_2}. \label{OrtRelat}
\end{equation}
%-------------------------

Equations~\eqref{WignerTransform} and  \eqref{TOrtRelat} imply that the 
Wigner representation of a tensor operator $\Tj_{j m}$ is equal to
the spherical harmonic $\Y_{j m} (\theta, \phi)$, i.e.,
%-------------------------
\begin{equation} \label{TensorOpSphHarm}
\wigner( \Tj_{j m} ) = W_{\Tj_{j m}} (\theta, \phi)=\Y_{j m} (\theta, \phi).
\end{equation}
%-------------------------
Tensor operators can be reconstructed from their Wigner representation by applying 
the inverse Wigner transformation (usually referred to as the Weyl transformation) which is defined as
%-------------------------
\begin{equation} 
\wigner^{-1} [   F (\theta, \phi)   ] := 
\int_{\theta=0}^{\pi} \int_{\phi=0}^{2 \pi}  \Delta_{J}(\theta, \phi) F (\theta, \phi) 
\: \sin{\theta} \,\mathrm{d}\theta  \, \mathrm{d}\phi . \label{inverseWigner}
\end{equation}
%-------------------------
By substituting  $F=W_A$ in Eq.~\eqref{inverseWigner}, one obtains that
$ A =  \wigner^{-1} [   W_{A} (\theta, \phi)   ]$.
The orthonormality relation of Eq.~\eqref{OrtRelat} implies 
that the inverse Wigner transformation of the spherical harmonics 
${\Y_{jm}(\theta, \phi) }$ are given by the tensor operators $\Tj_{jm}$.

\subsubsection{Wigner representation of coupled spins \label{ProdOpWignerTransform}}
We now generalize the definition of the kernel for a single spin
[see Eq.~\eqref{SingleKernelDefinition}] to multiple spins $J$, but for simplicity 
with identical $J$ for each spin, while a generalization
to systems that are composed of particles of different spin number $J$ is straightforward. 
\begin{result}\label{result1}
For $N$ coupled spins, the kernel is defined as the $N$-fold tensor product 
%-------------------------
\begin{equation*}
\Delta^{\{ 1 \ldots N \}}_{J}:= 
\Delta^{\{ 1 \ldots N \}}_{J}(\theta_1,\phi_1, \ldots,  \theta_N,  \phi_N) 
=  \bigotimes_{k=1}^{N} \Delta_{J} (\theta_k,  \phi_k) 
\end{equation*}
%-------------------------
of the individual
kernels and involves a set of $2N$ variables 
$(\theta_1,\phi_1, \ldots, \theta_N,  \phi_N)$ describing points on $N$ 
spheres.\footnote{A similar result has been attained in Eq.~(9) of \cite{tilma2016}.}
Using the definition of the  kernel $\Delta_{J} (\theta,  \phi)$ from Eq.~\eqref{SingleKernelDefinition}
and applying the correspondence of Table~\ref{tensordef},
one obtains the explicit form
%-------------------------
\begin{align}
 \Delta^{\{ 1 \dots N \}}_{J} 
& =\sum_{ \substack{j_1 \dots j_N \\ m_1 \dots m_N}} 
[\Tj_{j_1 m_1} \otimes \cdots  \otimes \Tj_{j_N m_N} ]\;
\mathcal{Y}
\nonumber \\
&	 = \sum_{ \substack{j_1 \dots j_N \\ m_1 \dots m_N}} 
[\sqrt{2J{+}1}^{N ( N-1 )}  \Tj^{\{1\}}_{j_1 m_1} \cdots \Tj^{\{N\}}_{j_N m_N} ]\;
\mathcal{Y}, \label{MultispinKernelDefinition}
\end{align}
where $\mathcal{Y}:=\prod_{k=1}^N \Y^{*}_{j_k m_k} (\theta_k, \phi_k)$.
Consequently, this defines
the Wigner transformation 
as
%-------------------------
\begin{equation}\label{MultiWignerTransform}
\wigner( A  ) :=
W_A(\theta_1,\phi_1, \ldots, \theta_N,  \phi_N)=
\tr (  \Delta^{\{ 1 \dots N \}}_{J} A  ),
\end{equation}
%-------------------------
and it satisfies the generalized Stratonovich postulates described in \ref{stratpostulatesmulti}.
\end{result}
We now apply orthonormality properties of
tensor operators 
[see Lemma~\ref{single-spin}(a)] and verify that
the Wigner representation for the linear embedded tensor operator $\Tj^{\{ k \}}_{j m}$ 
is proportional to the spherical harmonic $\Y_{jm}(\theta_k,\phi_k)$,
which depends on the angular variables $\theta_k$ and $\phi_k$. More precisely, we obtain the relation
%-------------------------
\begin{equation} 
\wigner ( \Tj^{\{ k \}}_{j m} ) = \Y^{\{k\}}_{jm} :=  
\underbrace{\Y_{00}}_{\#1}  
\cdots\hspace{-2mm} \underbrace{\Y_{00}}_{\#(k-1)}  \underbrace{\Y_{jm}(\theta_k,\phi_k)}_{\#k}
\underbrace{\Y_{00}}_{\#(k+1)} \hspace{-2mm} \cdots   
\underbrace{\Y_{00}}_{\#N}  
\phantom{:}=\Y_{jm}(\theta_k,\phi_k)/\sqrt{4 \pi}^{N-1}. \label{NSpinSpH}
\end{equation}
%-------------------------
Similarly as for the Wigner transformation, 
the inverse Wigner transformation 
of a spherical function $F=F(\theta_1,\phi_1, \ldots, \theta_N,  \phi_N)$
is generalized to multiple spins as 
%-------------------------
\begin{equation}\label{eq_inv_Wig_multiple}
\wigner^{-1}(F)
:= \hspace{-1mm} \idotsint\displaylimits_{\theta_k,\phi_k\geq 0}^{\theta_k\leq \pi, \phi_k\leq 2\pi} \hspace{-1mm}
\Delta^{\{ 1 \dots N \}}_{J} F 
\prod_{k=1}^{N} \sin{\theta_k}\, \mathrm{d}\phi_k \, \mathrm{d}\theta_k  .
\end{equation}
%-------------------------
Setting $F=W_A$ in Eq.~\eqref{eq_inv_Wig_multiple} also verifies that
$ A =  \wigner^{-1}  [   W_{A} (\theta_1,\phi_1 \dots \theta_N,  \phi_N)  ]$ holds.
In particular, the inverse Wigner transformation maps the spherical harmonic $\Y_{j m}^{\{ k \}}$ 
with variables $(\theta_k, \phi_k)$ to the linear tensor operators $ \Tj_{j m}^{\{ k \}}$ acting on the $k$th spin.
Finally, our approach establishes that the Wigner representation of 
products of embedded tensor operators can be written as products of the corresponding 
spherical harmonics involving different variables, i.e.,
\begin{align}
\wigner( \sqrt{2J{+}1}^{N ( N-1 )}   
\: \Tj_{j_1 m_1}^{\{ 1 \}} \cdots\: \Tj_{j_N m_N}^{\{ N \}} ) 
&=  \wigner(  
\:\Tj_{j_1 m_1} \otimes \cdots  \otimes \Tj_{j_N m_N}     ) \label{tensorprodOpTransform} \\
& = \Y_{j_1 m_1}(\theta_1, \phi_1)\: \cdots\: \Y_{j_N m_N}(\theta_N, \phi_N).   \nonumber
\end{align}

\subsubsection{Star product, star commutator, and Moyal equation  \label{subsecStarProdDef}}
In the following, we wish to compute the Wigner representation $W_{AB}$ of 
the product $AB$ of two operators $A$ and $B$ from their respective Wigner representations
$W_A$ and $W_B$. This is accomplished by recalling the defining relation 
%-------------------------
\begin{equation}\label{StarProdEquation}
W_{AB}  = W_A  \star  W_B
\end{equation}
%-------------------------
of the star product $\star$ of two Wigner functions. The star product mimics
the matrix product of two operators. Note that 
the product $AB$ is restricted to the subspace of tensor operators with rank at most $2J$,
just as for the operators $A$ and $B$.
%-------------------------

The time evolution of the density operator $\rho$ is governed by the von-Neumann equation
$i \partial \rho/\partial t = [ \mathcal{H} ,  \rho  ] = 
\mathcal{H}  \rho -  \rho  \mathcal{H}$,
see Eq.~\eqref{NeumannEq}.
This can be mapped to the Wigner representation by applying
Eq.~\eqref{WignerTransform} and exploiting that the symbols $W_{\rho \mathcal{H} }$ 
and $W_{ \mathcal{H} \rho }$ can according to Eq.~\eqref{StarProdEquation} be restated 
in terms of star products. Hence, the equation of motion in the Wigner representation
is given as
%-------------------------
\begin{equation} \label{MoyalEqDefinition}
i \frac{  \partial W_{\rho} }{\partial t} =  [ W_{\mathcal{H}} ,  W_{\rho}  ] _{\star} :=
W_{\mathcal{H}} \star  W_{\rho} -  W_{\rho}   \star  W_{\mathcal{H}}.
\end{equation}
%-------------------------
This defines the star commutator $[\cdot,\cdot]_{\star}$, which constitutes an 
analogue of the matrix commutator.

%--------------------------------------------------------------------------------------------------------
\subsection{Star product for a single spin \texorpdfstring{$1/2$}{1/2}} \label{singlespinstar}
Wigner representations and their defining star products are
well studied in the case of  infinite-dimensional quantum-mechanical operators \cite{schroeck2013,SchleichBook}.
Similarly as in the infinite-dimensional case,
the star product
from Eq.~\eqref{StarProdEquation} can be computed
using an integral or differential form (as discussed in Sections~\ref{prior}
and \ref{coreres1}).
The integral form is an integral transformation of the
Wigner functions $W_A(\theta_1,\phi_1)$ and $W_B(\theta_2,\phi_2)$ which is weighted with a so-called trikernel.
We provide an explicit expression for the integral form in Eq.~\eqref{integralstarprod} 
of \ref{IntegralStarProd} (along the lines of \cite{VGB89})
by evaluating the exact star
product and by applying expansion
formulas from Sec.~\ref{deepformulas} below.
This result provides a formal definition of the star product, but is less useful in applications.

In contrast, the differential star product is more convenient for
explicit calculations since only partial derivatives and the pointwise product of
$W_A$ and $W_B$ are required, which is afterwards followed by a projection.

For an arbitrary spin number $J$,
Klimov and Espinoza \cite{StarProd} state the differential star product 
for two so-called Beresin P symbols $P_A$ and $P_B$
as a sum of the pointwise product of two functions combined with partial derivatives.
The order and number of the partial derivatives grow rapidly with increasing $J$,
and a truncation of the resulting P symbol is required. 
Klimov and Espinoza \cite{StarProd} also provide formulas to compute the
star product of Wigner functions. 
Their method is virtually equivalent to transforming the Wigner functions into
P symbols, then after computing the star product of P symbols, the P symbols are transformed back.
In summary, the star product of two spin-$1/2$ Wigner functions $W_A$ and $W_B$ can be computed
by first decomposing $W_A$ and $W_B$ into spherical harmonics and by reweighting the expansion
coefficients one obtains $\tilde{W}_A$ and $\tilde{W}_B$, where $\tilde{W}_A$ and $\tilde{W}_B$ are
proportional to the P symbols $P_A$ and $P_B$. The differential star product of the P symbols
is then applied to $\tilde{W}_A$ and $\tilde{W}_B$, and one  obtains the four summands 
$a \, \tilde{W}_A \tilde{W}_B$, 
$b \, \{ \tilde{W}_A , \tilde{W}_B \}$, 
$c \, ( \partial \tilde{W}_A / \partial \theta) \, ( \partial \tilde{W}_B / \partial \theta)$,
and
$d \, ( \partial \tilde{W}_A / \partial \phi) \, ( \partial \tilde{W}_B / \partial \phi) /\sin^2{\theta}$,
where $a$, $b$, $c$, and $d$ denote suitable  prefactors.
The resulting function $\tilde{W}_{AB}$ is transformed back by reweighting the terms
in its decomposition into
spherical harmonics. Finally,
a result $W_{AB}$ is obtained
that satisfies the defining property $W_{AB} = W_A \star W_B$.

We consider only the case of $J=1/2$ and provide a simplified approach,
which nevertheless
leads to the same star product and the same equation of motion that is given by the Poisson bracket
(see \cite{klimov2002ExactEvolution} and \cite{VGB89}).
The resulting differential star product $W_A \star W_B$ 
[see Result~\ref{result2} below] is simply a sum of the pointwise product
$W_A W_B$ and the Poisson bracket $\{ W_A  , W_B \} $ of the two Wigner functions.
In Sec.~\ref{deepformulas} we provide formulas necessary to evaluate the star product,
and Sec.~\ref{starprodproof} 
contains the details on how the star product is calculated.
This simplified approach allows us then to extend the star product
to multiple, coupled spins as detailed in Sec.~\ref{app_coupled}.

\subsubsection{Matrix products, pointwise products, and Poisson brackets\label{deepformulas}} 
We detail how to expand products of tensor operators (resp.\ spherical harmonics)
into a linear combination of tensor operators (resp.\ spherical harmonics).
A similar expansion is described for the Poisson bracket of spherical harmonics.
The product of two irreducible  
tensor operators can be expanded as
\cite{StarProd,Varshalovich}
%-------------------------
\begin{equation} \label{ProdTop}
\Tj_{j_1 m_1} \Tj_{j_2 m_2}
= 
\sum_{L=|j_1-j_2|}^{n} 
{}^{J} Q_{j_1 j_2 L} \,
C^{LM}_{j_1 m_1 j_2 m_2}
\Tj_{L M}.
\end{equation}
%-------------------------
Here, the upper bound of the summation 
does not need to exceed $2J$ and is 
given by $n:=\min( j_1{+}j_2, 2J )$
as  $^{J}Q_{j_1 j_2 L}=0$ for 
$L>2J$; note $M=m_1{+}m_2$.\footnote{
	The lower bound in the summation can be enlarged to $\max ( |j_1{-}j_2|, m_1{+}m_2)$ 
	without changing the result.}
Also, $C^{LM}_{j_1 m_1 j_2 m_2}$ are the Clebsch-Gordan coefficients \cite{messiah1962}, 
and the coefficients 
%-------------------------
\begin{equation}
^{J}Q_{j_1 j_2 L}
:=
(-1)^{2 J+L} 
\sqrt{(2 j_1{+}1) (2 j_2{+}1)} \:
\begin{Bmatrix}
j_1 & j_2 & L \\
J & J & J
\end{Bmatrix}
\end{equation}
%-------------------------
are proportional to Wigner $6$-$j$ symbols \cite{messiah1962} and 
depend only on $j_1$, $j_2$, and $L$,
but are independent of $m_1$, $m_2$, and $M$.
This also conforms with the fact that only 
tensor operators of rank zero and one are allowed for the case of a spin $1/2$ in the
product $AB$ (see Sec.~\ref{subsecStarProdDef}).

In order to determine the Poisson bracket of two spherical functions, we first
recall its definition [refer also to Eq.~\eqref{Equation_Motion}]
%-------------------------
\begin{equation}  \label{PBDef}
\{ W_F, W_G \}^{\{ i \}} :=
W_F
\left( \tfrac{\overleftarrow{\partial}}{\partial \phi_i} 
\tfrac{1}{R \sin\theta_i}
\tfrac{\overrightarrow{\partial} }{\partial \theta_i}
-\tfrac{\overleftarrow{\partial} }{\partial \theta_i} 
\tfrac{1}{R \sin\theta_i}     
\tfrac{\overrightarrow{\partial} }{\partial \phi_i}   \right)
W_G,
\end{equation} 
%-------------------------
where the  arrows $\leftarrow$ and $\rightarrow$ 
indicate whether the derivatives act to the left or the right, respectively.
Moreover, the normalization factor is set to  $R=\sqrt{3/(8 \pi)}$.
Based on \cite{fuzzy}, the Poisson bracket
of two spherical harmonics can be expanded as
%-------------------------
\begin{equation}  \label{YPB}
\{ \Y_{j_1 m_1} , \Y_{j_2 m_2} \}
=
\sum_{L=|j_1-j_2|}^{j_1+j_2}
U_{j_1 j_2 L}\,
C^{LM}_{j_1 m_1 j_2 m_2}
\Y_{LM}.
\end{equation} 
%-------------------------
The product of two spherical harmonics decomposes into a linear combination 
as 
(see Sec.~12.9 of \cite{arfken2005mathematical})
\begin{equation} \label{Yprod}
\Y_{j_1 m_1} \Y_{j_2 m_2}
=
\sum_{L=|j_1-j_2|}^{j_1+j_2}
Z_{j_1 j_2 L}\,
C^{LM}_{j_1 m_1 j_2 m_2}
\Y_{LM}.
\end{equation} 
%-------------------------
Here, the coefficients $Z_{j_1 j_2 L}$ and $U_{j_1 j_2 L}$  
depend only on $j_1$, $j_2$, and $L$ and are given by
%-------------------------
\begin{align} 
Z_{j_1 j_2 L}
:= &
\, \sqrt{  \tfrac{(2 j_1{+}1)(2 j_2{+}1)}{4 \pi(2 L{+}1)}     }
C_{ j_1 0 j_2 0}^{L 0},\\
U_{j_1 j_2 L}  := &  
- \tfrac{i}{2 R} [
1 -   (-1  )^{L-j_1-j_2} 
]
\sqrt{j_1(j_1 {+} 1) L (L {+} 1)}
 \times \sqrt{\tfrac{(2 j_1 {+} 1)(2j_2 {+} 1)}{4 \pi (2L {+} 1) }}
C_{ j_1 1\, j_2\, 0}^{L 1}.
\end{align} 
%-------------------------

Although the Poisson bracket $\{\cdot,\cdot\}$ will
not be completely analogous to the star commutator $[\cdot,\cdot]_\star$ for arbitrary $J$, there is a 
strong relation between these two operations. 
Recalling that the star commutator in the Wigner representation corresponds
to the commutator of matrix representations (see Sec.~\ref{subsecStarProdDef}),
we can in analogy compare the Poisson bracket from Eq.~\eqref{YPB} with the usual commutator of tensor operators.
Using Eq.~\eqref{ProdTop}, the commutator
of tensor operators can be brought into a similar form ($M=m_1+m_2$)
%-------------------------
\begin{equation}
[  \Tj_{j_1 m_1} , \Tj_{j_2 m_2} ]
=\hspace{-2mm}
\sum_{L=|j_1-j_2|}^n \hspace{-1mm}
^{J}Q'_{j_1 j_2 L}\,
C^{LM}_{j_1 m_1 j_2 m_2}
\Tj_{L M} \label{comm_decomp}
\end{equation}
%-------------------------
by applying $^{J}Q'_{j_1 j_2 L}:=[  1-(  -1 )^{j_1+j_2 -L}] ^{J}Q_{j_1 j_2 L}$ and the
symmetry properties 
$C^{LM}_{j_1 m_1 j_2 m_2}=(  -1 )^{j_1+j_2 -L} C^{LM}_{j_2 m_2 j_1 m_1}$
of the Clebsch-Gordan coefficients.
We compare Eq.~\eqref{YPB} with Eq.~\eqref{comm_decomp} and note that
the coefficients $Q'^{(J)}_{j_1 j_2 L}$  and $ U_{j_1 j_2 L}$ 
will in general differ. However, their nonzero values 
within the range $j_1,j_2,L \leq 2J$
appear 
at coinciding values of 
$j_1$,  $j_2$, and $L$,
highlighting the 
close relation of the Poisson bracket and 
the star commutator.
Finally, we provide a particular case where this equivalence is strict up to a prefactor:
%-------------------------
\begin{subequations}
	\label{CommutatorPBCorrespondance}
	\begin{align}
	\{ \Y_{10} , \Y_{jm} \} & =  i \sqrt{2} \: m \: \Y_{jm}, 
	& \{ \Y_{1, \pm1} , \Y_{jm} \} & =  \mp i \sqrt{\left( j \mp m \right)\left( j \pm m + 1  \right)} \:    \Y_{j, m \pm1 },
	\label{CommutatorPBCorrespondanceAB}
	\\ 
	[ ^J\Jspace I_z , \Tj_{jm} ] & =  m \: \Tj_{jm},  
	& [ ^J\Jspace I_{\pm} , \Tj_{jm} ] & = \sqrt{\left( j \mp m \right)\left( j \pm m + 1  \right)} \:   
	\Tj_{j,  m \pm1}.\label{CommutatorPBCorrespondanceCD}
	\end{align}
\end{subequations}
The Equations~\eqref{CommutatorPBCorrespondanceCD} 
show how the basic definition of spherical tensor operators relies on 
commutators, c.f.\ Eq.~\eqref{TOpDefALL}.
By comparing them to the Equations~\eqref{CommutatorPBCorrespondanceAB} it is clear that the defining relation
is also satisfied by the Poisson bracket of spherical harmonics 
up to the prefactor $i$ (and an additional prefactor implied by $\mp N_J \wigner(^J\Jspace I_{\pm}) = \Y_{1,\pm1}$ 
and $N_J \sqrt{2} \wigner(^J\Jspace I_{z}) = \Y_{1,0}$).\footnote{
	The prefactor $N_J$ is implied by the formula
	$\Tj_{1, \pm 1}=\mp ^J\Jspace I_{\pm} N_J=\mp ^J\Jspace I_{\pm}/\sqrt{\tr( ^J\Jspace I_{\pm} \, ^J\Jspace I_{\mp}})$
	where $^J\Jspace I_{\pm} = {}^J \Jspace I_x \pm i {}^J \Jspace I_y$. The trace
	is given by 
	$\tr(^J\Jspace I_{\pm} \, ^J\Jspace I_{\mp}) =  \tr[  ({}^J \Jspace I_x \pm i {}^J \Jspace I_y)
	({}^J \Jspace I_x \mp i {}^J \Jspace I_y)   ]
	= \tr[^J\Jspace (I^2)] - \tr[(^J\Jspace I_z)^2]$,
	where  $^J\Jspace (I^2) = (^J \Jspace I_x)^2 +  (^J \Jspace I_y)^2 + (^J \Jspace I_z)^2$,
	$\tr[ ^J\Jspace (I^2)]=J(J{+}1)(2J{+}1)$, and
	$\tr[ (^J\Jspace I_z)^2] = \sum_{m=-J}^J m^2 =J(J{+}1)(2J{+}1)/3$. 
	It follows that $N_J=1/\sqrt{\tr(^J\Jspace \, I_{\pm} ^J\Jspace I_{\mp})}
	=1/\sqrt{2J(J{+}1)(2J{+}1)/3}$.
    }
The particular cases of Eq.~\eqref{CommutatorPBCorrespondance} are also
considered in Equation~(5.13) of \cite{VGB89}.

\newcounter{globprefactor}
\setcounter{globprefactor}{\value{footnote}}

\subsubsection{Evaluation of the star product\label{starprodproof}}

In this section, we detail the explicit form of the differential star product 
for a single spin $1/2$ while ensuring that its form
conforms with Sec.~\ref{subsecStarProdDef}.
We build on the work in  
\cite{StarProd,VGB89}
and provide a simplified approach. The differential star product
is given as the sum of the pointwise product and the Poisson bracket of two spherical functions,
followed by the projection onto rank-one and rank-zero spherical harmonics,
i.e., by truncating spherical harmonics with rank greater than one.
Two distinct symbols are used: the exact  star product $\star$ is obtained from
the prestar product $\prestar$  after
truncating certain spherical harmonics.

\begin{result}\label{result2}
	Given  the Wigner functions $W_{F}(\theta,\phi)$ and $W_{G}(\theta,\phi)$ of 
	the operators $F$ and $G$ acting on  a single spin $1/2$,
	the prestar product (i.e., the product that results in the star product after truncation) is defined as
	%-------------------------
	\begin{align}
	W_{F}(\theta,\phi) \prestar W_{G}(\theta,\phi) &:=
	\sqrt{2 \pi} W_{F} W_{G}
	{-}
	\tfrac{i}{2} \{ W_{F}  , W_{G} \} \label{prestar_single}
	\intertext{using the Poisson bracket $\{\cdot,\cdot\}$ from Eq.~\eqref{PBDef}. 
		Note that the factor $\sqrt{2 \pi}=1/W_{\unity}$ is the inverse of the identity Wigner function. 
		The corresponding star product}
	W_{F}(\theta,\phi) \star W_{G}(\theta,\phi) &:=  \Proj[  W_{F}(\theta,\phi) \prestar W_{G}(\theta,\phi)]  \label{StarStatement}
	\end{align}
	is obtained by projecting onto spherical functions 
	of rank zero or one,
	e.g., by applying the projection operator
	\begin{equation}
	\label{projectionoperatordef} 
	\Proj \, \Y_{j m} := (1 + \mathcal{L}^2/12 - \mathcal{L}^4/24) \Y_{j m} = 
	\begin{cases}
	\Y_{j m} \quad \textup{for} \quad j < 2\\
	0 \quad \textup{for} \quad j = 2
	\end{cases}
	\end{equation}
	which uses the
	angular momentum operator $\mathcal{L}$ with
	eigenvalues $\mathcal{L}^2 \, \Y_{j m} = j(j{+}1) \Y_{j m}$.\footnote{The projector
			$\Proj f(\theta, \phi) =  \sum_{j=0}^1  \sum_{m=-j}^j  \Y_{j m} (\theta, \phi)  
			\int_{\theta=0}^{\pi} \int_{\phi=0}^{2 \pi}  f (\theta, \phi)  
			\Y_{j m}^*(\theta, \phi) 
			\: \sin{\theta} \,\mathrm{d}\theta  \, \mathrm{d}\phi$ can be applied to 
			an arbitrary
			spherical function,
			but Equation~\eqref{StarStatement} is fulfilled by the differential
			operator in Equation~\eqref{projectionoperatordef}.
		}
\end{result}
%%-------------------------
We will now verify that this definition satisfies the defining property
$W_{FG}(\theta,\phi) = W_{F}(\theta,\phi) \star W_{G}(\theta,\phi)$
of a star product [see Eq.~\eqref{StarProdEquation}]. We start by expanding the operators
$F  = \sum_{j=0}^{1}  \sum_{m=-j}^{j} f_{j m}     \T_{j m}$ and
$G  = \sum_{j=0}^{1}  \sum_{m=-j}^{j} g_{j m}     \T_{j m}$ into
tensor operators $\T_{j m}$. One directly 
obtains that
\begin{equation} \label{Operatorprodeval}
FG   
= \sum_{j_1,j_2=0}^{1}  \sum_{m_1=-j_1}^{j_1}  \sum_{m_2=-j_2}^{j_2}   
f_{j_1 m_1}  g_{j_2 m_2}   \:   \T_{j_1 m_1} \T_{j_2 m_2}.
\end{equation}
This summation involves products $\T_{j_1 m_1}  \T_{j_2 m_2}$ of tensor operators
which can be rewritten following Eq.~\eqref{ProdTop} as 
%-------------------------
\begin{equation} \label{ProdTopspinhlaf}
\T_{j_1 m_1} \T_{j_2 m_2}
=
\sum_{L=|j_1-j_2|}^{n} \,
Q^{(1/2)}_{j_1 j_2 L}
\,
C^{LM}_{j_1 m_1 j_2 m_2}
\T_{L M},
\end{equation}
where $n$ can be limited to $n=\min( j_1{+}j_2, 2J )$ and $M=m_1{+}m_2$.
%-------------------------
In order to compare Eqs.~\eqref{Operatorprodeval} and \eqref{ProdTopspinhlaf} with their
counterparts in the Wigner space, we also compute the star product
$W_{F}(\theta,\phi) \star  W_{G}(\theta,\phi)$. Recall from Eq.~\eqref{TensorOpSphHarm}
that the Wigner representations of $F$ and $G$ are given by 
$W_{F}  = \sum_{j=0}^{1}  \sum_{m=-j}^{j} f_{j m}   \Y_{j m}$ and
$W_{G}  = \sum_{j=0}^{1}  \sum_{m=-j}^{j} g_{j m}    \Y_{j m}$.
The prestar product (i.e., the product that results in the star product after truncation)
evaluates to
%-------------------------
\begin{equation}
W_{F}(\theta,\phi) \prestar  W_{G}(\theta,\phi) =  
\sum_{j_1, j_2=0}^{1}  
\sum_{m_1=-j_1}^{j_1}    \sum_{m_2=-j_2}^{j_2}
f_{j_1 m_1}  g_{j_2 m_2} 
\,\Y_{j_1 m_1} \prestar  \Y_{j_2 m_2}, \label{StarProdEval}
\end{equation} 
%-------------------------
where the explicit formula of Eq.~\eqref{prestar_single} results in
\begin{equation}
\Y_{j_1 m_1} \prestar  \Y_{j_2 m_2} =
\sqrt{2\pi}\, \Y_{j_1 m_1} \Y_{j_2 m_2}
-
\tfrac{i}{2}\, \{ \Y_{j_1 m_1} , \Y_{j_2 m_2} \} 
=
\sum_{L=|j_1-j_2|}^{j_1+j_2}
{\Lambda_{j_1 j_2 L}}
C^{LM}_{j_1 m_1 j_2 m_2}
\Y_{LM}.
\label{StarProdSph}
\end{equation} 
Here, we have applied the formulas in Eqs.~\eqref{YPB} and  \eqref{Yprod}
and use the notation $\Lambda_{j_1 j_2 L}:=\sqrt{2\pi} Z_{j_1 j_2 L}
-
(i/2)  U_{j_1 j_2 L}$. The corresponding star product $\star$ is obtained
if we substitute the upper summation bound in Eq.~\eqref{StarProdSph}
with $n=\min( j_1{+}j_2, 2J )$, which is the same bound as in Eq.~\eqref{ProdTopspinhlaf}.
We are now ready to compare the tensor operators in 
Eqs.~\eqref{Operatorprodeval} and \eqref{ProdTopspinhlaf}
with their respective complements in the Wigner space
in Eqs.~\eqref{StarProdEval} and \eqref{StarProdSph}.
Consequently, we have to compare the explicit values of the coefficients
$Q^{(1/2)}_{j_1 j_2 L}$ and $\Lambda_{j_1 j_2 L}$ and we obtain that
%-------------------------
\begin{align*}
& Q^{(1/2)}_{0 0 0} = Q^{(1/2)}_{0 1 1} = Q^{(1/2)}_{1 0 1}
=\Lambda_{0 0 0}  =\Lambda_{0 1 1}   = \Lambda_{1 0 1} = \tfrac{1}{\sqrt{2}},
\\
& Q^{(1/2)}_{1 1 0} =\Lambda_{1 1 0}  = -{\tfrac{\sqrt{3}}{\sqrt{2}}},\; 
Q^{(1/2)}_{1 1 1} =\Lambda_{1 1 1}  = -1,
\end{align*}
%-------------------------
and all other values are zero. This verifies that 
$\Y_{j_1 m_1} \star \Y_{j_2 m_2}=
\Proj (  \Y_{j_1 m_1} \prestar \Y_{j_2 m_2}) = \wigner (   \T_{j_1 m_1} \T_{j_2 m_2})$,
and one respectively concludes that $W_{FG}(\theta,\phi) = W_{F}(\theta,\phi) \star W_{G}(\theta,\phi)$.
The preceding discussion is summarized as 
%-------------------------
\begin{theorem} \label{StarProdLemma}
	The explicit form of the star product $\star$ in Eq.~\eqref{StarStatement}
	observes its defining property from Eq.~\eqref{StarProdEquation},
	i.e.\  $W_{FG}(\theta,\phi) = W_{F}(\theta,\phi) \star W_{G}(\theta,\phi)$.
\end{theorem}
%-------------------------
The explicit values for the prestar product $\prestar$  
(i.e., the product that results in the star product after truncation)
for Wigner functions
forming a basis for a single spin $1/2$ are presented in Table~\ref{Table_StarProds}. 
The corresponding star product $\star$ is obtained by truncating the underlined parts 
in Table~\ref{Table_StarProds}.

\begin{table}[tb]
	\centering 
	\caption{ 
		Prestar product $\prestar$ (i.e., the product that results in the star product after truncation)
		of spherical harmonics for a single spin $1/2$. The corresponding
		star product $\star$ is obtained by discarding 
		the underlined components $\Y_{jm}$ with $j>1$.
		\label{Table_StarProds}}
	\begin{tabular}{@{\hspace{2mm}}l
			@{\hspace{5mm}}l@{\hspace{5mm}}l@{\hspace{5mm}}l@{\hspace{5mm}}l@{\hspace{2mm}}}
		\\[-2mm]
		\hline\hline
		\\[-2mm]
		$ \downarrow \prestar  \rightarrow$         &    $\Y_{00}$   &     $\Y_{1,-1}$    &    $\Y_{10}$     
		&  $\Y_{11}$ 
		\\[1mm] 
		\hline 
		\\[-2mm]
		%------------------------------------------------------------------------------------------------------------------------------------
		$\Y_{00}$     
		& $\tfrac{1}{\sqrt{2}}\Y_{00}$     
		&   $\tfrac{1}{\sqrt{2}}\Y_{1,-1}$    
		&   $\tfrac{1}{\sqrt{2}}\Y_{10}$  
		&    $\tfrac{1}{\sqrt{2}}\Y_{11}$
		\\[3mm] 
		%------------------------------------------------------------------------------------------------------------------------------------
		$\Y_{1,-1}$       
		&   $\tfrac{1}{\sqrt{2}}\Y_{1,-1}$      
		&   $0\underline{{+}\sqrt{\tfrac{3}{5}}\Y_{2,-2}}$ 
		&     $\tfrac{1}{\sqrt{2}}\Y_{1,-1} \underline{{+} \sqrt{\tfrac{3}{10}}\Y_{2,-1}}$
		& $-\tfrac{1}{\sqrt{2}}\Y_{00}{+}\tfrac{1}{\sqrt{2}}\Y_{10}$
		\\ &&&& \hspace{11.2mm} $\underline{{+}\tfrac{1}{\sqrt{10}}\Y_{20}}$ 
		\\[3mm] 
		%------------------------------------------------------------------------------------------------------------------------------------
		$\Y_{10}$         
		&  $\tfrac{1}{\sqrt{2}}\Y_{10}$    
		&   $-\tfrac{1}{\sqrt{2}}\Y_{1,-1}\underline{{+}\sqrt{\tfrac{3}{10}}\Y_{2,-1}}$ 
		&  $\tfrac{1}{\sqrt{2}}\Y_{00}\underline{{+}\sqrt{\tfrac{2}{5}}\Y_{20}}$  
		&    $\tfrac{1}{\sqrt{2}}\Y_{11}\underline{{+}\sqrt{\tfrac{3}{10}}\Y_{21}}$
		\\[3mm] 
		%------------------------------------------------------------------------------------------------------------------------------------
		$\Y_{11}$    
		&  $\frac{1}{\sqrt{2}}\Y_{11}$  
		&  $ -\tfrac{1}{\sqrt{2}}\Y_{00}{-}\tfrac{1}{\sqrt{2}}\Y_{10}$
		&   $-\tfrac{1}{\sqrt{2}}\Y_{11}\underline{{+}\sqrt{\tfrac{3}{10}}\Y_{21}}$  
		&   $0\underline{{+}\sqrt{\tfrac{3}{5}}\Y_{22}}$ 
		\\ && \hspace{11.2mm} $\underline{{+}\tfrac{1}{\sqrt{10}}\Y_{20}}$   
		\\[3mm] \hline \hline
	\end{tabular} 
\end{table}
%-------------------------

\subsubsection{Equation of motion based on the star product}
We can now apply Result~\ref{result2} to the differential equation 
given in Eq.~\eqref{MoyalEqDefinition} in order to provide its explicit form
for a single spin $1/2$. This results in the time evolution \cite{VGB89,klimov2002ExactEvolution}
%-------------------------
\begin{equation}\label{NeumannWignerEvolution}
\partial W_{A}/\partial t = 
\{ W_{A}  , W_{\mathcal{H}}  \}
\end{equation}
%-------------------------
in the Wigner space which
is governed by the Poisson bracket.
A detailed visualization of the whole structure of our formalism leading to 
Eq.~\eqref{NeumannWignerEvolution} is given
in Figure~\ref{GraphicalStarProduct} of
\ref{AppendixStarProd}.
It can be inferred from Table~\ref{Table_StarProds} 
(as $\Y_{jm}$ with $j>1$ in the decomposition are 
symmetric with respect to the order of multiplication) that 
$\{ W_A,W_B\}= i ( W_A {\star} W_B - W_B {\star} W_A  ) =
i ( W_A {\prestar} W_B - W_B {\prestar} W_A  )$ holds for
a single spin $1/2$. This means that a
truncation is not required to compute the time evolution in this 
particular case. In the case of Hamiltonians that contain only
$^J\Jspace I_x$, $^J\Jspace I_y$ and $^J\Jspace I_z$ in the form
$^J \mathcal{H}=c_x \, ^J\Jspace I_x + c_y \, ^J \Jspace I_y + c_z \,^J \Jspace I_z$
Eq.~\eqref{NeumannWignerEvolution}
holds for arbitrary spin $J$ [up to a global prefactor $N_J$
as implied by Eq.~\eqref{CommutatorPBCorrespondance}],\footnoteref{\theglobprefactor}
and agrees with the results of  \cite{VGB89,klimov2002ExactEvolution}.
Consequently, the time 
differential 
for an arbitrary spin $J$ evolving under a linear Hamiltonian
is given as
%-------------------------
\begin{equation} \label{evolutionArbitraryJ}
\partial W_{^J\Jspace A}/\partial t = 
N_J \{ W_{^J\Jspace A}  , W_{^J \mathcal{H}}  \}, 
\end{equation}
%-------------------------
where $N_J$ is a global prefactor\footnoteref{\theglobprefactor}
and $^J\Jspace A$ denotes an arbitrary spin-$J$ operator.

%--------------------------------------------------------------------------------------------------------
\subsection{Star product for multiple coupled spins with spin number
\texorpdfstring{$J=1/2$}{J=1/2} \label{app_coupled}}
We extend the star product from Result~\ref{result2} to multiple coupled spins.
To this end, we introduce a projection operator 
$\FullProj:=\prod_{k=1}^N   
\Projk$
which restricts resulting spherical harmonics 
to rank zero and one\footnote{
	In general, an arbitrary, multivariate spherical function
	$f=f(\theta_1, \phi_1, \ldots, \theta_N, \phi_N )$
	is projected using
$\Projk   f = \sum_{j_k=0}^1  \sum_{m_k=-j_k}^{j_k}
\Y_{j_k m_k}(\theta_k, \phi_k) 
\int_{\theta_k=0}^{\pi} \int_{\phi_k=0}^{2 \pi} 
f\;
\Y^*_{j_k m_k}(\theta_k, \phi_k) 
\sin{\theta_k}\, \mathrm{d}\phi_k \, \mathrm{d}\theta_k$.}
and which can be via Equation~\eqref{projectionoperatordef}
written as
%-------------------------
\begin{equation}
\label{projectiondefinition}
	\Projk := (1 + (\mathcal{L}^{\{k\}})^2/12 - (\mathcal{L}^{\{k\}})^4/24).
\end{equation}
%-------------------------
The angular momentum operator $\mathcal{L}$
acts as $(\mathcal{L}^{\{k\}})^2 \, \Y_{j m}(\theta_k, \phi_k)$
and has the
eigenvalues $j(j{+}1) \Y_{j m}(\theta_k, \phi_k)$.
This projection will be used to truncate superfluous terms in the following definition of the star product:
%-------------------------
\begin{result}\label{result3}
	The prestar product  (the product that results in the star product after truncation)
	of two Wigner functions $W_A$ and $W_B$ 
	corresponding to operators $A$ and $B$ in a 
	system of $N$ coupled spins $1/2$  is defined 
	as 
	\begin{equation*}
	W_A \prestar W_B :=    W_A  ( \prod_{k=1}^{N}  \prestar^{ \{k \} } ) W_B,
	\end{equation*}
	where the individual prestar operators are given by  $\prestar^{ \{k \}} :=  \sqrt{2 \pi} 
	- i\{ \cdot  , \cdot  \}^{\{ k \}}/2  $ (cf.\ Result~\ref{result2},
	Eq.~\eqref{prestar_single}) and 
	$\{ \cdot , \cdot  \}^{\{ k \}}$ denotes the Poisson bracket taken with respect to the
	variables $\theta_k$ and $\phi_k$, see Eq.~\eqref{PBDef}. The star product 
	\begin{equation}
	W_A \star W_B :=   \FullProj ( W_A \prestar W_B ) \label{multispinstarproddefinition} 
	\phantom{:}=
	\FullProj (  W_A   [   \prod_{k=1}^{N}  (  \sqrt{2 \pi}
	- \tfrac{i}{2}   \{ \cdot  , \cdot  \}^{\{ k \}}  )  ]  W_B  ).
	\end{equation}
	is obtained by applying 
	the projection operator $ \FullProj$.
\end{result}
The star product 
for coupled spins in Result~\ref{result3}
allows us to establish the form of the Wigner representation for multispin product operators
$\T^{\{1\}}_{j_1 m_1} \cdots \T^{\{N\}}_{j_N m_N}$, which consists of matrix products
of single-spin operators, cf.\ Table~\ref{tensordef}. In the Wigner representation,
matrix products are substituted by star products:
%-------------------------
\begin{lemma} \label{ProductOperatorTransform}
	The Wigner representation of product 
	operators $\T^{\{1\}}_{j_1 m_1} \cdots \T^{\{N\}}_{j_N m_N}$ 
	is given by the prestar products  (i.e., the product that results in the star product after truncation)
	of the Wigner representations
	$ \wigner (  \T^{\{k\}}_{j_k m_k} )$
	of the individual single-spin operators, i.e.,
	%-------------------------
	\begin{equation}\label{Eq_ProductOperatorTransform}
	\wigner (   \T^{\{1\}}_{j_1 m_1}  \cdots\, \T^{\{N\}}_{j_N m_N}  ) = 
	\wigner (  \T^{\{1\}}_{j_1 m_1} ) \star \cdots \star \wigner (  \T^{\{N\}}_{j_N m_N} ) 
	= \wigner (  \T^{\{1\}}_{j_1 m_1} ) \prestar \cdots \prestar \wigner (  \T^{\{N\}}_{j_N m_N} ).
	\end{equation}
	%-------------------------
\end{lemma}
%-------------------------
\begin{proof}
	From Eq.~\eqref{NSpinSpH}, we know that $\wigner (  \T^{\{k\}}_{j_k m_k} ) = \Y^{\{k\}}_{j_k m_k}$.
	All the Poisson brackets in the star product vanish as their arguments 
	operate on different spins. Therefore, the right hand side of Eq.~\eqref{Eq_ProductOperatorTransform}
	is equal to [c.f. Eq.~\eqref{multispinstarproddefinition}]
	%-------------------------
	\begin{align*}
	& \Y^{\{1\}}_{j_1 m_1} \prestar  \cdots \prestar  \Y^{\{N\}}_{j_N m_N} =
	\sqrt{2 \pi}^{N ( N-1 ) } \: \Y^{\{1\}}_{j_1 m_1}  \cdots\, \Y^{\{N\}}_{j_N m_N} \\
	& =
	\Y_{j_1 m_1}(\theta_1, \phi_1)\, \cdots\, \Y_{j_N m_N}(\theta_N, \phi_N)  / \sqrt{2}^{N ( N-1 ) }.
	\end{align*}
	%-------------------------
	Equation~\eqref{Eq_ProductOperatorTransform}  is now a consequence of Eq.~\eqref{tensorprodOpTransform}.
\end{proof}
%-------------------------
As a consequence of Lemma~\ref{ProductOperatorTransform} and the linearity of 
the star product, the Wigner representation of an arbitrary product operator
$A_1 A_2 \cdots A_N$ can be simplified as
\begin{equation*}
\wigner (A_1 A_2 \cdots A_N) = \sqrt{2 \pi}^{N ( N-1 ) } \: \wigner(A_1) \wigner(A_2) \cdots \wigner(A_N),
\end{equation*}
where each linear operator is a linear combination of tensor operators acting on spin $k$
$A_k = \sum_{j_k,m_k} c_{j_km_k} \T^{\{k\}}_{j_k m_k}$.
For example, the Wigner representation of
Cartesian product operators is obtained by substituting $A_k$
with $I_{k\alpha_k}$ for $\alpha_k \in \{ x,y,z \}$.

The product of single-spin operators is computed
via Lemma~\ref{single-spin}(c)
as $\T^{\{ k \}}_{j_1 m_1} \T^{\{ k \}}_{j_2 m_2} = 
(\T_{j_1 m_1} \T_{j_2 m_2})^{\{ k\}} / \sqrt{2}^{N-1}$. 
Also, the star product of Wigner functions can be concisely
stated by applying
the notation for embedded Wigner functions 
$W_{A}^{\{k\}} :=  
\Y_{00}  \cdots \Y_{00}\allowbreak{} W_{A}(\theta_k,\phi_k) \allowbreak{}
\Y_{00}
\cdots   
\Y_{00}
\allowbreak{} =W_A (\theta_k,\phi_k)/\sqrt{4 \pi}^{N-1}$.
This results in the following
%-------------------------
\begin{lemma} \label{EmbSingleSpinStarProduct}
	The star product of the two 
	Wigner functions
	$\wigner (  \T^{\{k\}}_{j_1 m_1} )=\Y^{\{ k \}}_{j_1 m_1}$ and 
	$\wigner (  \T^{\{k\}}_{j_2 m_2} )=\Y^{\{ k \}}_{j_2 m_2} $ is given by
	%-------------------------
	\begin{equation} 
	\Y^{\{ k \}}_{j_1 m_1} \star \Y^{\{ k \}}_{j_2 m_2}  = 
	W_{(\T_{j_1 m_1}\T_{j_2 m_2})} (\theta_k, \phi_k) / \sqrt{8 \pi}^{N-1}
	 =
	W_{(\T_{j_1 m_1}\T_{j_2 m_2})}^{\{ k\}} / \sqrt{2}^{N-1}.
	\label{EmbSinleSpinWIgnerProduct}
	\end{equation}
	%-------------------------
\end{lemma}
%-------------------------
\begin{proof}
	We set $F:=\T_{j_1 m_1}$ and $G:=\T_{j_2 m_2}$, and Lemma~\ref{StarProdLemma} verifies 
	that $W_{FG}(\theta,\phi) = W_{F}(\theta,\phi) \star W_{G}(\theta,\phi)$.
	Applying the definition of the multispin star product form Eq.~\eqref{multispinstarproddefinition}
	of Result~\ref{result3}	to $\Y^{\{ k \}}_{j_1 m_1} \star \Y^{\{ k \}}_{j_2 m_2}$ results in
	%-------------------------
	\begin{equation*}
	\sqrt{2 \pi}^{N-1}  
	\Projk
	[ \tfrac{ \Y_{j_1 m_1}(\theta_k, \phi_k) }{ \sqrt{4 \pi}^{N-1} }  
	\prestar^{\{k \}} \tfrac{ \Y_{j_2 m_2}(\theta_k, \phi_k)}{\sqrt{4 \pi}^{N-1}}  ].
	\end{equation*}
	%-------------------------
	The formula
	$  \Projk   [  \Y_{j_1 m_1}(\theta_k, \phi_k) \prestar^{\{k\}} \Y_{j_2 m_2}(\theta_k, \phi_k)  ]\allowbreak{}
	=W_{FG}(\theta_k, \phi_k)$ from Thm.~\ref{StarProdLemma} or
	Eq.~\eqref{StarStatement}
	concludes the proof.
\end{proof}
%-------------------------
After these preparations, we can prove that the star product given 
in Result~\ref{result3} actually satisfies its defining
property from Eq.~\eqref{StarProdEquation}:
\begin{theorem} \label{NSpinStarProduct}
	In a system of $N$ interacting spins $1/2$, the Wigner representations
	$W_A$ and $W_B$ of two operators $A$ and $B$ satisfy the equation
	$
	\mathcal{W}(  AB ) = W_A \star W_B
	$.
\end{theorem}
%-------------------------
\begin{proof}
	We introduce the abbreviations for the multiple indexes $\vec{j}:=(j_1,\ldots,j_N)$,  $\vec{m}:=(m_1,\ldots,m_N)$, 
	$\vec{j}':=(j_1',\ldots,j_N')$, as well as  $\vec{m}':=(m_1',\ldots,m_N')$.
	The product $AB$ can be expanded as
	\begin{align*}
	\sum_{\vec{j},\vec{m},\vec{j}',\vec{m}'} \hspace{-2mm} a_{\vec{j},\vec{m}} b_{\vec{j}',\vec{m}'}
	2^{N(N-1)} \,  \T^{\{1\}}_{j_1 m_1} \T^{\{1\}}_{j'_1 m'_1}     
	\cdots \T^{\{N\}}_{j_N m_N} \T^{\{N\}}_{j'_N m'_N},
	\end{align*}
	and the matrix 
	product 
	can be 
	independently  evaluated on each individual spin.
	And each product $\T^{\{k\}}_{j_k m_k} \T^{\{k\}}_{j'_k m'_k}$ 
	can be written as
	$(\T_{j_k m_k} \T_{j_k' m_k'})^{\{k\}} / \sqrt{2}^{N-1}$, cf.\ Lemma~\ref{single-spin}(c), and its
	Wigner transformation is given by $W_{(\T_{j_k' m_k'}\T_{j_k' m_k'})}^{\{ k\}} / \sqrt{2}^{N-1}$.
	On the other hand, we get from Lemma~ \ref{ProductOperatorTransform} that
	\begin{align*}
	W_A & = \sum_{\vec{j},\vec{m}} 
	a_{\vec{j},\vec{m}} \sqrt{2}^{N(N-1)}
	\Y^{\{ 1 \}}_{j_1 m_1} \prestar  \cdots \prestar \Y^{\{ N \}}_{j_N m_N},
	\\
	W_B & = \sum_{\vec{j}',\vec{m}'} 
	b_{\vec{j}',\vec{m}'} \sqrt{2}^{N(N-1)}
	\Y^{\{ 1 \}}_{j_1 m_1} \prestar   \cdots \prestar \Y^{\{ N \}}_{j_N m_N},
	\end{align*}
	and the 
	star product is given by
	$W_A \star W_B =  \sum_{\vec{j},\vec{m},\vec{j}',\vec{m}'}  
	a_{\vec{j},\vec{m}} b_{\vec{j}',\vec{m}'} 
	2^{N(N-1)} C$,
	where 
	\begin{align*}
	C & =
	\FullProj ( \Y^{\{ 1 \}}_{j_1 m_1} \prestar \cdots \prestar \Y^{\{ N \}}_{j_N m_N} \prestar  \Y^{\{ 1 \}}_{j'_1 m'_1} 
	\prestar 
	\cdots \prestar \Y^{\{ N \}}_{j'_N m'_N} ) \\
	& = [\Proj^{\{1\}} (  \Y^{\{ 1 \}}_{j_1 m_1} {\prestar} \Y^{\{ 1 \}}_{j'_1 m'_1} )]  \prestar  
	\cdots \prestar [ \Proj^{\{N\}} (   \Y^{\{ N \}}_{j_N m_N}  {\prestar}     \Y^{\{ N \}}_{j'_N m'_N} )  ] \\
	& =W_{(\T_{j_1 m_1}\T_{j_1' m_1'})}^{\{ 1\}} / \sqrt{2}^{N-1}  \prestar  
	\cdots \prestar W_{(\T_{j_N m_N}\T_{j_N' m_N'})}^{\{ N\}} / \sqrt{2}^{N-1}.
	\end{align*}
	Note that the second equality holds since 
	two spherical harmonics $\Y^{\{ k \}}_{j_k m_k}$ and $\Y^{\{ \ell \}}_{j_\ell m_\ell}$  
	do star-commute under the assumption that
	$k \neq \ell$, i.e.\ $[\Y^{\{ k \}}_{j_k m_k}, \Y^{\{ \ell \}}_{j_\ell m_\ell}]_{\prestar}=0$.
	The third equality follows from Lemma~\ref{EmbSingleSpinStarProduct}  which shows that
	$\Y^{\{ k \}}_{j_k m_k} \star \Y^{\{ k \}}_{j_k' m_k'}  =
	W_{(\T_{j_k m_k}\T_{j_k' m_k'})}^{\{ k\}} / \sqrt{2}^{N-1}$.
	The proof is now a consequence of 
	Lemma~\ref{ProductOperatorTransform} which verifies that
	$W_{(\T_{j_1 m_1}\T_{j_1' m_1'})}^{\{ 1\}} \prestar \cdots \prestar 
	W_{(\T_{j_N m_N}\T_{j_N' m_N'})}^{\{ N\}} =
	\wigner [(\T_{j_1 m_1}\T_{j_1' m_1'})^{\{ 1\}}  \cdots  (\T_{j_N m_N}\T_{j_N' m_N'})^{\{ N\}} ]$.
\end{proof}

After verifying the correctness of the star product from Result~\ref{result3},
we highlight how the star product governs the time evolution of
an arbitrary number $N$ of coupled spins $1/2$.
We introduce the notations
$a:=\sqrt{2 \pi}$ and $b_k:=- \tfrac{i}{2}   \{ \cdot  , \cdot  \}^{\{ k \}}  $
and 
start by rewriting the star product [see
Eq.~\eqref{multispinstarproddefinition} in Result~\ref{result3}] into a more convenient form
\begin{equation} \label{ArbitraryNStarProduct}
\prestar = \prod_{k=1}^N (a+b_k) = 
\sum_{\ell=0}^N a^{N-\ell} \, [ \sum_{\genfrac{}{}{0pt}{}{k_1, k_2 ,\ldots, k_\ell}{k_\mu \neq k_\nu \text{ if } \mu \neq \nu}}
b_{k_1} b_{k_2} \cdots\,  b_{k_\ell}   ],
\end{equation}
where $k_\mu \in \{1 ,\ldots, N \}$. 
The first four terms in the sum 
of Eq.~\eqref{ArbitraryNStarProduct} are given by 
\begin{align*}
\prestar  =& a^N + a^{N-1} (b_1 + b_2 +\cdots + b_N) + a^{N-2} (b_1 b_2 + \cdots + b_{N-1} b_N)  \\
+&  a^{N-3} (b_1 b_2 b_3 +  \cdots + b_{N-2} b_{N-1} b_{N}) + \cdots,
\end{align*}
and there are in total $\sum_{\ell=0}^N \binom{N}{\ell} =  2^N$ terms.
The star product is then obtained by applying the projector $\FullProj$
from Eq.~\eqref{projectiondefinition}.
Result~\ref{result3} 
now determines the equation of motion 
via the star commutator from  Eq.~\eqref{MoyalEqDefinition}
while the terms with even indices $\ell$ in
Eq.~\eqref{ArbitraryNStarProduct} cancel each other out.
\begin{result} \label{result4}
The equation of motion in a system of $N$ coupled spins $1/2$
is given by
\begin{equation} 
\label{fullequationofmotion}
i\frac{  \partial W_{\rho}  }{\partial t} =  2 (W_\unity)^{-1} \, \FullProj [ W_{\mathcal{H}}
\sum_{\genfrac{}{}{0pt}{}{\ell=1}{\ell \text{ odd}}}^N c^{\ell} 
  \sum_{\genfrac{}{}{0pt}{}{k_1, k_2 ,\ldots, k_\ell}{k_\mu \neq k_\nu \text{ if } \mu \neq \nu}} 
p_{k_1} p_{k_2} \cdots\,  p_{k_\ell}   \,   W_{\rho}  ] ,
\end{equation}
where $c:=-i/\sqrt{8\pi}$, $W_\unity =1/ \sqrt{2\pi}^{N}$ is the Wigner transform of the identity operator (see Table~\ref{CartesianWignerRepr}), $\FullProj$ denotes the projection from Eq.~\eqref{projectiondefinition},
and $p_{k_\mu}:= \{ \cdot  , \cdot  \}^{\{ k_\mu \}}$ is the Poisson bracket from
Eq.~\eqref{PBDef}. The first two terms in the expansion are 
\begin{align*} 
i\frac{  \partial W_{\rho} }{\partial t}  &= 2 (W_\unity)^{-1} \, \FullProj 
[ c \,   W_{\mathcal{H}}  (p_1 + p_2 +\cdots + p_N)     W_{\rho}\\ 
&+c^3  \, W_{\mathcal{H}}  (p_1 p_2 p_3 +  \cdots + p_{N-2} p_{N-1} p_{N})   W_{\rho} + \cdots ].
\end{align*}
\end{result}
The first term in this expansion is given as a sum $W_{\mathcal{H}}  (p_1 + p_2 +\cdots + p_N)     W_{\rho}$
of Poisson brackets,
which corresponds to a classical evolution of a phase-space
probability distribution $W_{\rho}$.
This truncated version of the expansion could be used to study
the evolution of spin-$1/2$ systems in a semi-classical approximation.
And the first-order approximation $- (W_\unity)^{-1}/\sqrt{2 \pi} \, \FullProj 
[ \,   W_{\mathcal{H}}  (p_1 + p_2 +\cdots + p_N)     W_{\rho}]$
to the time derivative 
corresponds to the classical equation of motion, and the number of terms
(i.e.\  the number of  Poisson brackets $p_k$)
scales linearly with the number $N$ of degrees of freedom.
The complete, exact equation of motion of a spin-$1/2$ system is then established by introducing
quantum corrections as a power series of odd powers in $c$, similar as in
the infinite-dimensional case.
The number of these quantum corrections grows exponentially for increasing number
of coupled spins.
Consequently,
the equation of motion is a sum of those terms that contain odd number of products of Poisson brackets.
The contribution of each term $p_{k_1} p_{k_2} \cdots\, p_{k_\ell}$ shrinks exponentially for 
increasing $N$ as the number of Poisson 
brackets grows.

%--------------------------------------------------------------------------------------------------------
\subsection{Results for multiple coupled spins \texorpdfstring{$1/2$}{1/2} \label{summaryofresults}}
The Wigner formalism for an arbitrary number of coupled spins $1/2$ is completely determined
by the previous sections:
the star product and the equation of motion 
are given in Results~\ref{result3} and \ref{result4},
respectively.
In the following, these results are summarized and simplified for the special cases of two
and three coupled spins $1/2$, as these cases are  important for applications.

\subsubsection{Two coupled spins\label{TwoSpinSummary}}
The star product from Result~\ref{result3} is now detailed 
in a convenient formula
for the case of
two coupled spins $1/2$:
\begin{corollary}\label{twospinstarprodcorollary}
In case of two coupled spins $1/2$,  we obtain 
the prestar product as
%-------------------------
\begin{align} \label{twospinstarproductequation}
\prestar &=
(  \sqrt{2 \pi}
{-} \tfrac{i}{2}   \{ \cdot  , \cdot  \}^{\{ 1 \}}  )
(  \sqrt{2 \pi}
{-} \tfrac{i}{2}   \{ \cdot  , \cdot  \}^{\{ 2 \}}  ) 
\\ &= 
2 \pi - i \sqrt{\tfrac{\pi}{2}} (  \{ \cdot  , \cdot  \}^{\{ 1 \}}  
{+}  \{ \cdot  , \cdot \}^{\{ 2 \}}  ) -  \{ \cdot  , \cdot  \}^{\{ 1 \}}   \{\cdot  , \cdot  \}^{\{ 2 \}}/4.
\nonumber
\end{align}
%-------------------------
The star product $W_A (\theta_1, \phi_1, \theta_2, \phi_2 ) \star W_B (\theta_1, \phi_1, \theta_2, \phi_2 )$ 
of two Wigner functions can be consequently computed as 
$ \Proj ^{\{ 1,2 \}}  
[W_A (\theta_1, \phi_1, \theta_2, \phi_2 ) \prestar W_B (\theta_1, \phi_1, \theta_2, \phi_2 ) ]$,
where the corresponding projections $\Proj ^{\{ 1,2 \}} = \Proj ^{\{ 2 \}}\Proj ^{\{ 1 \}}$  act on two spheres
by projecting onto rank-one and rank-zero spherical harmonics;
refer to the definition of $\Proj ^{\{ k \}}$ in Eq.~\eqref{projectiondefinition}.
\end{corollary}
Table~\ref{tensordef} implies that tensor operators acting on single spins
are expressed as $\Ta_{j_1 m_1}=\T_{j_1 m_1} \otimes \T_{0,0}$ 
and $\Tb_{j_2 m_2}= \T_{0,0}  \otimes \T_{j_2 m_2}$, and their  
Wigner transformations
from Result~\ref{result1}
are $\wigner ( \T_{j m}^{\{ 1 \}}  ) = \Y_{j m}^{\{ 1 \}}=\Y_{j m}(\theta_1, \phi_1)/ \sqrt{4 \pi}$
and $\wigner ( \T_{j m}^{\{ 2 \}}  ) = \Y_{j m}^{\{ 2 \}}=\Y_{j m}(\theta_2, \phi_2)/ \sqrt{4 \pi}$.
Similarly, one obtains the form $\wigner ( 2 \T_{j_1 m_1}^{\{ 1 \}}  \T_{j_2 m_2}^{\{ 2 \}}   ) 
= 
\Y_{j_1 m_1}(\theta_1, \phi_1) \Y_{j_2 m_2}(\theta_2, \phi_2)$ of the Wigner representation 
for bilinear operators, cf.\ Result~\ref{result1}.
The star commutator 
%-------------------------
\begin{equation*}
[W_A, W_B]_\star = W_A  \star  W_B -W_B  \star  W_A = 
- i \sqrt{2 \pi}   \Proj ^{\{ 1,2 \}}   (  \{ W_A  , W_B  \}^{\{ 1 \}}  +  \{ W_A  , W_B  \}^{\{ 2 \}}  ) 
\end{equation*}
%-------------------------
is given by the antisymmetric part of the star product from Result~\ref{result3}, which in the case of two spins $1/2$ results in the 
truncated Poisson bracket over both spheres.
The time evolution of the density matrix $\rho$ under the Hamiltonian $\mathcal{H}$ is proportional 
to the
star commutator (see Result~\ref{result4}):
\begin{corollary}
\label{twospinEQMcorollary}
The equation of motion for two coupled spins $1/2$  is given by
%-------------------------
\begin{equation} \label{TwoSpinTimeEvolution}
\tfrac{\partial W_{\rho}}{\partial t} = 
\sqrt{2 \pi}   \Proj ^{\{ 1,2 \}}   
(  \{ W_{\rho}  , W_{\mathcal{H}}  \}^{\{ 1 \}}  {+}  \{ W_{\rho}   ,  W_{\mathcal{H}}  \}^{\{ 2 \}}  ) .
\end{equation}
%-------------------------
\end{corollary}

\subsubsection{Three coupled spins} \label{threespinsummary}
For three coupled spins, we also obtain the star product by applying
Result~\ref{result3}:
\begin{corollary}
\label{threespinstarprodcorollary}
The prestar product for three coupled spins $1/2$ simplifies to
%-------------------------
\begin{equation}   \label{threespinstarproductequation}
\prestar= (  \sqrt{2 \pi}
{-} \tfrac{i}{2}   \{ \cdot  , \cdot  \}^{\{ 1 \}}  )
(  \sqrt{2 \pi}
{-} \tfrac{i}{2}   \{ \cdot  , \cdot  \}^{\{ 2 \}}  ) (  \sqrt{2 \pi}
{-} \tfrac{i}{2}   \{ \cdot  , \cdot  \}^{\{ 3 \}}  ),
\end{equation}
%-------------------------
and the star product $W_A \star W_B =   \Proj ^{\{ 1,2,3 \}} ( W_A \prestar W_B )$
is obtained by applying the projection $\Proj ^{\{ 1,2,3 \}}=\Proj ^{\{ 3 \}}   \Proj ^{\{ 2 \}}  \Proj ^{\{ 1 \}}$;
refer to the definition of $\Proj ^{\{ k \}}$ in Eq.~\eqref{projectiondefinition}.
\end{corollary}
Normalized linear tensor operators are given as
$\Ta_{j_1 m_1}=\T_{j_1 m_1} \otimes \T_{00}  
\otimes \T_{00}$,  $\Tb_{j_2 m_2}= \T_{00}  \otimes \T_{j_2 m_2} \otimes \T_{00}$
and  $\T ^{\{ 3 \}}_{j_3 m_3}= \T_{00} \otimes \T_{00} \otimes \T_{j_3 m_3}$.
Their Wigner representation from Result~\ref{result1} is 
$\wigner ( \T_{j m}^{\{ 1 \}}  ) \allowbreak{}=
\Y_{j m}^{\{ 1 \}}\allowbreak{}= \allowbreak{} \Y_{j m}(\theta_1, \phi_1)/ 4 \pi$.
In the bilinear case, the Wigner functions 
have the form
\begin{align*}
\wigner ( \sqrt{2}^3  \Ta_{j_1 m_1}  \Tb_{j_2 m_2} ) & = 
\Y_{j_1 m_1}(\theta_1, \phi_1) \Y_{j_2 m_2}(\theta_2, \phi_2)/\sqrt{4\pi},\\
\wigner ( \sqrt{2}^3  \Tb_{j_2 m_2}  \T ^{\{ 3 \}}_{j_3 m_3} ) &=
\Y_{j_2 m_2}(\theta_2, \phi_2) \Y_{j_3 m_3}(\theta_3, \phi_3)/\sqrt{4\pi},\\
\wigner ( \sqrt{2}^3  \Ta_{j_1 m_1}  \T ^{\{ 3 \}}_{j_3 m_3} ) &=
\Y_{j_1 m_1}(\theta_1, \phi_1) \Y_{j_3 m_3}(\theta_3, \phi_3)/\sqrt{4\pi}.
\end{align*}
The correctly normalized trilinear operator
$\sqrt{2}^6   \Ta_{j_1 m_1}  \allowbreak{}  \Tb_{j_2 m_2}  \allowbreak{} \T ^{\{ 3 \}}_{j_3 m_3}$ 
results in the Wigner function
$\Y_{j_1 m_1}(\theta_1, \phi_1)  
\allowbreak{}\Y_{j_2 m_2}(\theta_2, \phi_2)  \allowbreak{} \Y_{j_3 m_3}(\theta_3, \phi_3)$.
The time evolution is determined by the star commutator
%-------------------------
\begin{equation} 
[W_A, W_B]_\star =
-  2 \pi i \Proj ^{\{ 1,2,3 \}} \sum_{k=1}^3  \{ W_A  , W_B  \}^{\{ k \}}
+ \tfrac{i}{4}  \Proj ^{\{ 1,2,3 \}} 
W_A (\{ . , .  \}^{\{ 1 \}}  \{. , . \}^{\{ 2 \}}  \{ . , . \}^{\{ 3 \}}) W_B,   \label{threespinstarcommutator} 
\end{equation}
i.e., the antisymmetric part of the star product from Result~\ref{result3}. Using Result~\ref{result4},
we obtain the equation of motion:
%-------------------------
\begin{corollary}
\label{threespinEQMcorollary}
The equation of motion for three coupled spins $1/2$ is determined as
%-------------------------
\begin{equation} 
\tfrac{\partial W_{\rho}}{\partial t} = 
  2 \pi \Proj ^{\{ 1,2,3 \}} \sum_{k=1}^3  \{ W_{\rho}  , W_{\mathcal{H}}  \}^{\{ k \}}
- \tfrac{1}{4}  \Proj ^{\{ 1,2,3 \}} 
W_{\rho} (\{ . , .  \}^{\{ 1 \}}  \{. , . \}^{\{ 2 \}}  \{ . , . \}^{\{ 3 \}}) W_{\mathcal{H}}.   
\end{equation}
%-------------------------
Here, the triple Poisson bracket $p_1 p_2 p_3$  in Result~\ref{result4}
is the first quantum correction (which vanishes except when acting on
trilinear Wigner functions) and leads to the explicit form
\begin{align*}
&     \{ \cdot  , \cdot  \}^{\{ 1 \}}  \{ \cdot  , \cdot  \}^{\{ 2 \}}  \{ \cdot  , \cdot  \}^{\{ 3 \}}  
=  \tfrac{1}{R^3 \sin \theta_1 \sin \theta_2 \sin \theta_3} 
\\
\times (\, +&\overleftarrow{\partial}^3_{\phi_1,\phi_2,\phi_3}
\overrightarrow{\partial}^3_{\theta_1,\theta_2,\theta_3}
- 
\overleftarrow{\partial}^3_{\phi_2,\phi_3,\theta_1}
\overrightarrow{\partial}^3_{\phi_1,\theta_2,\theta_3} 
- 
\overleftarrow{\partial}^3_{\phi_1,\phi_3,\theta_2}
\overrightarrow{\partial}^3_{\phi_2,\theta_1,\theta_3}
- 
\overleftarrow{\partial}^3_{\phi_1,\phi_2,\theta_3}
\overrightarrow{\partial}^3_{\phi_3,\theta_1,\theta_2}
\\
+ &
\overleftarrow{\partial}^3_{\phi_1,\theta_2,\theta_3}
\overrightarrow{\partial}^3_{\phi_2,\phi_3,\theta_1}
+ 
\overleftarrow{\partial}^3_{\phi_2,\theta_1,\theta_3}
\overrightarrow{\partial}^3_{\phi_1,\phi_3,\theta_2}
+ 
\overleftarrow{\partial}^3_{\phi_3,\theta_1,\theta_2}
\overrightarrow{\partial}^3_{\phi_1,\phi_2,\theta_3} 
-
\overleftarrow{\partial}^3_{\theta_1,\theta_2,\theta_3}
\overrightarrow{\partial}^3_{\phi_1,\phi_2,\phi_3}
\, ),
\end{align*}
where the notation 
${\partial}^3_{\theta_1,\theta_2,\theta_3}= \partial^3 / 
(\partial \theta_1  \partial \theta_2 \partial \theta_3)$ is used
and  the direction of an arrow signifies 
whether the derivative is taken with respect to 
the function on the left or right.
\end{corollary}

%--------------------------------------------------------------------------------------------------------
\subsubsection{Geometrical interpretation of the scalar product 
	of vector operators
	in the Wigner representation \label{VectorOperator}}
Let us consider the following two vector operators in a system of two coupled spins $1/2$ as
$\mathrm{I}_k=(I_{kx},I_{ky},I_{kz})$ 
for $k\in\{1,2\}$. The scalar product of these two operators
yields
\begin{equation}\label{vector}
\mathrm{I}_1 \cdot \mathrm{I}_2 = I_{1x} I_{2x} + I_{1y} I_{2y} + I_{1z} I_{2z}
=\sum_{m=-1}^1 \T_{1m} \otimes \T_{1m}^\dagger/2,
\end{equation}
where the
second equality is given by a decomposition into tensor operators.
Equation~\eqref{vector} can be generalized to
arbitrary $J$.\footnoteref{\theglobprefactor}
Many important coupling Hamiltonians of two angular momenta
can be described in this form  including
the scalar coupling and the spin-orbit coupling.
The Wigner representation directly follows as
\begin{equation} \label{VecOpWF}
\wigner(\mathrm{I}_1 \cdot \mathrm{I}_2) 
=\sum_{m=-1}^1 \Y_{1m}(\theta_1,\phi_1) \Y_{1m}^*(\theta_2,\phi_2)/2.
\end{equation}	
Given the unit vectors $\vec{r}_k$ in $\mathbb{R}^3$ which are parametrized 
in spherical coordinates as
\begin{equation*}
\vec{r}_k :=
(x_k,
y_k,
z_k)^T
=
(
\sin\theta_k \cos\phi_k,
\sin\theta_k \sin\phi_k,
\cos\theta_k)^T,
\end{equation*}
their scalar product is given by $\vec{r}_1 \cdot \vec{r}_2=\cos\gamma$,
where $\gamma$ denotes the angle between the two unit vectors $\vec{r}_1$ and $\vec{r}_2$.
Consequently the addition theorem of spherical harmonics \cite{arfken2005mathematical} results in
\begin{equation} \label{additionTH}
P_j(\cos\gamma)
= \frac{4\pi}{2j+1}  \sum_{m=-j}^j \Y_{jm}(\theta_1,\phi_1) \Y_{jm}^*(\theta_2,\phi_2),
\end{equation}
where $P_j(\alpha)$ is the Legendre polynomial of degree $j$.
Thus, one can rewrite the Wigner function in Eq.~\eqref{VecOpWF} in terms of the angle $\gamma$
as $\wigner(\mathrm{I}_1 \cdot \mathrm{I}_2) = R^2 \cos \gamma $, with $R:=\sqrt{{3}/(8 \pi)}$.

The arguments of the Wigner function of two coupled spins
$W=W(\theta_1,\phi_1,\theta_2,\phi_2)$ 
can also be given in terms of the unit vectors $\vec{r}_1$ and $\vec{r}_2$ 
as $W=W(\vec{r}_1,\vec{r}_2)$, consequently Eq.~\eqref{VecOpWF}
becomes $\wigner(\mathrm{I}_1 \cdot \mathrm{I}_2) = R^2 \, \vec{r}_1 \cdot \vec{r}_2$
by applying Eq.~\eqref{additionTH}.
Expanding this expression results in
\begin{equation*}
\wigner(\mathrm{I}_1 \cdot \mathrm{I}_2)  = R^2 \, (x_1x_2+y_1y_2+z_1z_2) 
= R^2 [ \cos\theta_1\cos\theta_2  + \sin\theta_1\sin\theta_2\cos (\phi_1 {-} \phi_2)  ].
\end{equation*}
%-------------------------
\begin{figure}[b]
	\centering
	\includegraphics{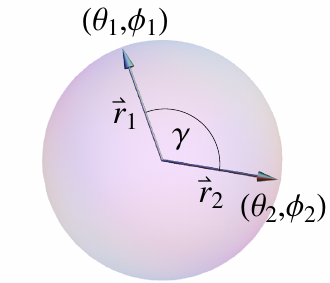}
	\caption{\label{ScalarProdc}(Color online) 
		The Wigner function $W(\theta_1,\phi_1,\theta_2,\phi_2)$ of two spins
		is determined by their arguments which define two points on the surface of the unit sphere. 
		These points correspond to the unit vectors $\vec{r}_1$ and $\vec{r}_2$.
	}
\end{figure}
%-------------------------

In general, the Wigner function of two spins is a complex number $W(\theta_1,\phi_1,\theta_2,\phi_2)$
which depends on 
the arguments $\theta_1$, $\phi_1$, $\theta_2$, and $\phi_2$. These arguments define two points on
a sphere, see Fig.~\ref{ScalarProdc}. The Wigner function
$\tilde{W}(\vec{r}_1,\vec{r}_2):=\wigner(\mathrm{I}_1 \cdot \mathrm{I}_2)$
is now completely determined by the angle $\gamma$ between the two
vectors $\vec{r}_1$ and $\vec{r}_2$. The value of the Wigner function is
given by 
$\tilde{W}(\vec{r}'_1,\vec{r}'_2) = 0$ for the particular choices of 
$\vec{r}'_1=(0,0,1)$ and $\vec{r}'_2 = (1,0,0)$.
And similarly for 
$\vec{r}''_1=({1}/{\sqrt{2}},{1}/{\sqrt{2}},0)$ and $\vec{r}''_2 = (0,{1}/{\sqrt{2}},{1}/{\sqrt{2}})$,
one obtains $\tilde{W}(\vec{r}''_1,\vec{r}''_2) = R^2/2$.

\subsubsection{Spins evolving under a natural Hamiltonian} \label{naturalhamiltonian}

Let us finally consider the case where 
an arbitrary 
number $N$ of coupled spins $1/2$ evolve under a Hamiltonian 
\begin{equation*}
\mathcal{H}=\sum_{k=1}^N \sum_{j, m} a_{j,m,k} \T^{\{k\}}_{j m}  + 
\sum_{k_1 \neq k_2}^N \sum_{\substack{j_1,j_2\\m_1, m_2}} b^{j_2,m_2,k_2}_{j_1,m_1,k_1} 
\T^{\{k_1\}}_{j_1 m_1} \T^{\{k_2\}}_{j_2 m_2}
\end{equation*}
which
contains only linear and bilinear interactions,
i.e., natural interactions of physical systems.
Refer also to Eq.~\eqref{vector} in Sec.~\ref{VectorOperator} for the form of the coupling Hamiltonian.
\begin{corollary} \label{result5}
For natural Hamiltonians consisting only of linear and bilinear terms, 
the time evolution 
of a system of $N$ interacting spins $1/2$ is given by
%-------------------------
\begin{equation} \label{NaturalHamTimeEvolution}
\partial W_\rho / \partial t =
\sqrt{2 \pi}^{N-1} \Proj ^{\{ 1 \dots N \}} \sum_{k=1}^N  \{ W_\rho  , W_\mathcal{H}  \}^{\{ k \}}, 
\end{equation}
%-------------------------
where $W_\rho$  denotes the Wigner function of an arbitrary $N$-spin density matrix $\rho$
and $\FullProj$ is the projection from Eq.~\eqref{projectiondefinition}.
\end{corollary}
Exact time evolution of spin-$1/2$ Wigner functions under natural Hamiltonians
is therefore given by the sum of Poisson brackets, i.e., the classical 
equation of motion for phase-space probability distributions. The only
non-classical term is the projection $\FullProj$
from Eq.~\eqref{projectiondefinition}.

\section{Advanced examples \label{advancedexamples}}
In this section, we consider two advanced examples to convey our approach
of using sums of product operators and 
directly determining the time evolution of quantum systems in Wigner space.
We analyze the case of 
two coupled spins evolving under the CNOT gate (see Sec.~\ref{Ex_CNOT}).
Finally, we present an example 
for the time evolution of three coupled spins $1/2$ (see Sec.~\ref{Ex_three}).

\subsection{CNOT gate\label{Ex_CNOT}}
We continue our discussion of Wigner functions for 
two coupled spins $1/2$ from Sec.~\ref{Ex_two}
and consider the evolution of pure states under
the controlled NOT (CNOT) gate \cite{NC00}. Section~\ref{CNOT_evol} starts
with the computation of the time evolution in the Wigner frame.
In Sec.~\ref{CNOT_entangling}, 
we analyze the 	creation of entanglement using Wigner functions
and their pictorial representations.

%-------------------------
\begin{figure}[tb]
	\centering
	\includegraphics{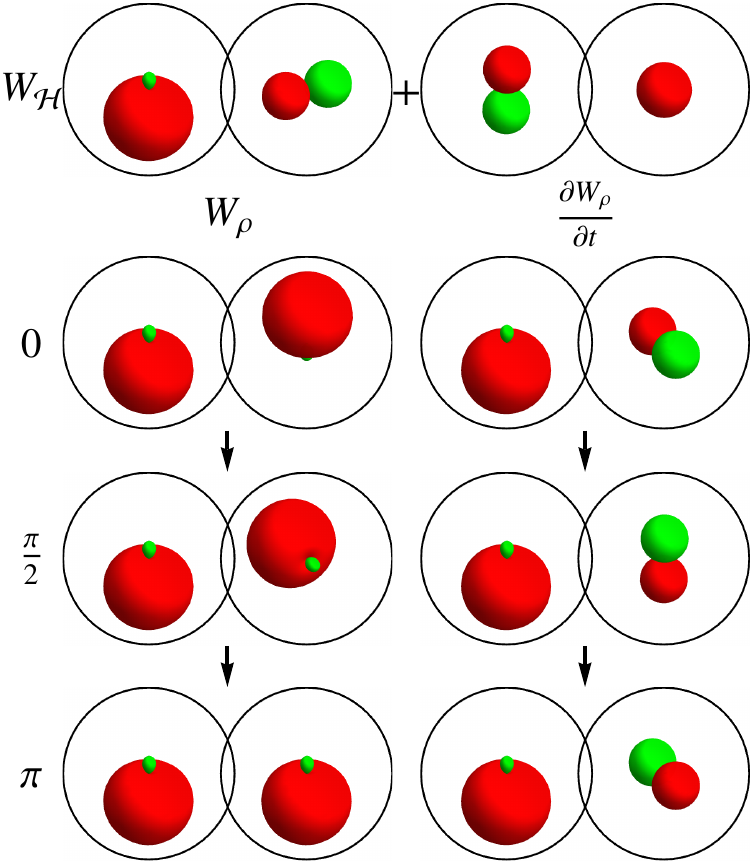}
	\caption{\label{Figure4}(Color online) Evolution of the density matrix of a pure 
		quantum state $\ket{\psi(0)}=\ket{\beta \alpha}$
		under the CNOT gate, implemented
		by the Hamiltonian 
		$\mathcal{H}= \w [I_{1\beta} I_{2x} + I_{1z}/2 ]$.
		PROPS representations of the Wigner functions
		$W_{\rho}$ and $\partial W_{\rho} / \partial t $
		are shown for the times $t=0$, $\w t = \pi/2$, and $\w t = \pi$.
		The control spin is set to $\ket{\beta}$, 
		and the second spin flips, i.e., 
		$\ket{\psi(\pi/ \w)}=\ket{\beta \beta}$.\footnoteref{\thecolors}
	}
\end{figure}
%-------------------------

%-------------------------
\begin{figure}[tb]
	\centering
	\includegraphics{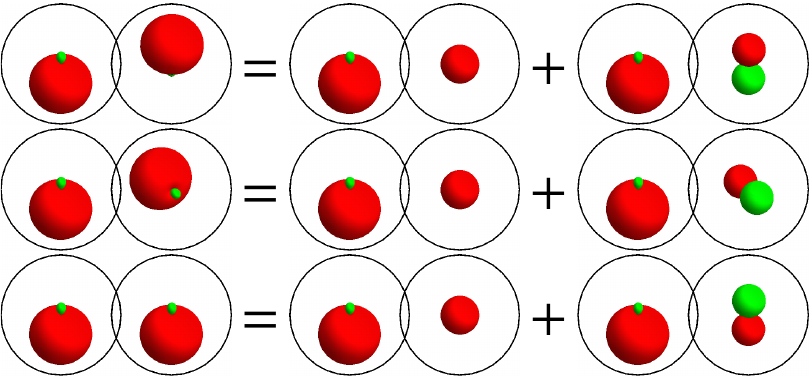}
	\caption{\label{Figure5}(Color online) The Wigner functions $W_\rho(t)$ from Fig.~\ref{Figure4} 
		are decomposed for $\w t \in \{0,\pi/2,\pi\}$ into 
		$1/(4\pi)$ and $\cos (\w t)\, (\lambda \cos\theta_2) - \sin (\w t)\, (\lambda \sin\theta_2 \sin\phi_2)$
		via
		Eqs.~\eqref{CNOTDensityOp}-\eqref{CNOTWignerF}.\footnoteref{\thecolors}
	}
\end{figure}
%-------------------------

\subsubsection{Evolution under the CNOT gate\label{CNOT_evol}}
In the following, we consider the time evolution of pure spin-1/2
states. Let us first introduce the notation $|\alpha \rangle := (1,0)^T$ and 
$|\beta \rangle := (0,1)^T$ 
(cf.\ p.~308 in \cite{bransden2000}, p.~126 in \cite{rae2008}, or
p.~3 in \cite{carrington1967}), which is very similar
to the notation $|0 \rangle$ and $|1 \rangle$ often used in quantum mechanics and
quantum information theory \cite{NC00}, but avoids confusion with
different conventions in the literature relating $|0 \rangle$
to either the excited or ground state.
A pure initial state $\ket{\psi(0)}:=\ket{\beta}\otimes\ket{\alpha}\equiv \ket{\beta \alpha}$
is evolving under the effective Hamiltonian
\begin{equation} \label{CNOTHam}
\mathcal{H}= \w ( I_{1\beta} I_{2x} + I_{1z}/2 )= \w[ (\unity_4/2{-}I_{1z}) I_{2x} + I_{1z}/2 ],
\end{equation} 	
where 
$I_{1\beta}:=I_{\beta} \otimes \unity_2$ and $I_\beta := \unity_2/2-I_z =\ket{\beta}\bra{\beta} $
projects onto the pure state $\ket{\beta}$; likewise 
$I_\alpha := \unity_2/2+I_z =\ket{\alpha}\bra{\alpha}$.
Exponentiation
of $-it\mathcal{H}$ 
leads to the unitary
%-------------------------
\begin{equation} \label{unitary}
U_t = \exp{(-i\mathcal{H}t)}
=
\xi(t)
\begin{pmatrix}
1 & 0 & 0 & 0\\
0 & 1 & 0 & 0\\
0 & 0 & \tfrac{1{+}e^{i\w t}}{2} & \tfrac{1{-}e^{i\w t}}{2}\\
0 & 0 & \tfrac{1{-}e^{i\w t}}{2} & \tfrac{1{+}e^{i\w t}}{2}
\end{pmatrix}
\end{equation}
%-------------------------
of determinant one with $\xi(t)=\exp(-i \w t/4)$, and  $U_{T}$ with $T=\pi/\w$ is the CNOT gate.
The initial state $\ket{\psi(0)}$ evolves into
\begin{equation} \label{psisolution}
\ket{\psi(t)}=U_t\ket{\psi(0)}=\tfrac{\xi(t)}{2}[ (1{+}e^{i\w t}) \ket{\beta\alpha}  {+}   (1{-}e^{i\w t}) \ket{\beta\beta}  ],
\end{equation}
where $\ket{\psi(T)} \propto \ket{\beta\beta}$. In preparation to switch to Wigner functions,
Eq.~\eqref{psisolution} is rewritten in its density-matrix form 
\begin{gather} 
\rho(t)  = \ket{\psi(t)}\bra{\psi(t)} = \rho_A\, \rho_B(t) \text{ where} \label{CNOTdensityop} \\
\rho_A:=I_{1\beta} \text{ and } \rho_B(t):=\unity_4/2 + \cos(\w t)I_{2z} - \sin(\w t)I_{2y}. \label{CNOTDensityOp}
\end{gather}
Recalling the respective Wigner functions from Table~\ref{CartesianWignerRepr}, one obtains 
for $\rho_A$, $\rho_B(t)$, and $\mathcal{H}$ 
the Wigner functions
\begin{subequations}
	\label{CNOTWignerF}
	\begin{align}
	&W_{\rho_A}=\tfrac{1}{4\pi} - \lambda \cos\theta_1 = \wigner(I_{1\beta}),  \\
	&W_{\rho_B}(t)=\tfrac{1}{4\pi} + \cos(\w t)\, \lambda \cos\theta_2 - \sin(\w t)\, \lambda\sin\theta_2 \sin\phi_2,  \\
	&W_{\mathcal{H}}= \w [ 2 \pi ( \tfrac{1}{4\pi} {-} \lambda \cos\theta_1 )\lambda \sin\theta_2 \cos\phi_2  
	+ \lambda \cos\theta_1/2 ], 
	\end{align}
\end{subequations}
where $\lambda=\sqrt{3}/(4\pi)=R/\sqrt{2\pi}$.
The product of  $W_{\rho_A}$ and $W_{\rho_B}(t)$ 
yields the overall Wigner function $W_\rho(t)=2\pi W_{\rho_A} W_{\rho_B}(t)$.
Its time evolution is shown in Fig.~\ref{Figure4} where only one of the two spherical functions
varies in time, reflecting the product form of $W_\rho(t)$.

The explicit form of the time evolution can also be derived from Eq.~\eqref{TwoSpinTimeDeriv}, hence
$\partial W_{\rho} / \partial t = \sqrt{2 \pi}   \Proj ^{\{ 1,2 \}}   (  P_1  + P_2  )$,
where the Poisson brackets can be computed as
$P_1 = 2 \pi W_{\rho_B}(t) \{ W_{\rho_A}  , W_{\mathcal{H}}  \}^{\{ 1 \}}$ and
$P_2 = 2 \pi W_{\rho_A} \{  W_{\rho_B}(t)   ,  W_{\mathcal{H}}  \}^{\{ 2 \}}$.
As the Wigner function $W_{\rho_A}$ depends only on the variable $\theta_1$ and
$W_{\mathcal{H}}$ does not depend on the variable $\phi_1$,
it is straightforward to deduce that $P_1=0$. This implies that $W_{\rho_A}$ is time independent.
The other Poisson bracket $P_2$ can be written as
\begin{align}
P_2 = \w & [\sqrt{2\pi} R  \wigner(I_{1\beta}) ]^2  [ \cos{(\w t)}\, \{\cos\theta_2 , 
\sin\theta_2 \cos\phi_2 \}^{\{ 2 \}} \nonumber \\ 
&- \sin{(\w t)}\, \{\sin\theta_2 \sin\phi_2 , \sin\theta_2 \cos\phi_2 \}^{\{ 2 \}} ].\label{CNOTPoisson} 
\end{align}
Applying the definition of Eq.~\eqref{Equation_Motion}, the Poisson brackets in Eq.~\eqref{CNOTPoisson} 
are computed as
\begin{equation*}
\{\cos\theta_2 , \sin\theta_2 \cos\phi_2 \}^{\{ 2 \}}=-\sin\theta_2 \sin\phi_2 /R,\quad
\{\sin\theta_2 \sin\phi_2 , \sin\theta_2 \cos\phi_2 \}^{\{ 2 \}}=\cos\theta_2/R.
\end{equation*}
The idempotency $(I_{k\beta})^2=I_{k\beta}$ implies 
$2 \pi \Proj ^{\{ 1 \}} \wigner(I_{1\beta})^2=\wigner(I_{1\beta})$ where 
$\Proj ^{\{ 1 \}}$ projects onto rank-one and rank-zero spherical harmonics,
i.e., the term from Eq.~\eqref{CNOTPoisson} results in
$[\sqrt{2\pi} R  \wigner(I_{1\beta}) ]^2=R^2  \wigner(I_{1\beta})$.
Finally, the equation of motion
based on Eq.~\eqref{TwoSpinTimeDeriv} is
\begin{equation*}
\partial W_{\rho} / \partial t =  2 \pi \w \lambda \wigner(I_{1\beta}) 
 [- \cos{(\w t)} \sin\theta_2\sin\phi_2
-\sin{(\w t)} \cos\theta_2  ],
\end{equation*}
which conforms with the explicitly known Wigner function
from Eq.~\eqref{CNOTWignerF} as $\partial W_{\rho} / \partial t = 2 \pi  W_{\rho_A} \partial W_{\rho_B}(t) / \partial t$.

Similarly, one could start with $\tilde{\rho}_A=I_{1\alpha}$  and one would
obtain for $t=0$ that
$\tilde{P}_2 \propto   \wigner(I_{1\beta})\wigner(I_{1\alpha})
\propto (1-\sqrt{3}\cos\theta_1)(1+\sqrt{3}\cos\theta_1)$ and the result
$1-3\cos^2\theta_1$ is proportional to $\Y_{20}$, which is projected by $\Proj ^{\{ 1 \}}$ to zero.
Consequently, the quantum state $\tilde{\rho}(t)$ would be constant, reflecting the nature of the CNOT gate.

Figure~\ref{Figure4} visualizes the time evolution: starting from 
$\ket{\psi(0)}=\ket{\beta \alpha}$ one has the control state $\ket{\beta}$, 
and the state of the  second spin flips from $\ket{\alpha}$ to $\ket{\beta}$,
resulting in $\ket{\psi(T)} \propto \ket{\beta \beta}$. The Wigner function of the
density matrix $I_{1\beta}$ of the pure state $\ket{\beta}$ is 
proportional to $1 - \sqrt{3} \cos\theta$ and is depicted in Fig.~\ref{Figure4}
as a big positive lobe in red (i.e.\ dark gray) lying below a small negative lobe
in green (i.e.\ light gray), refer to the spherical function in the left circle of $W_\rho$.
The spherical function in the right circle of $W_\rho$,
starts with a  big positive lobe lying over a small negative lobe, and this object
is rotated. At time $T/2$, one observes for the second spin 
an equal superposition of $\ket{\alpha}$ and $\ket{\beta}$.
The form of the Wigner function $W_\rho(t)$ during the time evolution
is further highlighted in Fig.~\ref{Figure5} by decomposing it into 
a time-independent part $1/(4\pi)$ and 
a time-dependent part
$\cos(\omega t)\, (\lambda \cos\theta_2) - \sin(\omega t)\, (\lambda \sin\theta_2 \sin\phi_2)$.
The time-dependent part is simply a rotation
of $I_{2z}$ around the $x$ axis.

\subsubsection{Entanglement creation with the CNOT gate\label{CNOT_entangling}}
In order to highlight the generation of entanglement, the time evolution 
under the Hamiltonian of Eq.~\eqref{CNOTHam}
from 
Sec.~\ref{CNOT_evol} is applied to the initial state $\ket{\gamma(0)}=(\ket{\alpha \alpha}+\ket{\beta \alpha})/ \sqrt{2}$.
The notation $|\alpha \rangle$ and
$|\beta \rangle$ for spin-1/2 eigenstates  was introduced in Sec.~\ref{CNOT_evol}.
This results in the time-dependent state $\ket{\gamma(t)}=
U_t \ket{\gamma(0)}=
( \xi(t) \ket{\alpha \alpha} + \ket{\psi(t)})/\sqrt{2}$, cf.\ Eqs.~\eqref{unitary}-\eqref{psisolution}.
In particular for $t=T$ with $T=\pi/ \w$,
one obtains (up to a phase) 
a maximally entangled Bell state
$\ket{\phi^+}=(\ket{\alpha \alpha}+\ket{\beta \beta})/ \sqrt{2}$.

%-------------------------
\begin{figure}[tb]
	\centering
	\includegraphics{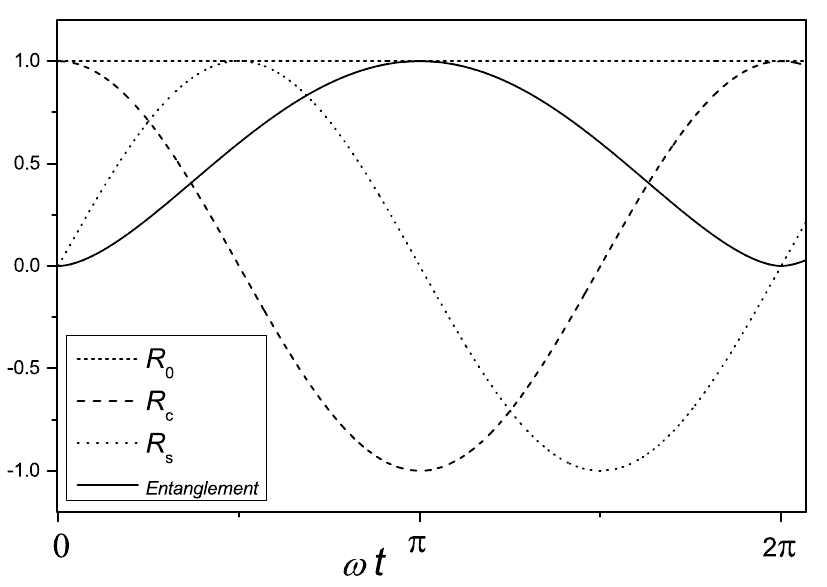}
	\caption{\label{Figure7} Entanglement and contributions of
		the density operator
		$\sigma(t)=\ket{\gamma(t)}\bra{\gamma(t)}=R_0 + \cos(\w t)R_c + \sin(\w t)R_s$ 
		[see Eq.~\eqref{RCRS}]
		during the evolution under the Hamiltonian of Eq.~\eqref{CNOTHam}.
		The von-Neumann entropy of the partial trace is used as entanglement measure \cite{NC00}.
	}
\end{figure}
%-------------------------

Equivalently, the time evolution can be described
on the density operator 
$\sigma(t)=\ket{\gamma(t)}\bra{\gamma(t)}=I_{1\alpha}I_{2\alpha}+\rho(t)+A(t)+A^{\dagger}(t)$, 
where $\rho(t)$ is given in Eq.~\eqref{CNOTdensityop} and 
$A(t)=\bar{\xi}(t)\ket{\psi(t)}\bra{\alpha \alpha}$. The density operator can also
be rewritten as $\sigma(t)=R_0+\cos(\w t)R_c+\sin(\w t)R_s$\footnote{
	One can establish that $\sigma(t)$ satisfies the von-Neumann equation~\eqref{NeumannEq}
	by verifying the commutators $[\mathcal{H},R_0]=0$, $[\mathcal{H},R_c]/i=R_s$, and $[\mathcal{H},R_s]/i=-R_c$
	with the Hamiltonian $\mathcal{H}$ from Eq.~\eqref{CNOTHam}.}, where 
\begin{align}
	R_0 &=[+I_{1x}I_{2x}{-}I_{1y}I_{2y} {+} I_{1x}I_{2\alpha} {+} I_{1\alpha}I_{2\alpha} {+} \tfrac{1}{2} I_{1\beta}]/2,
	\nonumber \\
	R_c &=[-I_{1x}I_{2x}{+}I_{1y}I_{2y} {+} I_{1x}I_{2\alpha} {+} I_{1\beta}I_{2z}]/2, \label{RCRS} \\
	R_s &=[I_{1y}(I_{2\alpha}{-}I_{2x}) - (I_{1x}{+}I_{1\beta})I_{2y}]/2.   \nonumber 
\end{align}
The evolution of these parts
is shown in Fig.~\ref{Figure7}
together with the entanglement of the density operator
as functions of time. Also, we obtain the decomposition
\begin{align} 
	\sigma(0)&=R_0+R_c=(\tfrac{1}{2}\unity_4{+}I_{1x})I_{2\alpha}, \label{fnsst} \\   
	\sigma(T)&= R_0-R_c=\tfrac{1}{4} \unity_4 + I_{1x}I_{2x}-I_{1y}I_{2y} + I_{1z}I_{2z}.
	\nonumber
\end{align}
We switch now to the Wigner functions
\begin{align}
W_{\sigma}(0)&= 2\pi (\tfrac{1}{4\pi} +W_{1x}) W_{2\alpha}, \nonumber \\
W_{\sigma}(T)&=2\pi (\tfrac{1}{8\pi} + W_{1x}W_{2x}-W_{1y}W_{2y} + W_{1z}W_{2z}), \label{Wigner_ent}\\
W_\mathcal{H}&=2\pi W_{1\beta} W_{2x} + \tfrac{1}{2} W_{1z}, \nonumber
\end{align}
for the density operators of Eq.~\eqref{fnsst}
and the Hamiltonian $\mathcal{H}$ of Eq.~\eqref{CNOTHam}
by applying Table~\ref{CartesianWignerRepr},
where $W_{ka}$ denotes the Wigner function of $I_{ka}$.
Figure~\ref{Figure8} depicts the Wigner functions $W_{\sigma}(0)$, $W_{\sigma}(T)$,
and $W_{R_s}$ in  their PROPS representations.
The Wigner function 
$W_{\sigma}(t)= W_{R_0} + \cos(\w t) W_{R_c} +\sin(\w t)W_{R_s}$
satisfies the equation of motion in Eq.~\eqref{TwoSpinTimeDeriv}.\footnote{
	This can be demonstrated for $\w W_{R_s} = \partial  W_{\sigma}(0)  /\partial t$
	[and likewise for $W_{\sigma}(t)$]
	by calculating the Poisson brackets 
	%-------------------------
	$
	\{ W_{\sigma}(0) , W_\mathcal{H} \}^{\{ 1 \}}
	= 
	\w [ -(2\pi)^2 \{ W_{1x},  W_{1\beta} \}^{\{ 1 \}}  W_{2\alpha}  W_{2x}   
	+\pi  \{ W_{1x},  W_{1z} \}^{\{ 1 \}} W_{2\alpha}]$ and
	$\{ W_{\sigma}(0) , W_\mathcal{H} \}^{\{ 2 \}}
	=  
	\w [ (2\pi)^2\allowbreak{} \{ W_{2\alpha},  W_{2x} \}^{\{ 2 \}}   
	[W_{1x}+(4\pi)^{-1}]  W_{1\beta} ]$.
	Afterwards, they are substituted back into the equation of motion;
	note the projection formulas $\Proj ^{\{ 2 \}} W_{2\alpha}  W_{2x}=W_{2x}/(4\pi)$ and
	$\Proj ^{\{ 1 \}} W_{1x}  W_{1\beta}=W_{1x}/(4\pi)$ as well as
	$\{ W_{1x},  W_{1\beta} \}^{\{ 1 \}}=-W_{1y}/\sqrt{2\pi}$ and
	$\{ W_{2\alpha},  W_{2x} \}^{\{ 2 \}} = -W_{2y}/\sqrt{2\pi}$.}

As for $\ket{\phi^+}$, the maximal entangled pure states
\begin{align*}
\ket{\phi^\pm}=(\ket{\alpha \alpha} \pm \ket{\beta \beta})/ \sqrt{2} 
\, \text{ and } \,
\ket{\psi^\pm}=(\ket{\alpha \beta} \pm \ket{\beta \alpha})/ \sqrt{2}
\end{align*}
of a system of two spins $1/2$ have the density matrices
\begin{align*}
\ket{\phi^+}\bra{\phi^+} & =\, \tfrac{1}{4} \unity_4 
+ I_{1x}I_{2x} - I_{1y}I_{2y} + I_{1z}I_{2z}, &
\ket{\phi^-}\bra{\phi^-} & =\,  \tfrac{1}{4} \unity_4 
- I_{1x}I_{2x} + I_{1y}I_{2y} + I_{1z}I_{2z}, \\
\ket{\psi^+}\bra{\psi^+} & =\,  \tfrac{1}{4} \unity_4 
+ I_{1x}I_{2x} + I_{1y}I_{2y} - I_{1z}I_{2z}, &
\ket{\psi^-}\bra{\psi^-} & =\,  \tfrac{1}{4} \unity_4 
- I_{1x}I_{2x} - I_{1y}I_{2y} - I_{1z}I_{2z},
\end{align*}
whose Wigner functions can be computed as in Eq.~\eqref{Wigner_ent}.

\begin{figure}[tb]
	\centering
	\includegraphics{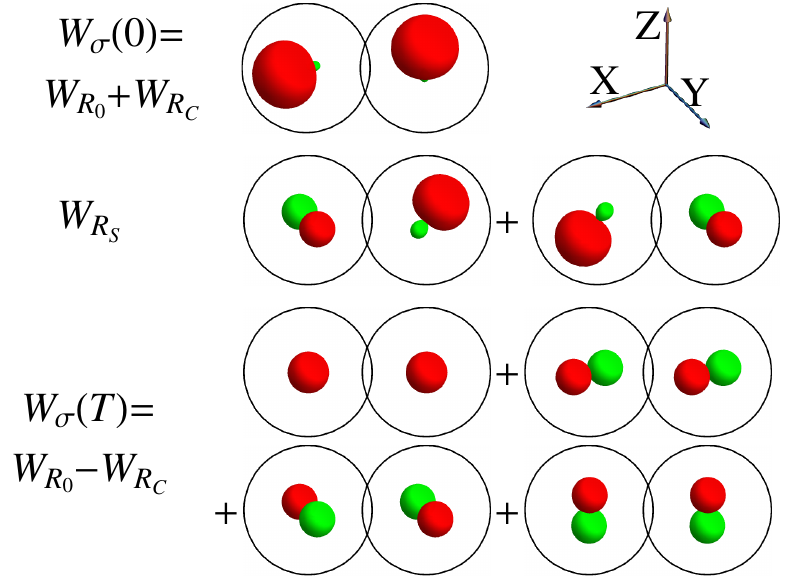}
	\caption{\label{Figure8} 
		(Color online) Illustration of the Wigner functions $W_{\sigma}(0)$, $W_{R_s}$, and $W_{\sigma}(T)$
		with $T=\pi/ \w$ using the PROPS representation.
		The generation of the maximally entangled Bell state
		$\ket{\phi^+}=(\ket{\alpha \alpha}+\ket{\beta \beta})/ \sqrt{2}$
		at time $T$ is reflected by the higher number of terms for
		$W_{\sigma}(T)$.\footnoteref{\thecolors}
	}
\end{figure}

As detailed in Sec.~\ref{VectorOperator} above, the Wigner transform of an operator of the form 
$I_{1x}I_{2x} + I_{1y}I_{2y} + I_{1z}I_{2z}$ results in the scalar product of two vectors
$W(\vec{r}_1,\vec{r}_2)=R^2 \, \vec{r}_1 \cdot \vec{r}_2$, 
providing a geometrical interpretation. Here, the argument of the Wigner function is given
by the unit vectors $\vec{r}_1$ and $\vec{r}_2$ in $\mathbb{R}^3$,
corresponding to the angles $\theta_1,\phi_1$ and $\theta_2,\phi_2$. 
As a result of the addition theorem of spherical harmonics,
the Wigner function is given by the scalar product $\vec{r}_1 \cdot \vec{r}_2$ of the two vectors.
The Wigner function of the maximally entangled state $\ket{\psi^-}\bra{\psi^-}$ is consequently given as 
\begin{equation*}
\wigner(\ket{\psi^-}\bra{\psi^-}) = 1/(8 \pi) - R^2 \, \vec{r}_1 \cdot \vec{r}_2 ,
\end{equation*}
it is thus entirely described by the angle between the two argument vectors. Similarly, the 
maximally entangled state $\ket{\phi^+}\bra{\phi^+}$ has the Wigner function 
$1/(8 \pi) + R^2 \, \vec{r}'_1 \cdot \vec{r}_2$, where the $y$ entry of $\vec{r}'_1$ 
is negated, i.e., $\vec{r}'_1$ 
has the entries
$[\vec{r}_1]_x$, $-[\vec{r}_1]_y$, and $[\vec{r}_1]_z$. 
Therefore, all Wigner functions
of maximally entangled pure states for two spins $1/2$ can be described 
using the scalar product of their argument vectors
after negating certain entries.

\subsection{Evolution of three coupled spins\label{Ex_three}}

Let us now also discuss an example for the case of three coupled spins.
The system starts from the traceless deviation density matrix
$\rho(0) =  I_{2x}$ and evolves
under the Hamiltonian $\mathcal{H}= \pi \JC  (2I_{1z} I_{2z} + 2I_{2z} I_{3z})$
which couples both the first and second spin as well as the second and third spin
with the same coupling strength $\JC$. This results 
in anti-phase and double anti-phase operators \cite{KeelerUnderstanding}.
The corresponding solution of the von-Neumann equation is given by
\begin{equation*}
\rho(t)= \sin (2\pi \JC t) [2I_{1z}I_{2y}{+}2I_{2y}I_{3z}]/2  +   [\cos (2\pi \JC t) {-}1]4 I_{1z}I_{2x}I_{3z}/2 
+[\cos (2\pi \JC t) {+}1] I_{2x}/2.
\end{equation*}
The detectable NMR signal corresponding to $I_{2x}$
is proportional to $[\cos (2\pi \JC t) +1]$, and the corresponding spectrum
has the well-known form of a triplet (see, e.g., Figure 18.9 in \cite{levittspind})
whose lines are separated by $\JC$ and whose relative intensities are given by
$1:2:1$.

The relevant Wigner functions  are given by
$W_{\mathcal{H}} =\pi \JC 2 R^2 (\cos\theta_1+\cos\theta_3) \cos\theta_2 /\sqrt{2\pi}$
and $W_{\rho}(t)=W_0 + \sin (2\pi \JC t) W_s + \cos (2\pi \JC t) W_c$, where

\begin{gather*}
W_0 =\tfrac{2R^3}{3}(1{-}3\cos\theta_1\cos\theta_3) \sin\theta_2\cos\phi_2, \;
W_s =\tfrac{R^2}{\sqrt{2\pi}}(\cos\theta_1{+}\cos\theta_3) \sin\theta_2\sin\phi_2,  \\
\text{and }\; W_c =\tfrac{2R^3}{3}(1{+}3\cos\theta_1\cos\theta_3) \sin\theta_2\cos\phi_2.
\end{gather*}
Their form can be inferred from Table~\ref{CartesianWignerRepr} and also the product structure 
of the Wigner functions $\wigner(I_{1a}I_{2b}I_{3c})=(2\pi)^N \wigner(I_{1 a}) \wigner(I_{2 b}) \wigner(I_{3 c})$
for $a,b,c\in\{x,y,z\}$. 
The evolution of these parts is shown in Figure \ref{Figure10}. 
For this particular case, the equation of motion is given by 
(see Corollary~\ref{result5})
\begin{equation*}
\partial W_\rho / \partial t =
2 \pi \Proj ^{\{ 1,2,3 \}} \sum_{k=1}^3  \{ W_\rho  , W_\mathcal{H}  \}^{\{ k \}}.     
\end{equation*}
We verify that $W_{\rho}(t)$ satisfies the equation of motion by checking
that the conditions $\partial W_0 / \partial t=0$, $\partial W_c / \partial t \propto W_s$,
and $\partial W_s / \partial t \propto W_c$ hold. In the first case,
the Poisson brackets $\{ W_0  , W_\mathcal{H}  \}^{\{ 1 \}}$ and 
$\{ W_0  , W_\mathcal{H}  \}^{\{ 3 \}}$ vanish as they 
are respectively proportional to $\{ \cos\theta_1  , \cos\theta_1  \}^{\{ 1 \}}$
and $\{ \cos\theta_3  , \cos\theta_1  \}^{\{ 3 \}}$. The Poisson 
bracket $\{ W_0  , W_\mathcal{H}  \}^{\{ 2 \}}$ is nonzero, however
its projection by $\Proj ^{\{ 1,2,3 \}}$ is zero. Similarly, one can calculate all Poisson brackets and
projections to complete the verification of 
the equation of motion for this example.

%-------------------------
\begin{figure}[tb]
	\centering
	\includegraphics{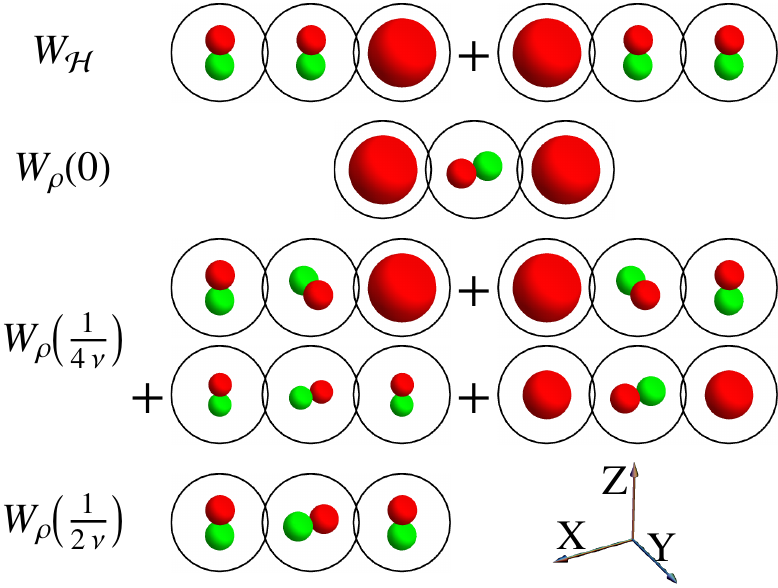}
	\caption{\label{Figure10} 
		(Color online) Visualization of the Wigner functions for the time evolution
		of three coupled spins: the Hamiltonian
		$\mathcal{H}= \pi \JC  (2I_{1z} I_{2z} + 2 I_{2z} I_{3z})$ acts on the starting 
		deviation density matrix
		$\rho(0)=I_{2x}$; $\rho[1/(4\JC)]=I_{1z}I_{2y} + I_{2y}I_{3z}-2I_{1z}I_{2x}I_{3z}+I_{2x}/2$ and 
		$\rho[1/(2\JC)]=-4 I_{1z}I_{2x}I_{3z}$.\footnoteref{\thecolors}
	}
\end{figure}
%-------------------------

\section{Discussion and connections \label{auxmaterial}}

In this section, we complement the description of the Wigner formalism for coupled 
spins from Sec.~\ref{theorysection}
and discuss connections to alternative or related characterizations.
This will allow for simpler interpretations of our formalism and will link
to notions which might provide further avenues to our work.
First, we draw
important connections between 
the Poisson bracket and the canonical angular momentum (see Sec.~\ref{connangmom}).
We continue in Sec.~\ref{InfDimConnection} by relating the Wigner formalism of
finite- and infinite-dimensional quantum systems.
Certain Wigner functions are interpreted in terms of quaternions (see Sec.~\ref{Quaternions}).
Finally, the evolution of non-hermitian states is considered in Sec.~\ref{nonhermitianstatessection}.

%--------------------------------------------------------------------------------------------------------
\subsection{Poisson bracket and the canonical angular momentum \label{connangmom}} 
We detail how the Wigner formalism for coupled spins relates to the angular momentum
of infinite-dimensional quantum systems described by the 
canonical angular momentum operator $\mathcal{L}=r \times p$.
The eigenfunctions of the canonical angular momentum operator $\mathcal{L}$ are spherical harmonics.
The corresponding pure eigenstates are represented by spherical functions $\psi(\theta, \phi)$ and 
evolve according to the time-dependent Schrödinger
equation 
\begin{equation} \label{Schroedinger}
\partial \psi(\theta, \phi) / \partial t = -i \mathcal{H} \psi(\theta, \phi).
\end{equation}
Each component $\mathcal{L}_\alpha$ of the vector operator  corresponding to the Hamiltonian 
$\mathcal{H} = \sum_{\alpha\in\{x,y,z\}}  \w_\alpha I_\alpha $
generates a rotation 
$\mathcal{L}_\alpha \psi(\theta, \phi)$
of the
spherical functions  around the axis $\alpha$ for
$\alpha \in \{x,y,z \}$.
Providing a direct correspondence,
spin operators $\mathrm{I}=(I_{x},I_{y},I_{z})$ generalize the angular momentum 
$\mathcal{L}=(\mathcal{L}_{x},\mathcal{L}_{y},\mathcal{L}_{z})$.
The Wigner representation for spins describes
in general a mixed quantum state using  a
linear combination
of spherical harmonics.
We show that the time evolution of 
the Wigner representation of
a spin $1/2$ is closely related to the
time evolution of an infinite-dimensional quantum system by rewriting the equation of motion into a form 
which is
analogous to the Schrödinger equation in Eq.~\eqref{Schroedinger}.

The equation of motion of a single spin $1/2$ is given by 
Eq.~\eqref{NeumannWignerEvolution} combining the Hamiltonian $\mathcal{H}$ and an operator $A$,
while applying the Poisson bracket from Eq.~\eqref{PBDef}. It can be reformulated by
defining a differential operator $\mathcal{D}$ as a function of the Hamiltonian, acting
on the Wigner function $W_{A}$ of $A$:
%--------------------------------------------------------------------------------------------------------
\begin{equation}
\partial W_{A}/\partial t = 
-i\mathcal{D}(  W_{A} )
= -i\left(   f^\theta_{(\mathcal{H}) } \tfrac{\partial}{\partial \phi}    -   
f^\phi_{(\mathcal{H})}\tfrac{\partial}{\partial \theta}    \right) W_{A}.
\end{equation}
%--------------------------------------------------------------------------------------------------------
Here, the definitions  $f^\theta_{(\mathcal{H}) }= ({\partial 
	W_{\mathcal{H}}}/{ \partial \theta}) ({i}/{R \sin \theta})$ and  
$f^\phi_{(\mathcal{H}) }= ({\partial W_{\mathcal{H}}}/{ \partial \phi})({i}/{R \sin \theta})$
have been applied.
Integrating the differential equation, one obtains the propagator
%--------------------------------------------------------------------------------------------------------
\begin{equation}
W_{A}(t) = 
\exp (  -i \mathcal{D}t )  W_{A}(0),
\end{equation}
%--------------------------------------------------------------------------------------------------------
where the differential operator $\mathcal{D}$ depends on $\mathcal{H}$. 
We consider the Hamiltonians $\mathcal{H}=\w_z I_z$, $\mathcal{H}=\w_x I_x$, and $\mathcal{H}= \w_y I_y$:
Assuming $\mathcal{H}= \w_z I_z$, it follows that
$W_{\mathcal{H}}= {\w_z} \Y_{1,0} /{\sqrt{2}} = R \, \w_z   \cos\theta$.
One obtains the differential operator 
$\mathcal{D}=- \w_z  i ({\partial}/{\partial \phi})    - 0 
= \w_z \mathcal{L}_z$, 
where $\mathcal{L}_z$ is the canonical angular momentum operator $\mathcal{L}_z 
= (r \times p )_z$ in spherical
coordinates (see pp.~662 in \cite{cohen1991quantum}).
The corresponding propagator is then given by $\exp( - i \w_z \mathcal{L}_z  t )$, which by its definition
rotates the Wigner function by an angle $\w_z t$ around the $z$ axis. In the
case of $\mathcal{H}=\w_x I_x$, the Wigner representation has the form $W_{\mathcal{H}}=\w_x R \sin \theta \cos \phi$
and the differential operator is given by the expected canonical angular momentum component
$\mathcal{D}=  i \w_x [  \sin \phi    ({\partial}/{\partial \theta})    + 
\cot \theta  \cos \phi ({\partial}/{\partial \phi})   ]= \w_x \mathcal{L}_x$. 
Finally, $\mathcal{H}=\w_y I_y$ leads to the differential operator $\mathcal{D}= \w_y \mathcal{L}_y$.

As a conclusion for a single
spin $1/2$, 
Wigner functions together with a star commutator correspond to 
canonical angular momentum operators with a cross product. 
This means that a Hamiltonian $\mathcal{H} = \sum_{\alpha\in\{x,y,z\}}  \w_\alpha I_\alpha $
is mapped to the differential operator 
$\mathcal{D} =f ( {\mathcal{L}} )= \sum_\alpha  \w_\alpha ({r} \times {p} )_\alpha$, and the time evolution
is given by 
%--------------------------------------------------------------------------------------------------------
\begin{equation}
\label{WignertoSchroedingerEq}
i \partial W_{\rho}/\partial t=[W_{\mathcal{H}},W_\rho]_{\star} = f ( {\mathcal{L}}  ) W_\rho.
\end{equation}
%--------------------------------------------------------------------------------------------------------
The corresponding propagator can be written as
%--------------------------------------------------------------------------------------------------------
\begin{equation}
\exp( -i \mathcal{D}t ) =  \exp [-i t f ( {\mathcal{L}} ) ].
\end{equation}
%--------------------------------------------------------------------------------------------------------
Generalizations to the case of 
linear Hamiltonians $^J\mathcal{H}=\w_x \, ^J \Jspace I_x + \w_y\,^J \Jspace I_y + \w_z\,^J \Jspace I_z$
with arbitrary $J$ are also possible, cf.\ Eq.~\eqref{evolutionArbitraryJ}.

Comparing Eq.~\eqref{WignertoSchroedingerEq} with Eq.~\eqref{Schroedinger}, one concludes
that the time evolution of Wigner functions for spins is formally equivalent to
the time evolution of the angular part of infinite-dimensional quantum states.
The time evolution of spin-$1/2$ Wigner functions described in Section~\ref{subsecStarProdDef}
is based on the equation of motion given by the star commutator
[refer to Eq.~\eqref{NeumannWignerEvolution}], however, 
Eq.~\eqref{WignertoSchroedingerEq} provides an alternative formulation for a spin $1/2$ based
on the Schrödinger equation by mapping the spin operator $I_\alpha$ onto $({r} \times {p} )_\alpha$.

\subsection{Finite- and infinite-dimensional degrees of freedom \label{InfDimConnection}}
We discuss now
important relations between Wigner representations for finite- and infinite-dimensional degrees of freedom.
In the finite-dimensional case (i.e.\ for spins), Wigner representations
have
been developed in Sec.~\ref{theorysection} for coupled systems extending  approaches
based on the Stratonovich postulates (see Sec.~\ref{AppendixWignerRepr}).
Methods for Wigner representations applicable to 
infinite-dimensional quantum systems with a flat phase space have been more widely discussed in the 
literature, see Sec.~\ref{introtowignerf} and
\cite{carruthers1983,hillery1997,kim1991,lee1995,gadella1995,zachos2005,schroeck2013,SchleichBook,Curtright-review}.
We explore how spin operators can be uniquely mapped onto functions over a phase space that
are restricted to the surface of a sphere. This provides
a rotational covariance for operators and their corresponding Wigner representations.
Expectation values of operators are calculated as 
quasi-probability weighted integrals of phase-space functions. 

We discuss basic properties of the (classical)
Wigner functions for infinite-dimensional spaces (see Sec.~\ref{introtowignerf})
which
have the quality of a flat phase space, in contrast to the case for spins.
Flat phase-space coordinates translate 
in the case of spins 
to
curvilinear spherical coordinates which 
form a Heisenberg pair as their star commutator
(as described in Eq.~\eqref{evolutionArbitraryJ} of
Sec.~\ref{subsecStarProdDef})
results in the canonical commutation relation $[q,p]_\star = i \hbar$.
We also compute an upper bound for the absolute value of Wigner functions and  we show 
by investigating its limit
for $J\rightarrow \infty$ that
arbitrary large
values corresponding to 
localized probability distributions
are possible.

\subsubsection{Phase-space coordinates} \label{connectionPScordinates}
Flat phase-space coordinates $(p,q)$ are replaced by curvilinear coordinates $(R\cos\theta,\phi)$
in the Wigner formalism for spins, where  $R$ denotes a proportionality factor.
This implies that $p$ and $q$ describe coordinates on the surface of a three-dimensional sphere.
In analogy to infinite-dimensional quantum mechanics, where the momentum operator $p$ generates
the translation $(p,q+\mathrm{d}q)$ in the $(p,q)$ phase-space,
the spin-$1/2$ Wigner function $R\cos\theta=\wigner (I_z)$ generates a rotation $(R\cos\theta,\phi+\mathrm{d} \phi)$ 
by the infinitesimal angle $\mathrm{d}\phi$
in the spherical phase-space coordinates $(R\cos\theta,\phi)$. In general,
the operator $^J \Jspace I_z$ generates a rotation of a spin $J$ around the $z$ axis, 
and $^J \Jspace I_z$ is mapped to the function
$R\cos\theta/N_J$.\footnoteref{\theglobprefactor}

\subsubsection{Commutators} \label{connectionComutators}
In the case of the infinite-dimensional Wigner representation,
the star commutator $[f,g]_\star$
is given by $ f {\star} g - g {\star} f  = i \hbar \{f,g\}$ up to $\mathcal{O}(\hbar^3)$
\cite{zachos2005,schroeck2013,SchleichBook}
where $\{f,g\} =\partial_q f \partial_pg - \partial_pf \partial_qg$
denotes here 
the Poisson bracket
from classical physics (see, e.g., Vol.~1, §42 of \cite{landau1976}).
Switching from flat phase-space coordinates $(p,q)$ to the curvilinear coordinates 
$(R\cos\theta,\phi)$
of the spin-$1/2$
Wigner representation and setting $\hbar \rightarrow 1$,
one obtains the same star commutator as in \ref{singlespinstar}, see Eq.~\eqref{NeumannWignerEvolution}.

The canonical commutation relation $[q,p]=i \hbar$,
which translates to $[q,p]_\star = i \hbar$ in the Wigner representation,
states that the infinite-dimensional coordinate and momentum operators are not simultaneously determined.
The coordinates $(R\cos\theta,\phi)$ also form a Heisenberg pair, as conjugate variables via 
the star commutator $[\phi, R\cos\theta]_\star = i$. In general, the formula
$[\phi, \wigner (^J \Jspace I_z)]_\star = i$ is implied by Eqs.~\eqref{CommutatorPBCorrespondance}
and \eqref{evolutionArbitraryJ}, and $^J \Jspace I_z$ is mapped to the function
$R\cos\theta/N_J$.\footnoteref{\theglobprefactor}

The canonical commutation relation implies that the infinite-dimensional operators
$p$ and $q$ have only infinite-dimensional matrix representations.
The same holds in case of spins for the coordinate $\phi$ describing the phase angle in the $x$-$y$ plane: 
it has 
no finite-dimensional matrix representation
as $\wigner[\wigner^{-1} ( \phi ) ] = \phi$ is valid only in the 
case of $J \rightarrow \infty$. 
This can be verified
by defining the inverse Wigner transform of $\phi$ as $\wigner^{-1}(\phi)$ 
[see Eq.~\eqref{inverseWigner}],
and this results in
\begin{align}
& \int_{\theta=0}^{\pi} \int_{\phi=0}^{2 \pi}   \phi \,  \Delta_{J}(\theta, \phi)  
\: \sin{\theta} \,\mathrm{d}\theta  \, \mathrm{d}\phi \\
& \propto \sum_{j=0}^{2J} \sum_{m=-j}^j \Tj_{jm} 
 \int_{\theta=0}^{\pi} 
P_{j |m|} (\cos{\theta}  ) \sin{\theta}  \,\mathrm{d}\theta 
\int_{\phi=0}^{2 \pi} \phi   \, \exp(-i m \phi)  \, \mathrm{d}\phi.  \nonumber
\end{align}
Here, the last integral is the Fourier series expansion of $\phi$ and it specifies
an infinite series in $m$. This implies that a unique matrix representation for $\phi$
exists only in case of $J\rightarrow \infty$.
The canonical commutation relation $[q,p]=i \hbar$ is only valid for 
infinite-dimensional representations.
Likewise, the kernel defining the Wigner transformation of spin operators
[see Eq.~\eqref{SingleKernelDefinition}]
becomes in the limit of $J\rightarrow \infty$
identical to the kernel of an infinite-dimensional quantum system \cite{amiet2000}.

\subsubsection{Normalization and upper bound} \label{connectionNormalization}
In general, normalized operators $\tr(A A^\dagger)=1$ are mapped to functions $W_A(\theta, \phi)$ on
the unit sphere, where the square $\abs{W_A(\theta, \phi)}^2$ of the complex absolute value  provides a 
normalized surface integral.
This can be checked by expanding the operator $A=\sum_{j,m} c_{jm} \Tj_{jm}$ into tensor operators, then its
Wigner transformation is given by $W_A=\sum_{j,m} c_{jm} \Y_{jm}$. The condition $\tr(A A^\dagger)=1$ 
is mapped to
the condition $\sum_{j,m} c_{jm} c^*_{jm}=1$; the normalized surface integral 
$\abs{W_A(\theta, \phi)}^2=| \sum_{j,m} c_{jm} \Y_{jm}|^2 $
for a linear combination of spherical harmonics
follows from the orthonormality of spherical harmonics in Eq.~\eqref{OrtRelat}.
Also, the norm of a matrix is conserved in the Wigner representation.

We compute now an upper bound for the absolute value of the Wigner function $W_A$ 
of a normalized operator $A$ with $\tr(A A^\dagger)=1$.
The Cauchy-Schwarz inequality, 
implies the inequality
\begin{equation}\label{left_right}
\tr(\Delta_J \Delta_J) \tr(A A^\dagger) \geq | \tr(A \, \Delta_J  )  |^2.
\end{equation}
The right-hand side is equal to $|W_A(\theta,\phi)|^2$ where we have applied the definition
of the Wigner function $W_A(\theta,\phi)$ from Eq.~\eqref{WignerTransform}. 
Assuming a normalized operator $A$,
the left-hand side of Eq.~\eqref{left_right} is equivalent to the trace of the square of
the kernel defined in Eq.~\eqref{SingleKernelDefinition}, i.e., the left-hand side is equal to 
\begin{equation*}
\tr(\Delta_J \Delta_J) =\sum_{j=0}^{2J} \sum_{m=-j}^j \Y_{jm} \Y^*_{jm} = \sum_{j=0}^{2J} \frac{2j{+}1}{4 \pi}
=\frac{(2J{+}1)^2}{4\pi}. 
\end{equation*}
The aforementioned statements imply the upper bound 
$|W_A(\theta,\phi)| \leq (2J{+}1)/\sqrt{4 \pi}$ for  
normalized operators $A$. For $J \rightarrow \infty$, the
upper bound goes to infinity, allowing localized but normalized quasiprobability distributions
$W_A(\theta,\phi)=\delta_{\theta-\theta'} \delta_{\phi-\phi'}/\sin{\theta}$
w.r.t.\ both $\theta$ and $\phi$, which
correspond to classical vectors
pointing to the surface of a sphere of unit radius.
Even though spins have no classical counterparts,
a classical description emerges from the quantum one
in the limit of $J \rightarrow \infty$. This follows as 
the growing number $2J{+}1$ of states
allow for larger values in the Wigner function,
while negative regions shrink.

\subsubsection{Implications}
The Stratonovich postulates provide an abstract formulation for the phase-space
representation of spins. Here, we showed the most important links between 
Wigner functions of 
finite-
and infinite-dimensional quantum systems and how to interpret basic properties
of phase-space representations. Phase-space coordinates $(R\cos\theta/N_J,\phi)$ of spins 
span the surface of a sphere. These two coordinates form a Heisenberg pair
with $[\phi, R\cos\theta/N_J]_\star = i$ and consequently $ R\cos\theta/N_J$ generates the
translation of the coordinate $\phi$ corresponding to the rotation of the sphere
around the $z$ axis. The coordinate $\phi$ has no unique matrix representation for a finite
spin $J$. We also showed that an upper bound for the absolute value
of a normalized Wigner functions is proportional 
to the number $(2J{+}1)$ of degrees of freedom. 

\subsection{Wigner functions and quaternions \label{Quaternions}}
In this section, we introduce a variant of Wigner functions based on quaternions.
Quaternions can be represented by $2\times 2$ matrices and the quaternionic product by 
matrix multiplication. 
In Sec.~\ref{matrixquaternion}, we show that the Wigner transformation of these matrices combined
with the star product
derived in Sec.~\ref{singlespinstar} also provides a valid representation of quaternions.
Quaternions can also be represented as a three-dimensional vector equipped with a scalar part,
offering a geometrical interpretation of the quaternionic product.
In Sec.~\ref{vectquaternion}, we show how a three-dimensional vector, 
corresponding to the Pauli vector can be mapped onto
Wigner functions. We provide an explicit form for the quaternionic product in this case.

Based on these results we show in Sec.~\ref{quaternionstarprod} that the star product of
spin-$1/2$ operators is formally analogous to the quaternionic product, when applied to
quaternionic Wigner representations. This offers a geometrical interpretation for the
star product derived in Sec.~\ref{singlespinstar}.

\subsubsection{Matrix representation of quaternions \label{matrixquaternion}}
The set $\mathbb{H}$ of quaternions  can be identified with a four-dimensional vector space
over real numbers $\mathbb{R}^4$. Every element $q \in  \mathbb{H}$ is given as a linear combination
of the basis elements $\mathrm{1}$, $\mathrm{i}$, $\mathrm{j}$, and $\mathrm{k}$, where 
$\mathrm{i}^2=\mathrm{j}^2=\mathrm{k}^2=\mathrm{i}\mathrm{j}\mathrm{k}=-\mathrm{1}$.
Quaternions can be identified with $2\times 2$ matrices spanned by 
the basis elements 
$\mathrm{1}\cong \unity_2$, $\mathrm{i}\cong -i\sigma_x$, 
$\mathrm{j}\cong -i\sigma_y$, $\mathrm{k}\cong -i\sigma_z$, where matrix multiplication replaces
the quaternionic product.
A quaternionic Wigner representation is obtained by mapping the Pauli
matrices to their respective Wigner functions (see Sec.~\ref{AppendixWignerRepr} for the Wigner transformation), 
i.e.,
${W}_{\mathrm{1}}:=1/\sqrt{2 \pi}$, ${W}_{\mathrm{i}}:=-i W_x=-i \wigner(\sigma_x)$, 
${W}_\mathrm{j}:=-i W_y=-i \wigner(\sigma_y)$, ${W}_\mathrm{k}:=-i W_z=-i \wigner(\sigma_z)$ 
where the Wigner representations have the form
\begin{equation}
W_x=2R \sin\theta\cos\phi,\,
W_y=2R \sin\theta\sin\phi,\, W_z=2R \cos\theta.   \label{WignerQuaternionbases} 
\end{equation}
The corresponding quaternionic multiplication is given by the star product described in
Sec.~\ref{singlespinstar}, namely
${W}_\mathrm{i} \star {W}_\mathrm{i}={W}_\mathrm{j} \star {W}_\mathrm{j}
={W}_\mathrm{k} \star {W}_\mathrm{k}={W}_\mathrm{i} \star {W}_\mathrm{j} \star {W}_\mathrm{k}=-{W}_\mathrm{1}$.

\subsubsection{Vectorial representation of quaternions \label{vectquaternion}}
There are other ways to represent quaternions and the quaternionic product.
Let us identify an arbitrary quaternion $q \in  \mathbb{H}$ with a four-dimensional real vector 
$h=h(q)=(r,\vec{v})$ where $r \in \mathbb{R}$
is a real number and $\vec{v} \in \mathbb{R}^3$ defines a three-dimensional real vector.
The product of two quaternions $(r_1,\vec{v}_1)$ and $(r_2,\vec{v}_2)$ is now given by \cite{altmann2005}
\begin{equation} \label{quaternionproduct}
(r_1,\vec{v}_1) (r_2,\vec{v}_2) =(r_1r_2 - \vec{v}_1 \cdot \vec{v}_2, r_1 \vec{v}_2 +
r_2 \vec{v}_1 + \vec{v}_1 \times \vec{v}_2),
\end{equation}
where $\vec{v}_1 \cdot \vec{v}_2$ denotes the scalar product and
$\vec{v}_1 \times \vec{v}_2$ is the cross product.

We can directly translate this into the Wigner representation.
The vectorial part $\vec{v}$ of the quaternion is replaced by the Wigner function
$\mathrm{W}_{\vec{v}}$\footnote{
	The Wigner representation of
	$\vec{v} \cdot \vec{\sigma}=v_x \sigma_x + v_y \sigma_y + v_z \sigma_z$ 
	is exactly $\mathrm{W}_{\vec{v}}$.
}, while the scalar part is replaced by the identity element ${W}_1 := 1/\sqrt{2 \pi}$. 
The Wigner function $\mathrm{W}_h$ of a quaternion $h=(r,\vec{v})$ is then given by
$\mathrm{W}_h = r {W}_1 + v_x {W}_x + v_y {W}_y + v_z {W}_z$
where the basis operators are defined as Wigner representations of the Pauli matrices,
see Eq.~\eqref{WignerQuaternionbases}.
The scalar product for the Wigner functions corresponding to quaternions
is given in terms of spherical functions by
%-------------------------
\begin{equation} \label{WignerScalarProd}
\langle f | g \rangle = \tfrac{1}{2} \int_{\theta=0}^{\pi} \int_{\phi=0}^{2 \pi}  f(\theta, \phi)   
g  (\theta, \phi)   \sin{\theta} \,\mathrm{d}\theta  \, \mathrm{d}\phi.
\end{equation}
%-------------------------
One can write the real coefficients $(r,v_x,v_y,v_z)$ as projections onto the basis functions
$r=\langle {W}_1 | \mathrm{W}_h \rangle$ and $v_\alpha=\langle {W}_\alpha | \mathrm{W}_h \rangle$
with $\alpha\in \{x,y,z\}$. Note that the factor $1/2$ before the integral in Eq.~\eqref{WignerScalarProd}
is included to obtain normalized
basis elements such that  $\langle {W}_1 | {W}_1 
\rangle=\langle {W}_\alpha | {W}_\alpha \rangle=1$.

This construction represents quaternions in the Wigner space as real valued functions.
The correspondence between $h=(r,\vec{v})$ and $\mathrm{W}_h=r{W}_1+\mathrm{W}_{\vec{v}}$
is also implied by the rotational
covariance of a three-dimensional vector $\vec{v}$ and its Wigner representation $\mathrm{W}_{\vec{v}}$.
The respective quaternionic product in this case is detailed in the following
\begin{lemma}\label{quaternionlemma}
	Given two quaternions $h_1=(r_1,\vec{v}_1)$ and $h_2=(r_2,\vec{v}_2)$, their
	product $h_3=h_1 h_2$ is defined in Eq.~\eqref{quaternionproduct}. The corresponding 
	Wigner representations
	$\mathrm{W}_{h_1}= r_1 {W}_1 + \mathrm{W}_{\vec{v}_1}$,
	$\mathrm{W}_{h_2}= r_2 {W}_1 + \mathrm{W}_{\vec{v}_2}$, and
	$\mathrm{W}_{h_3}= r_3 {W}_1 + \mathrm{W}_{\vec{v}_3}$
	satisfy $r_3 =(r_1 r_2 - \langle \mathrm{W}_{\vec{v}_1} | \mathrm{W}_{\vec{v}_2} \rangle )$
	and 
	$\mathrm{W}_{\vec{v}_3}=r_1 \mathrm{W}_{\vec{v}_2} + r_2\mathrm{W}_{\vec{v}_1}
	- \tfrac{1}{2} \{ \mathrm{W}_{\vec{v}_1}, \mathrm{W}_{\vec{v}_2} \} $,
	where the scalar product $\langle \mathrm{W}_{\vec{v}_1} | \mathrm{W}_{\vec{v}_2} \rangle$ is defined in
	Eq.~\eqref{WignerScalarProd} and $ \{ \mathrm{W}_{\vec{v}_1}, \mathrm{W}_{\vec{v}_2} \}$ is the 
	Poisson bracket 
	as defined 
	in Eq.~\eqref{PBDef}.
\end{lemma}
\begin{proof} 
	The equation
	$\langle \mathrm{W}_{\vec{v}_1} | \mathrm{W}_{\vec{v}_2} \rangle  = \vec{v}_1 \cdot \vec{v}_2=v_xv_x+v_yv_y+v_zv_z$ 
	is implied by the orthonormality of basis elements, i.e., $\langle {W}_a | {W}_b \rangle=\delta_{ab}$ holds
	for $a,b \in \{1,x,y,z\}$.
	Finally, the Wigner representation $\mathrm{W}_{\vec{v}_1 \times \vec{v}_2}$ of the cross product
	is equal to $- \{ \mathrm{W}_{\vec{v}_1}, \mathrm{W}_{\vec{v}_2} \}/2$, which
	can be verified by computing Poisson brackets 
	$-\{ {W}_\alpha, {W}_\beta \}/2=\sum_\gamma \epsilon_{\alpha \beta \gamma} {W}_\gamma$
	of basis elements. Here,  
	$\epsilon_{\alpha \beta \gamma}$ denotes the fully antisymmetric Levi-Civita symbol.
\end{proof}

\subsubsection{Relation to the star product \label{quaternionstarprod}}
We will now explain how the star product defined in Result~\ref{result2}
can be
decomposed into a sum of different terms which all relate to a simple multiplication, scalar product, or
cross product  originating from the quaternionic product in Eq.~\eqref{quaternionproduct}.

Consider the matrix representation  $s_1=r_1 \unity_2 - i \vec{v}_1 \cdot \vec{\sigma}$,
and $s_2=r_2 \unity_2 - i \vec{v}_2 \cdot \vec{\sigma}$ 
of two quaternions
with $\vec{\sigma}:=(\sigma_x,\sigma_y,\sigma_z)$ (see Sec.~\ref{matrixquaternion}).
The Wigner functions can be written as
$\mathsf{W}_{s_1}=r_1 {W}_1-i\mathrm{W}_{\vec{v}_1}$ and $\mathsf{W}_{s_2}=r_2 {W}_1-i\mathrm{W}_{\vec{v}_2}$.
The Wigner representation of three-dimensional vectors was discussed in Sec.~\ref{vectquaternion}
and it was shown in Sec.~\ref{matrixquaternion} that the 
quaternionic product corresponds to the star product 
[see Eq.~\eqref{StarStatement}]
\begin{equation}
\mathsf{W}_{s_1} \star \mathsf{W}_{s_2} = \label{starprodquateernion}
\sqrt{2 \pi} \Proj \mathsf{W}_{s_1} \mathsf{W}_{s_2}
{-}
\tfrac{i}{2} \Proj \{ \mathsf{W}_{s_1}  , \mathsf{W}_{s_2} \},
\end{equation}
which
consists of a sum of two terms.
The second term in Eq.~\eqref{starprodquateernion}
is proportional to a Poisson bracket and equals 
\begin{equation*}
- \tfrac{i}{2} \Proj \{ \mathsf{W}_{s_1}  , \mathsf{W}_{s_2} \}  
= \tfrac{i}{2} \{ \mathrm{W}_{\vec{v}_1}  , \mathrm{W}_{\vec{v}_2} \} 
= -i \mathrm{W}_{\vec{v}_1 \times \vec{v}_2};
\end{equation*}
this illustrates the connection to the cross product  $\vec{v}_1 \times \vec{v}_2$
corresponding to the vectorial parts of two quaternions.

The first term in Eq.~\eqref{starprodquateernion} is a projected
pointwise product of Wigner functions 
and can be computed as
\begin{equation}
\Proj \sqrt{2 \pi} ( \mathsf{W}_{s_1}  \mathsf{W}_{s_2} )
=r_1 r_2 {W}_1 + r_1 (-i\mathrm{W}_{\vec{v}_2}) + r_2 (-i\mathrm{W}_{\vec{v}_1})
- \sqrt{2 \pi}  \Proj  (\mathrm{W}_{\vec{v}_1} \mathrm{W}_{\vec{v}_2}) \label{first_term}
\end{equation}
by applying the definitions of
$\mathsf{W}_{s_1}$ and $ \mathsf{W}_{s_2}$.
The last term in Eq.~\eqref{first_term}
is the projection of $\sqrt{2 \pi}  \Proj  (\mathrm{W}_{\vec{v}_1} \mathrm{W}_{\vec{v}_2})$ 
onto the set of
rank-zero spherical harmonics (as there are no rank-one contributions). 
Hence, we obtain
\begin{align}
&\sqrt{2 \pi}  \Proj  (\mathrm{W}_{\vec{v}_1} \mathrm{W}_{\vec{v}_2}) \nonumber
=  
\sqrt{2 \pi}  \Y_{00} \int_{\theta=0}^{\pi} \int_{\phi=0}^{2 \pi} \Y^*_{00} \mathrm{W}_{\vec{v}_1} \mathrm{W}_{\vec{v}_2} 
\\
& \times
\sin{\theta} \,\mathrm{d}\theta  \, \mathrm{d}\phi
=\langle \mathrm{W}_{\vec{v}_1} | \mathrm{W}_{\vec{v}_2} \rangle {W}_1 =  
(\vec{v}_1 \cdot \vec{v}_2) {W}_1.
\end{align}
Finally, substituting all  results back into Eq.~\eqref{starprodquateernion}, the star product can be written as
\begin{equation*}
\mathsf{W}_{s_1} {\star} \mathsf{W}_{s_2} = 
(r_1 r_2 {-} \vec{v}_1 \cdot \vec{v}_2)\: {W}_1 
-i[ r_1 \mathrm{W}_{\vec{v}_2} {+} r_2 \mathrm{W}_{\vec{v}_1} {+} \mathrm{W}_{\vec{v}_1 \times \vec{v}_2}].
\end{equation*}
Consequently, the star product of Wigner functions for a single spin $1/2$ 
(as given in Sec.~\ref{singlespinstar})
can be nicely described in geometrical terms related to quaternionic products.

\subsubsection{Implications}
The density operator of a spin-$1/2$ state can be represented in a geometrical fashion by mapping 
the Pauli vector $\vec{\sigma}$ onto $\vec{v} \in \mathbb{R}^3$ and $\unity_2$ onto $r \in \mathbb{R}$,
refer to Sec.~\ref{quaternionstarprod}.
The Wigner transformation of this geometrized
density operator is then given by ${W}_\rho = r_\rho {W}_1 + \mathrm{W}_{\vec{v}_\rho} $,
and also ${W}_\mathcal{H} = r_\mathcal{H}{W}_1 + \mathrm{W}_{\vec{v}_\mathcal{H}} $, 
refer to Sec.~\ref{vectquaternion}.
The star product of two such Wigner functions 
was determined in Sec.~\ref{singlespinstar} and
can be decomposed into
a sum of different terms, where each term from Eq.~\eqref{prestar_single} can be given a geometrical interpretation.
The Poisson bracket  
$\tfrac{i}{2} \{ \mathrm{W}_{\vec{v}_\rho}  , \mathrm{W}_{\vec{v}_\mathcal{H}} \} $ 
in Eq.~\eqref{prestar_single}
is the analogue
of the cross product of two vectors and results in $ -i \mathrm{W}_{\vec{v}_\rho \times \vec{v}_\mathcal{H}}$.
The projection of the product $ \sqrt{2 \pi}  \Proj  (\mathrm{W}_{\vec{v}_\rho} \mathrm{W}_{\vec{v}_\mathcal{H}})$
in Eq.~\eqref{prestar_single} results in the scalar product of two vectors 
$ (\vec{v}_\rho \cdot \vec{v}_\mathcal{H} ) {W}_1$. Consequently, the 
time evolution from Eq.~\eqref{NeumannWignerEvolution}
translates to $\partial \mathrm{W}_{\vec{v}_\rho} /\partial t = 2 \mathrm{W}_{\vec{v}_\rho \times \vec{v}_\mathcal{H}}$.
For a single spin $1/2$ in an external magnetic field $\vec{B}$ the Hamiltonian has the form
$\vec{v}_\mathcal{H} = \gamma \vec{B}/2$ with $\gamma$ being the gyromagnetic ratio.
The time evolution is therefore given by 
$\partial \mathrm{W}_{\vec{v}_\rho} /\partial t = \gamma \mathrm{W}_{\vec{v}_\rho \times \vec{B}}$, 
and is formally identical to the classical equation of motion,
refer also to Theorem 5 in \cite{VGB89}. But we emphasize that
non-hermitian states cannot be represented using a single quaternion, 
which can be very well achieved using a single Wigner function, 
see Figure~\ref{NonHermitian}.
Non-hermitian spin operators are discussed further in
Section~\ref{nonhermitianstatessection}.

\subsection{Evolution of non-hermitian states \label{nonhermitianstatessection}}

We consider the time evolution of non-hermitian states 
to highlight that Wigner functions offer a natural way to
represent also non-hermitian spin operators.
In Section~\ref{singlespinnonhermitian}, we discuss the time evolution of the
coherence state $I_{-}$ of a single spin, while Section~\ref{twospinnonheritiansection}
provides an alternative representation for a two-spin example from Section~\ref{twospinwignereval}
utilizing a decomposition of the Wigner function
of a hermitian operator
into non-hermitian parts.

%-------------------------
\begin{figure}[t]
	\centering
	\includegraphics{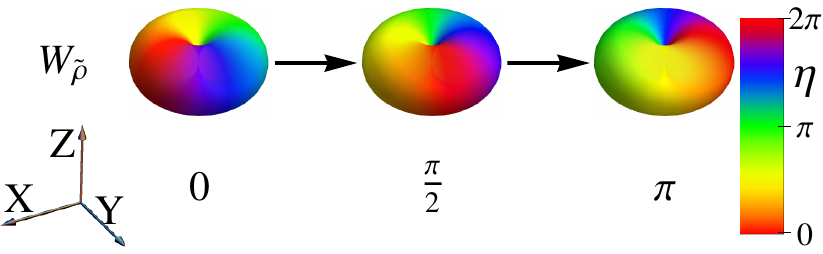}
	\caption{\label{NonHermitian}(Color online) Evolution of the non-hermitian state $\tilde{\rho}=I_{-}$ under the
		Hamiltonian $\mathcal{H}=\w I_z$. The Wigner function  $W_{\tilde{\rho}}(t)=\exp{(i\w t)} \Y_{1,-1}$
		is plotted at times $t=0$, $\w t = \pi /2$, and $\w t = \pi$. The rainbow colors 
		represent the complex phase factor $\exp{(i\eta)}$ of the Wigner function.
		The color changes clockwise from red (dark gray) at the phase value of zero via
		yellow to green (light gray) at the phase value of $\eta=\pi$ and continues
		via blue back to red at the phase value of $\eta=2\pi$, refer to the color bar.}
\end{figure}
%-------------------------

\subsubsection{Evolution of a single spin \label{singlespinnonhermitian}}
We consider the evolution of non-hermitian single-spin states.
Starting from the traceless deviation matrix $\tilde{\rho}=I_{-}=\T_{1,-1}$, the corresponding Wigner function
is $W_{\tilde{\rho}}(\theta,\phi)=\Y_{1,-1}(\theta,\phi)$. The time evolution
is then determined by Eq.~\eqref{Equation_Motion} as 
$
\partial W_{\tilde{\rho}}(\theta, \phi, 0)  / \partial t
= \w 	\{  \Y_{1,-1},  \Y_{10}    \}   / \sqrt{2}
=i \w  \Y_{1,-1}
$
[see, e.g., Eq.~\eqref{CommutatorPBCorrespondance}].
The complex Wigner function $ W_{\tilde{\rho}}(t)=\exp{(i\w t)} \Y_{1,-1}$
picks up only a phase factor during the evolution.
Figure~\ref{NonHermitian} shows $W_{\tilde{\rho}}(t)$ at different times. 
The Wigner functions capture the rotational covariance of the coherence state $I_-$.
The colors yellow and blue correspond to the phase factors $i$ and $-i$, respectively.

\subsubsection{Two coupled spins \label{twospinnonheritiansection}}
The solution of the Wigner function in Eq.~\eqref{twospinexamplesolution}
can be rewritten in terms of spherical harmonics as
%-------------------------
\begin{align*} 
W_{\rho}(t) &= \tfrac{1}{\sqrt{2}} c(t)   [ \Y_{1,-1}(\theta_1, \phi_1) 
{-}  \Y_{11}(\theta_1, \phi_1)  ] 
\Y_{00}(\theta_2, \phi_2)  \\
&+ \tfrac{i}{\sqrt{2}} s(t)   [ \Y_{1,-1}(\theta_1, \phi_1) 
{+} \Y_{11}(\theta_1, \phi_1)  ]  \Y_{10}(\theta_2, \phi_2).
\end{align*}
%-------------------------
Rearranging the terms as $W_{\rho}(t)=W_A(t) + W_A^*(t)$ with
\begin{equation} 
W_A(t)=\tfrac{1}{\sqrt{2}}\Y_{1,-1}(\theta_1,\phi_1)  
[  c(t) \Y_{00}(\theta_2, \phi_2) +  i  s(t) Y_{10}(\theta_2,\phi_2)],  \label{reagrangement} 
\end{equation}
the time dependence can entirely be brought into the second spin,
see Fig.~\ref{TwoSpinNonHermitian}. The overall Wigner function is still hermitian,
even though it is represented as a sum of two not necessarily hermitian operators.

%-------------------------
\begin{figure}[tb]
	\centering
	\includegraphics{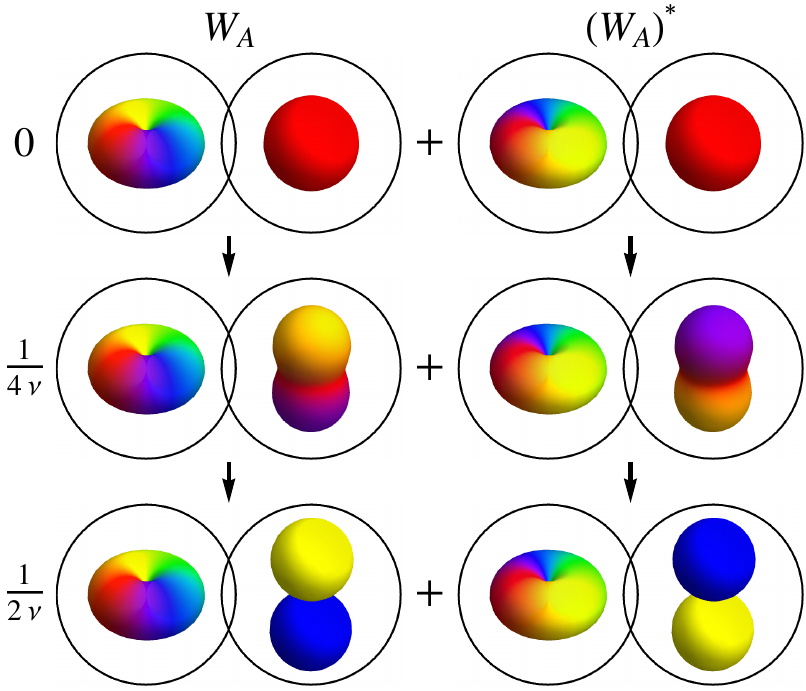}
	\caption{\label{TwoSpinNonHermitian} (Color online) 
		Alternative representation of the deviation density matrix
		$\rho(t) = \cos( \pi \JC t )  I_{1x} + \sin( \pi \JC t)   2 I_{1y} I_{2z}$ by the Wigner
		function in the form $W_{\rho}(t)=W_A(t) + W_A^*(t)$ as shown in Eq.~\eqref{reagrangement}.  
		The Wigner function $W_A^*(t)$ can be obtained by rotating the
		Wigner function $W_A(t)$ by $\pi$ around the $x$ axis.
		Only the spherical function of the second
		spin is time dependent. The second component of $W_A(t)$ corresponds to
		$\cos{( \pi \JC t ) } \Y_{00} +  i  \sin{( \pi \JC t ) } Y_{10}(\theta_2,\phi_2)$, and the first one to 
		$\Y_{1,-1}(\theta_1,\phi_1)$.
	}
\end{figure}
%-------------------------

\section{Conclusion\label{conclusion}}
We presented a general approach for representing arbitrary coupled spin operators in the form of spherical functions,
which we propose as an extension and unification of Wigner function formalisms for single spins.
In particular, we solved the open question of how to compute the time evolution 
of coupled spins
in a consistent Wigner frame.
Our approach
gives also rise to the possibility of visualizing spin operators in terms of a linear combination
of spherical harmonics.

A Wigner function is formally a quasi-probability distribution 
as negative values appear for certain operators.
The negativity of the Wigner function might be interpreted as a signature
of quantumness \cite{Ferrie11}.
However, the significance of the negativity of Wigner functions as an indicator of quantumness
is still debated 
in the literature,
even in the infinite-dimensional case \cite{Johansen97,BW99,BW98,KZ04,RMMJ05,DMWS06,Spekkens08,MKC09,KMMR09,WB12,MO14}.
In the finite-dimensional case of a spherical phase space, oscillating fringes and interference patterns
have been interpreted as quantum signatures \cite{DowlingAgarwalSchleich,Harland,Agarwal81,BC99,APS97};
we refer to \cite{Ferrie11,FE08,FE09} for a more systematic approach.
Recently, the negativity question has also been examined in the context of discrete Wigner functions
\cite{howard2014}.
Potential implications
from the discussion of the negativity of Wigner functions
do not impact our approach to describe the time evolution
of coupled spins systems 
using Wigner functions
and go well beyond the scope and intent of the current work.

We were especially interested in how the time evolution of a coupled spin system can be predicted if 
only the Wigner representations of
both the Hamiltonian and the density operator are available. We introduced
a general method for computing the time evolution of arbitrary coupled spin-$1/2$ systems 
while operating directly on their Wigner functions. This method
is based on a generalization of the Poisson bracket, 
which consists of partial derivatives for both the Hamiltonian and the density operator.
Hence, we provide an interpretation for 
how the time evolution in the Wigner space is governed.
We focused in this work on the
non-trivial case of coupled spins $1/2$, while 
the generalization to arbitrary spin numbers $J$ will be considered elsewhere.

In order to describe the time evolution of coupled spins $1/2$, the star product of two spherical
functions was discussed and developed in detail, and its properties were studied.
Simplified formulas for the time evolution
have been given for
up to three coupled spins $1/2$.
Multiple examples were analyzed and visualized to convey 
important features of our approach
and to stress the operator decomposition into sums of product operators.
We also discussed how the Wigner representation
of spins is related to the canonical angular momentum and quaternions.
Moreover, its relation to the classical,
infinite-dimensional Wigner representation was investigated.

There are different possibilities for mapping spin operators onto visualizable objects
like vectors or spherical functions.
We focused on the Wigner-function technique generalizing spherical functions 
and applied these  to coupled spins. This
phase-space representation transforms naturally under arbitrary rotations of the individual spins.
Similar visualization approaches such as the DROPS representation of \cite{DROPS} might have advantages
for certain applications, and we hope to extend our method for modeling the time evolution
in the Wigner representation to these approaches in the future. And interest in these Wigner methods
in applications can also appreciated from recent related work 
\cite{tilma2016,CBJ15}.

More broadly, we have made theoretical concepts
usually established in the context of Wigner functions more palpable
by visualizing their effects using three-dimensional  illustrations.
This provides a convenient tool for analyzing the time evolutions
of finite-dimensional coupled quantum systems 
using Wigner functions on a continuous phase space
and 
facilitates the adoption of the Wigner formalism 
for coupled spin systems.

\section*{Acknowledgments}
	B.K.\ acknowledges financial support from the scholarship program of the 
	Bavarian Academic Center for Central, Eastern and Southeastern Europe (BAYHOST).
	R.Z. and S.J.G. acknowledge support from the Deutsche
	Forschungsgemeinschaft (DFG) through Grant No.\ Gl 203/7-2.

\appendix

\section{Tensor operators, embedded operators, and normalization factors \label{EmbeddedOperatorsExplanation}}
Here, we provide a short tutorial on tensor operators, embedded operators,
and normalization factors
as far as they are used in Secs.~\ref{Examplessection}, \ref{theorysection}, and \ref{advancedexamples}.
The four single-spin tensor operators
$\T_{00}=\unity/\sqrt{2}$, $\T_{1,-1}$, $\T_{10}$, and $\T_{11}$
from Eq.~\eqref{TensorOperatorMatricesALL}
span the space of all $2\times 2$ matrices, and they conform to
the defining relation in Eq.~\eqref{TOpDefALL}.
Any $2\times 2$ matrix $A$ can be expanded 
in terms of the tensor operators
$\T_{jm}$ as 
$A=\sum_{j\in\{0,1\}} \sum_{-j \leq m\leq j} \tr( \T_{jm}^\dagger A)\, \T_{jm}$.
This follows since tensor operators are chosen as orthonormal, i.e.,
the product $\T_{jm}^\dagger \T_{j'm'}$ of any two of them
has trace one for identical operators with $j=j'$ and $m=m'$, and trace zero
otherwise. However,
Cartesian product operators, such as $I_x$ are only orthogonal, and 
not normalized [e.g.\ $\tr(I_x^\dagger I_x)=1/2$].

Operators for a two-spin system can be decomposed into sums of 
tensor products $\T_{j_1m_1} \otimes \T_{j_2m_2}$
of two $2\times 2$ matrices. There are sixteen linearly independent	
tensor product operators of this form, providing
an orthonormal basis of $4 \times 4$ matrices (cf.\ Lemma~\ref{basis}).
The tensor product operators can be divided into four subsets
depending on which spins they act on: The normalized
identity operator $\T_{00} \otimes \T_{00}$ acts on no spin at all.
The normalized linear operators $\T_{1m} \otimes \T_{00}$ with $m \in \{-1,0,1\}$ act on the first spin,
and $\T_{00} \otimes \T_{1m}$ acts on the second one. 
Finally, the normalized bilinear 
operators $\T_{1m_1} \otimes \T_{1m_2}$ act on both spins ($m_1,m_2 \in \{-1,0,1\}$).

%------------------------------------------------------------------------------------------------------------------------------------
\begin{table}[tb] 
	\centering
	\caption{\label{ProdOpSimplif} Cartesian product operators in a coupled two-spin system.
		Any linear combination of operators in the same column or row
		remains a product operator.}
	\begin{tabular}{@{\hspace{6mm}}l@{\hspace{16mm}}l@{\hspace{16mm}}l@{\hspace{16mm}}l@{\hspace{6mm}} }
		\\[-2mm]
		\hline\hline
		\\[-2mm]
		%------------------------------------------------------------------------------------------------------------------------------------
		$\unity_{4}$  &   $ I_{1x}$ &   $I_{1y}$ &  $I_{1z}$  \\[1mm]
		$I_{2x}$     &    $ I_{1x}I_{2x}$ &      $I_{1y}I_{2x}$  &  $I_{1z}I_{2x}$\\[1mm] 
		$I_{2y}$     &    $ I_{1x}I_{2y}$ &      $I_{1y}I_{2y}$  &  $I_{1z}I_{2y}$\\[1mm] 
		$I_{2z}$     &    $ I_{1x}I_{2z}$ &      $I_{1y}I_{2z}$  &  $I_{1z}I_{2z}$
		\\[1mm] \hline \hline
	\end{tabular} 
\end{table}
%------------------------------------------------------------------------------------------------------------------------------------

Similarly, one can also introduce embedded operators 
$\Ta_{1m}:=\T_{1m} \otimes \T_{00}$ and 
$\Tb_{1m}:=\T_{00} \otimes \T_{1m}$ which 
are single-spin operators embedded into
$4 \times 4$ matrices (cf.\ Lemma~\ref{single-spin}).
These embedded operators enable a description without
(explicit) tensor products. This implies that 
a bilinear operator $\T_{j_1m_1} \otimes \T_{j_2m_2}$ 
can written as a matrix product of embedded operators
with an additional normalization factor.
For example, the tensor product
\begin{equation*}
B=\T_{10} \otimes \T_{11} = 
\begin{pmatrix}
\tfrac{1}{\sqrt{2}} & 0\\
0 & \tfrac{-1}{\sqrt{2}}
\end{pmatrix}
\otimes
\begin{pmatrix}
0 & -1\\
0 & 0
\end{pmatrix}
=
\begin{pmatrix}
0 & \tfrac{-1}{\sqrt{2}} & 0 & 0 \\
0 & 0 & 0 & 0 \\
0 & 0 & 0 & \tfrac{1}{\sqrt{2}} \\
0 & 0 & 0 & 0 
\end{pmatrix}
\end{equation*}
is normalized. Using the normalized 
single-spin operators
\begin{equation*}
\Ta_{10}=
\begin{pmatrix}
\tfrac{1}{2} & 0 & 0 & 0 \\
0 & \tfrac{1}{2} & 0 & 0 \\
0 & 0 & \tfrac{-1}{2} & 0 \\
0 & 0 & 0 & \tfrac{-1}{2} 
\end{pmatrix}  \,\text{ and }\,
\Tb_{11}=
\begin{pmatrix}
0 & \tfrac{-1}{\sqrt{2}} & 0 & 0 \\
0 & 0 & 0 & 0 \\
0 & 0 & 0 & \tfrac{-1}{\sqrt{2}} \\
0 & 0 & 0 & 0 
\end{pmatrix},
\end{equation*}
we can form the matrix product
\begin{equation*}
\Ta_{10} \Tb_{11}=
\begin{pmatrix}
0 & \tfrac{-1}{2\sqrt{2}} & 0 & 0 \\
0 & 0 & 0 & 0 \\
0 & 0 & 0 & \tfrac{1}{2\sqrt{2}} \\
0 & 0 & 0 & 0 
\end{pmatrix}
=B/2,
\end{equation*}
which differs from
$B=\T_{10} \otimes \T_{11}$ by a factor of $2$; 
$B/2$ has a squared norm of $\tr(B^\dagger B)/4=1/4$.
The normalization factors for general spin systems are given in 
Table~\ref{tensordef}.
Products of embedded single-spin operators commute if they act on different spins,
e.g. $\Ta_{10} \Tb_{11}=\Tb_{11} \Ta_{10}$.

\section{Visualization of Wigner functions \label{visualizationExplain}}

In this appendix, we shortly discuss certain 
properties of the visualization of Wigner functions
introduced in Sec.~\ref{Examplessection}.
The Wigner functions of product operators 
in a two-spin system
are in general of the form
$\lambda f(\theta_1,\phi_1)g(\theta_2,\phi_2)$ with  $\lambda\in \C$.
This is visualized as two overlapping circles to reflect the product nature of the spherical functions.
The first circle contains the Wigner function $\sqrt{\lambda}f(\theta_1,\phi_1)$, and 
the second one contains $\sqrt{\lambda} g(\theta_2,\phi_2)$.
The prefactor $\lambda$ is distributed equally, but different choices are possible,
and the size of the visualized
objects can vary assuming that $\lambda$ stays constant.

If the desired operator is a linear combination of product operators, it is represented as a sum of 
product operators in the PROPS representation.
However the decomposition of an operator into sums of product operators is not unique, consequently there are different
visualizations possible for the same operator. In a system of $N$ spins 1/2 there are $4^N$ independent basis operators,
and these can be combined in order to minimize the number of product operators in a sum decomposition.
We provide a {\sc mathematica} package for computing with examples of Wigner functions
of spin operators, their Poisson brackets, and 
the time evolution \cite{supplemental}.

For a two-spin system there are sixteen independent product operators, but any operator can be decomposed
into a sum of at most four product operators. This is verified using Table~\ref{ProdOpSimplif}
where basis operators are grouped such that
any linear combination within the same row or column results again in a product operator,
for example
\begin{equation}
aI_{2z} + b I_{1x}I_{2z} + c I_{1y}I_{2z} + d I_{1z}I_{2z}  = (a \, \unity + b I_{1x} + c I_{1y} + d I_{1z})I_{2z}.
\end{equation}

\section{Stratonovich postulates \label{stratpostulates}}
We summarize the postulates on which the continuous Wigner representation
for spins relies. Starting from the original set of postulates defined by Stratonovich \cite{stratonovich},
which are then
generalized to the case of coupled spins.

\subsection{Stratonovich postulates for a single spin \label{stratpostulatessingle}}
The Wigner representation for spins is a generalization of the flat phase-space representation
$(p,q)$ of quantum mechanics to the phase-space representation $(R\cos\theta,\phi)$ over the sphere.
In the generalized Wigner representation, spin operators are mapped onto spherical functions.
This mapping is based on the Stratonovich postulates \cite{stratonovich}, providing four fundamental
requirements for pairs of operators and their corresponding Wigner functions. Let us denote the
Wigner function of an arbitrary spin operator $A$ as $W_A$, then the postulates are given by
\begin{align*}
\text{(i)} \quad & \text{linearity:}  \quad \text{ $A \rightarrow W_A$ is one-to-one}, \\
\text{(ii)} \quad & \text{reality:} \quad  \text{if $B=A^\dagger$, then $W_B=W_A^*$},  \\
\text{(iiia)} \quad & \text{normalization:} \quad  \tr(\unity A) = \int_{S^2} W_{\unity} W_A  \mathrm{d} 
\Omega,  \\
\text{(iiib)} \quad & \text{traciality:} \quad  \tr( B^\dagger A ) = \int_{S^2} W_B^* W_{A} \mathrm{d} 
\Omega,  \\
\text{(iv)} \quad & \text{covariance:} \quad  
\text{$W_{\mathcal{R}(A)}(\Omega)=W_A(\mathcal{R}^{-1} \Omega)$},
\end{align*}
where $\mathrm{d} \Omega$ denotes the surface element 
$\sin\theta \mathrm{d}\theta \mathrm{d}  \phi$ in the spherical phase space.
A consequence of the postulate (ii) is that Wigner functions of a hermitian operator are real functions.
The postulates (iiia) and (iiib) can be reinterpreted by introducing the scalar products 
$\langle B|A\rangle=\tr(B^\dagger A)$ for matrices and 
$\langle W_B|W_A\rangle=\int_{S^2} W_B^* W_{A} \mathrm{d} \Omega$
for Wigner functions (by integrating over the unit sphere).
Then, the postulates (iiia) and (iiib) can be summarized as $\langle B|A\rangle=\langle W_B|W_A\rangle$.
The postulate (iv) determines the Wigner function of an operator $A$
if $A$ is rotated by a unitary conjugation resulting in $\mathcal{R}(A)$,
then the arguments of the Wigner function (or equivalently the sphere) are 
rotated inversely by the same rotation.

%--------------------------------------------------------------------------------------------------------
\subsection{Generalization of the Stratonovich postulates \label{stratpostulatesmulti}}

We extend the Wigner formalism to multiple, coupled spins.
The postulates  (iii) and (iv) are generalized in order to interpret 
the Stratonovich postulates for the Wigner representation of $N$ coupled spins 
as mapping operators onto sums of function on $N$ spheres.
The condition in postulate (iii) is extended to an integral 
$\int_{S_1^2} \dots \int_{S_N^2} W_B^* W_{A} \mathrm{d} \Omega_1 \dots \mathrm{d} \Omega_N$
over $N$ spheres.
For the PROPS representation, Postulate (iv) is generalized by interpreting $\mathcal{R}$ as an 
arbitrary rotation in the form of a product of
local rotations as $\mathcal{R}=\prod_k^N \mathcal{R}_k$. 
An arbitrary local rotation $\mathcal{R}_k$ acting on spin $k$ 
is described in the Wigner space as the inverse rotation of the $k${th} sphere
$W_A(\Omega_1, \dots,  \mathcal{R}_k^{-1} \Omega_k, \dots \Omega_N )$ of 
the corresponding Wigner representation. The mapping of spin operators 
onto spherical functions described in \cite{DROPS} satisfies the generalized Stratonovich postulates under simultaneous 
rotations $\mathcal{R}$ built from equal rotations $\mathcal{R}_k$ of all spheres.
Consequently, the current approach generalizes the rotational covariance criterion in \cite{DROPS} 
to arbitrary local rotations.
Finally, the generalized postulates are
\begin{align*}
\text{(i)} \quad & \text{linearity:}  \quad \text{ $A \rightarrow W_A$ is one-to-one}, \\
\text{(ii)} \quad & \text{reality:} \quad  \text{if $B=A^\dagger$, then $W_B=W_A^*$},  \\
\text{(iiia)} \quad & \text{normalization:} \:  \tr(\unity A) = \idotsint\displaylimits_{S_1^2\dots S_N^2} W_{\unity} W_A 
\prod_k^N \mathrm{d}  \Omega_k,  \\
\text{(iiib)} \quad & \text{traciality:} \quad  \tr(B^\dagger A) = \idotsint\displaylimits_{S_1^2\dots S_N^2}  W_B^* W_{A} 
\prod_k^N \mathrm{d}  \Omega_k, \\
\text{(iv)} \quad & \text{covariance:} \quad  
W_{\mathcal{R}(A)}(\Omega_1,\ldots,\Omega_N)
=W_A(\mathcal{R}_1^{-1} \Omega_1,\ldots,\mathcal{R}_N^{-1}\Omega_N).  
\end{align*}

%-------------------------
\begin{figure}[ht!]
	\centering
	\includegraphics{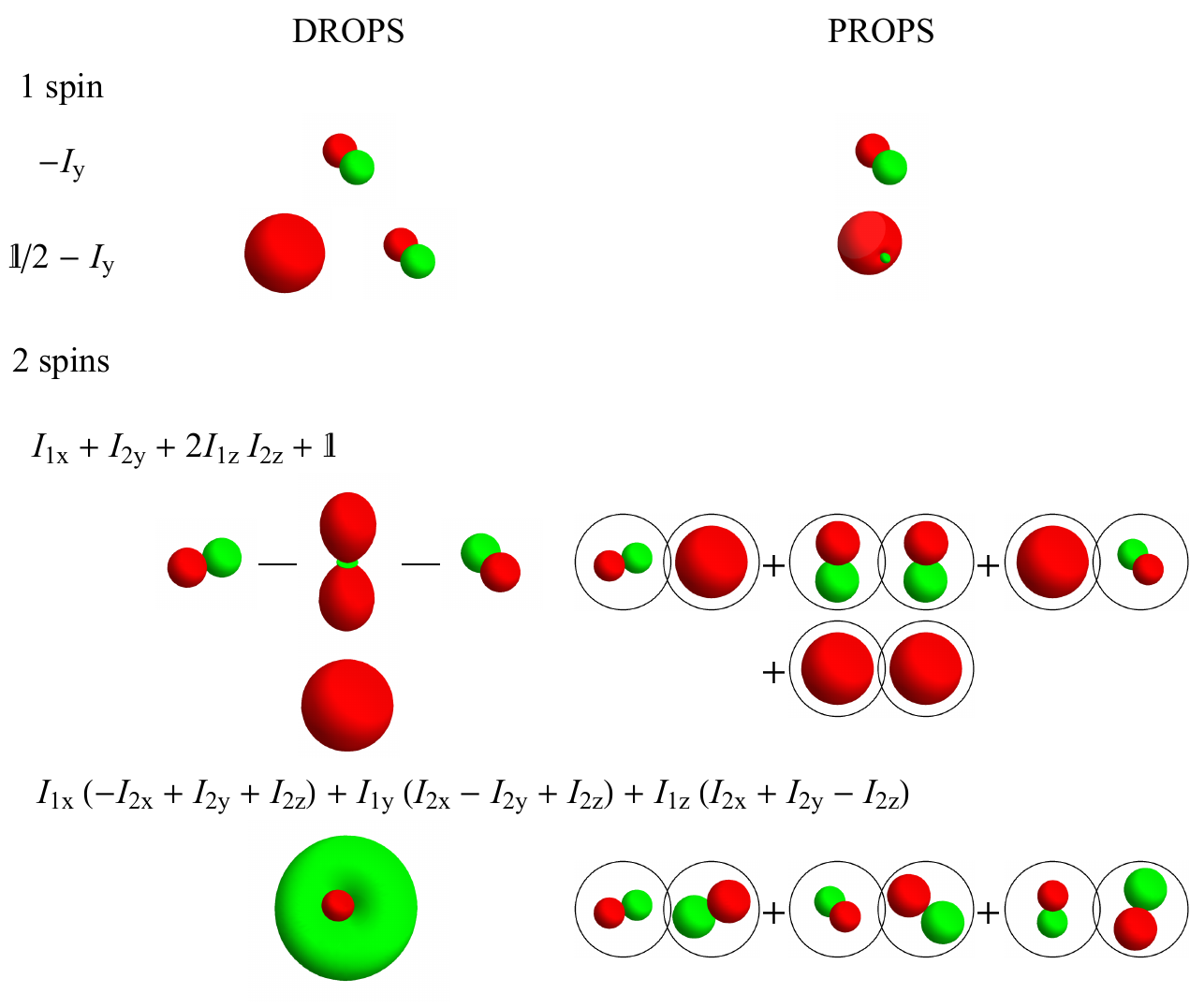}
	\caption{\label{DROPS_PROPS} (Color online) The
		DROPS representation using the LISA basis is compared to the PROPS representations. 
		For a single spin $1/2$, the DROPS representation consists in general of
		two spherical function where the first one corresponds to the identity part $\unity_2$ which acts on no spin at all
		and the second one corresponds in this example to the operator $-I_y$. Both contributions
		are combined in the PROPS representation.        
		The DROPS representation 
		of two coupled spins $1/2$
		consists in general of four spherical functions: the identity part
		$\unity_4$, linear operators such as $I_{2x}$ and $I_{2y}$ acting on the first  and  second  spin, as well as
		bilinear operators such as $2I_{1z}I_{2z}$ acting on both spins. 
		The number of product operators in the PROPS representation
		can be in certain cases reduced by combining some of its components, e.g., 
		$I_{1x}+I_{2y}+2I_{1z}I_{2z}+\unity_4$ equals to 
		$(I_x{+}\unity_2/2)\otimes \unity_2 +  \unity_2 \otimes
		(I_y{+}\unity_2/2) + 2I_{1z}I_{2z}$.\footnoteref{\thecolors}
	}
\end{figure}
%-------------------------

\section{Comparison to the DROPS representation \label{dropsappendix}}

We compare our visualization technique with the
so-called DROPS representation from \cite{DROPS} both of which represent operators of  
coupled spin systems in the form of a collection of spherical functions.
A particular case of the DROPS representation is given by the LISA basis \cite{DROPS}
which is spanned  by
a set of tensor operators $T^{(\ell)}_{jm}$ with $\ell \in L$.
Each element $T^{(\ell)}_{jm}$ is mapped to a spherical harmonic 
$\Y_{jm}$ and
transforms naturally 
under simultaneous rotations of all spins.
The DROPS representation using the LISA basis
provides a compact visualization of coupled spin operators 
by plotting $|L|$  spherical functions. 
Figure~\ref{DROPS_PROPS} 
contains a comparison for spin operators of one and two spins.

\begin{figure}[ht!]
	\centering
	\includegraphics{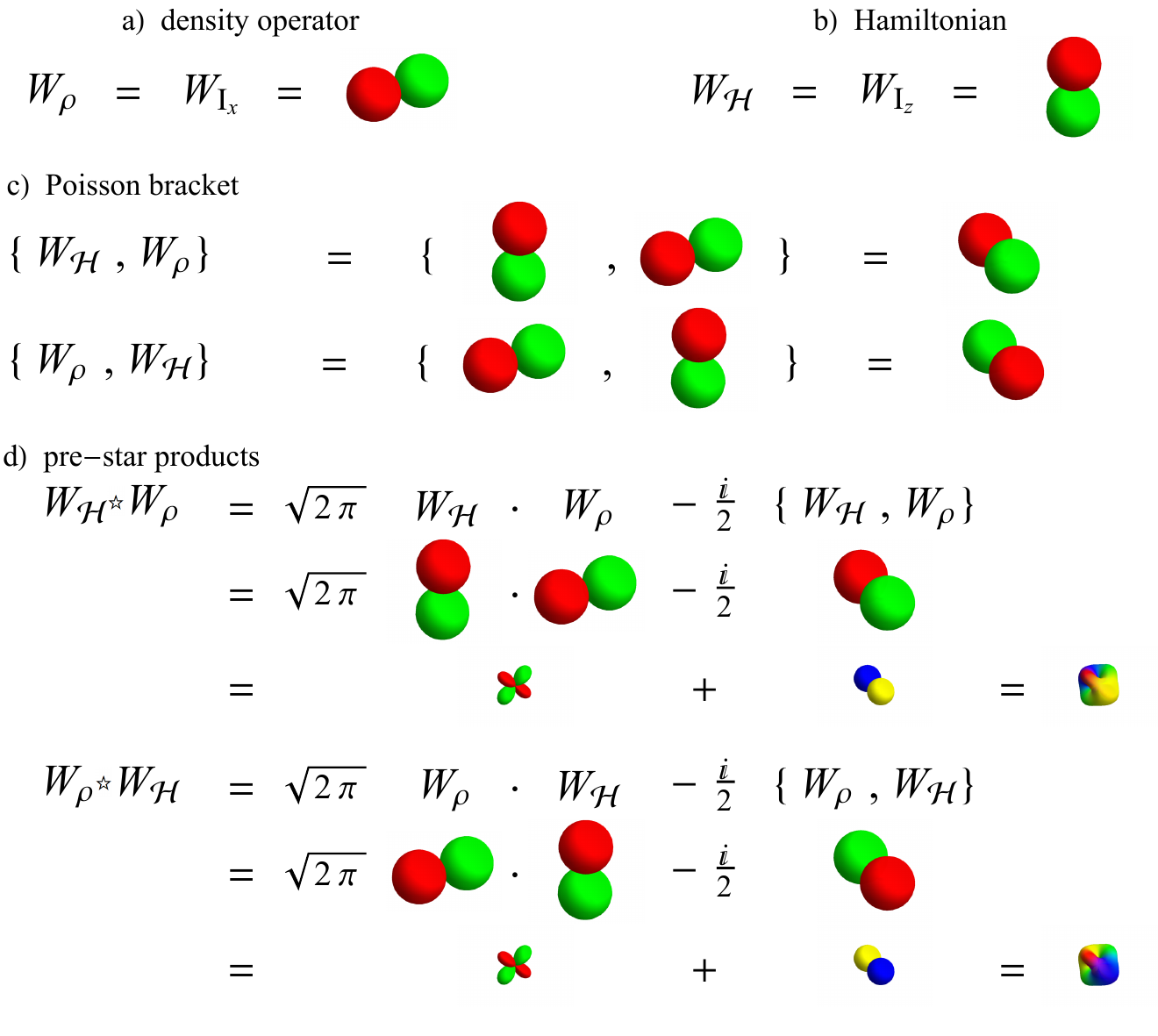}
	\caption{(Color online)
		Graphical visualization of the equation of motion for a 
		single spin $1/2$.
		a) and b) depict 
		the Hamiltonian and deviation density matrix from 
		Eq.~\eqref{InitSingleSpin}.
		Individual steps illustrate how the equation of motion for a single spin $1/2$ in
		Figure~\ref{GraphicalStarProduct2} g)
		[refer to Eq.~\eqref{Equation_Motion}] is derived by specifying c) the Poisson brackets
		(c.f. Fig~\ref{graphicalPB}), 
		d) and the pre-star products of the Wigner
		functions $W_\mathcal{H}$ and $W_\rho$.
		Refer also to Figure~\ref{GraphicalStarProduct2}.\footnoteref{\thecolors}
		\label{GraphicalStarProduct}}
\end{figure}

\section{Time evolution of a single spin~\texorpdfstring{$1/2$}{1/2} \label{AppendixStarProd}}

\begin{figure}[t]
	\centering
	\includegraphics{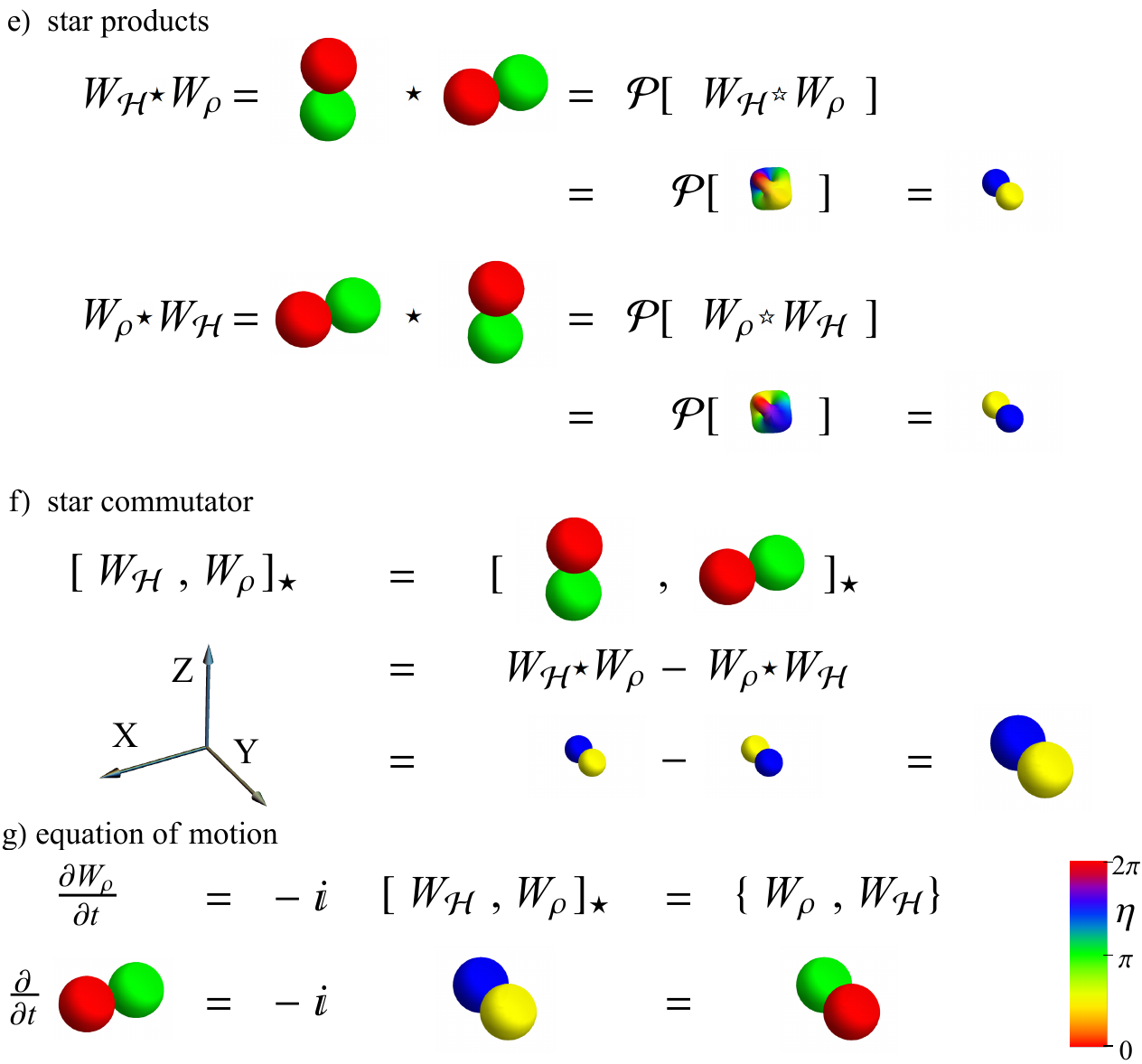}
	\caption{(Color online) Continuation of Figure~\ref{GraphicalStarProduct}.
		e) star products, f) the star commutator, and g) the equation of motions for the Wigner
		functions $W_\mathcal{H}$ and $W_\rho$.
		Refer to the main text for further details.\footnoteref{\thecolors}
		\label{GraphicalStarProduct2}
	}
\end{figure}

We further detail the derivation of the equation of motion for a single
spin $1/2$ which had been deferred from Sec.~\ref{Ex_one_Wigner} to this appendix.
The Figures~\ref{GraphicalStarProduct}-\ref{GraphicalStarProduct2} illustrate the
individual steps required to derive the time evolution of the Wigner function for
a single spin $1/2$: Subfigures a) and b) contain graphical representation 
of the Hamiltonian and density operator from
Eq.~\eqref{InitSingleSpin}. In Subfigure c), the Poisson bracket
of $W_\mathcal{H}$ and $W_\rho$ is computed following Fig.~\ref{graphicalPB};
the Poisson bracket is antisymmetric with respect to the order of its two arguments.
The pre-star product is decomposed in Subfigure d)
into a sum of the pointwise product and the
Poisson bracket of $W_\mathcal{H}$ and $W_\rho$, and the summands are weighted by the prefactors 
$\sqrt{2\pi}$ and $-i/2$, respectively [see Eq.~\eqref{prestar_single}].
In Subfigure e), the star product, which is the analog of the matrix product
of the two Wigner functions $W_\mathcal{H}$ and $W_\rho$
(see Sec.~\ref{subsecStarProdDef}), is obtained by projecting the
pre-star product onto spherical harmonics of rank-zero and one [refer to Eq.~\eqref{StarStatement}].
Subfigure 
f) shows the star commutator as the star analog of the matrix commutator
(see Sec.~\ref{subsecStarProdDef}).
Finally, 
Subfigure g)
determines 
the time evolution of the
density operator by specifying its time derivative.
The result in this particular case conforms with the general argument 
that the time evolution is given by the Poisson bracket of $W_\mathcal{H}$ and $W_\rho$.

\section{Integral form of the star product \label{IntegralStarProd}}

Similarly as in \cite{VGB89}, we evaluate the integral form of the star product of Wigner functions
corresponding to arbitrary operators $A$ and $B$ acting on  a single spin $J$. Based on the explicit form of the
kernel for the Wigner transformation in Eq.~\eqref{WignerTransform} and the expansion formula of tensor-operator
matrix products in Eq.~\eqref{ProdTop}, we evaluate the explicit form of the trikernel.
Let us first consider the
matrix product $AB$  of two arbitrary operators $A$ and $B$ of a spin $J$ and the corresponding Wigner function
$W_{AB}$, which is obtained according to Eq.~\eqref{WignerTransform} as
\begin{equation}\label{app_subst}
\wigner ( AB ) := W_{AB} (\theta, \phi)= \tr [  \Delta_{J}(\theta, \phi) AB ].
\end{equation}
The definition of the integral star product is derived by
substituting $A$ and $B$ in Eq.~\eqref{app_subst}
with 
the formula 
for the inverse Wigner transform of  $W_A$ and $W_B$
from Eq.~\eqref{inverseWigner}.
This leads to
\begin{align}	 
&W_{AB} (\theta, \phi)= W_A(\theta_1, \phi_1) \star  W_B(\theta_2, \phi_2) \label{integralstarprod} =    \\
& \idotsint\displaylimits_{\theta_1,\theta_2,\phi_1,\phi_2\geq 0}^{\theta_1,\theta_2\leq \pi, \phi_1,\phi_2\leq 2\pi}  
\Delta^{(T)}_{J} W_A W_B
\: \sin{\theta_1} \sin{\theta_2} \,\mathrm{d}\theta_1 
\,\mathrm{d}\theta_2  \, \mathrm{d}\phi_1 \, \mathrm{d}\phi_2, \nonumber
\end{align}
where $W_A:=W_A(\theta_1, \phi_1)$ and $W_B:=W_B(\theta_2, \phi_2)$ denote Wigner functions of
$A$ and $B$ depending on different variables, and the trikernel is given by
\begin{equation}
\Delta^{(T)}_{J} = \tr [  \Delta_{J}(\theta, \phi) \Delta_{J}(\theta_1, \phi_1) \Delta_{J}(\theta_2, \phi_2)  ].
\end{equation}
Using the expansion formula in Eq.~\eqref{ProdTop} to evaluate the explicit form of the trikernel, one obtains
%-------------------------
\begin{equation} 
\Delta^{(T)}_{J} 
= 
\sum_{\vec{j},\vec{m}}
\Y^*_{j_1 m_1}(\theta_1, \phi_1)
\Y^*_{j_2 m_2}(\theta_2, \phi_2)  \sum_{L=|j_1-j_2|}^{n} 
Q^{(J)}_{j_1 j_2 L} \,
C^{LM}_{j_1 m_1 j_2 m_2}
\Y_{L M}(\theta, \phi),
\end{equation}
%-------------------------
where $\vec{j}=(j_1,j_2)$, $\vec{m}=(m_1,m_2)$ and $0 \leq j_1,j_2 \leq 2J$, and
the upper limit $n$ given as $n:=\min( j_1{+}j_2, 2J )$; note $M=m_1{+}m_2$.
With this form of the trikernel, Eq.~\eqref{integralstarprod} can be interpreted in the following way:
Integrating over the product $W_A W_B$ multiplied by the trikernel, one obtains the 
decomposition of both $W_A$ and $W_B$ 
into spherical harmonics,
and the Wigner transformation of
the product $\Tj_{j_1,m_1}\Tj_{j_2,m_2}$ is obtained
for each pair $\Y_{j_1m_1}$ and $\Y_{j_2m_2}$ of spherical harmonics.

\allowdisplaybreaks

\section*{References}

%\bibliography{references}

\end{document}